\documentclass[final]{article}

\usepackage{textcomp}
\usepackage{graphicx}
\usepackage{soul}
\usepackage{a4wide}
\usepackage{graphicx}
\usepackage{authblk}
\usepackage{amssymb}
\usepackage{epstopdf}
\usepackage[sans]{dsfont}
\usepackage[applemac]{inputenc}
\usepackage[english]{babel}
\usepackage{latexsym}
\usepackage{subfigure}
\usepackage{color}
\usepackage{float}
\usepackage{amsmath}
\usepackage{amsfonts}
\numberwithin{equation}{section}
\usepackage{amsthm}
\usepackage[numbers]{natbib}
\usepackage[bookmarks=true,colorlinks=true,linkcolor={blue},urlcolor={blue}, citecolor={blue},pdfstartview={XYZ null null 1.22}]{hyperref}%

\newtheorem{proposition}{Proposition}

\def\RR{\mathbb R}

\def\be{\begin{equation}}
\def\ee{\end{equation}}
\def\bea{\begin{eqnarray}}
\def\eea{\end{eqnarray}}

\begin{document}
\title{Kinetic models of collective decision-making \\ in the presence of equality bias}

\author{Lorenzo Pareschi\thanks{University of
Ferrara, Department of Mathematics and Computer Science, Via N. Machiavelli 35 44121 Ferrara, Italy ({\tt lorenzo.pareschi@unife.it}).},\qquad Pierluigi Vellucci\thanks{University of Rome 1, Department of Basic and Applied Sciences for Engineering, Via A. Scarpa 16 00161, Roma, Italy ({\tt pierluigi.vellucci@sbai.uniroma1.it}).},\qquad Mattia Zanella\thanks{University of
Ferrara, Department of Mathematics and Computer Science, Via N. Machiavelli 35 44121 Ferrara, Italy ({\tt mattia.zanella@unife.it}).}}

\maketitle

\begin{abstract}
We introduce and discuss kinetic models describing the influence of the competence in the evolution of decisions in a multi-agent system. The original exchange mechanism, which is based on the human tendency to compromise and change opinion through self-thinking, is here modified to include the role of the agents' competence. In particular, we take into account the agents' tendency to behave in the same way as if they were as good, or as bad, as their partner: the so-called \emph{equality bias}. This occurred in a situation where a wide gap separated the competence of group members. We discuss the main properties of the kinetic models and numerically investigate some examples of collective decision under the influence of the equality bias. The results confirm that the equality bias leads the group to suboptimal decisions.  
\end{abstract}

\tableofcontents

\section{Introduction}

%``\emph{The first condition so that dialogue is possible, is the mutual respect, which implies the duty to loyally understand what the other says}'' said the Italian philosopher Norberto Bobbio in \cite{BP}. How many times, while discussing with a friend, or with a person we just met on the train, we apply this good rule? How many times, in front of the insistence of the other, we prefer to let it go, and how many times, instead, we find ourselves to press on a subject that, perhaps, we do not sufficiently know? We tend to think that everyone deserves an equal say in a debate, after all many democracies - which allow to freely express, for example, thoughts like those entrusted to these words - are based since decades on this assumption.
%
%Assumption, seemingly harmless, which can be yet harmful when we take decisions together as part of a group. A suggestion for make optimal decisions is to weight the different opinions according to how competent each agent of the group is \cite{G}; whenever they differ in competence it has been shown that an equal weighting is suboptimal \cite{BOLRRF,K}.

The fallibility of human judgement is evident in our everyday life, especially regarding our self-evaluation ability. Several tests have been designed in cognitive psychology and clinical research in order to find an experimental evidence for this phenomenon, see \cite{EH,FL,H,HF} and the references therein, showing that subjects are in general overconfident about the correctness of their belief. This lack in \emph{metacognition}, i.e. the self-assessment of our own knowledge skills, goes hand in hand with the grade of competence of each subject.

The correlation between \emph{competence} and \emph{metacognitive skills} is somehow double and might be summarized in the following sentence: “the same knowledge that underlies the ability to produce correct judgement is also the knowledge that underlies the ability to recognize correct judgement” \cite{KD}. Here the authors found a systematic bias of the most incompetent agents on their metacognition than the most experts; behavior which is usually known as \emph{Dunning-Kruger effect}. In other words, incompetence, is not only follows by poor choices but also disable to recognize that these are wrong or improvable. Furthermore the overconfidence of the novices emerges together with the underconfidence of highly competent individuals which tend to negatively estimate their skills.

This coupled deviation from an objective self-evaluation of personal abilities has been recently investigated in \cite{MPNAS}, where authors asked how people deal with individual differences in competence, in the context of a collective perceptual decision-making task, developing a metric for estimating how participants weight their partner's opinion relative to their own. Empirical experiments, replicated across three slightly different countries as Denmark, Iran, and China, show how participants assigned nearly equal weights to each other's opinions regardless of the real differences in their competence. This \emph{equality bias}, whereby people behave as if they are as good or as bad as their partner, is particularly costly for a group when a competence gap separates its members.

Drawing inspiration from these recent results, we propose a multi-agent model which takes into account the influence of competence during the formation of a collective decision\cite{BT,PT1}. After the seminal models for wealth/opinion exchange for a multi-agent system introduced in \cite{CC,T} some recent works considered additional parameters to quantify the personal knowledge or conviction \cite{BT,DJT,DMPW,PT1} or constrained versions of these models \cite{AHP,APZa,HZ}. For example, individuals with high conviction are resistant to change opinion, and can play the role of leaders \cite{APZa,DMPW}. In \cite{LCCC} there exists a threshold conviction beyond of which  one of the two choices provided to the individuals prevails, spontaneously breaking the existing symmetry of the initial set-up.

More precisely, we introduce a binary exchange mechanism for opinion and competence deriving a kinetic equation of  Boltzmann-type  \cite{AHP,APZa,BT,PT1,PT2,T}. The binary collision terms for competence and opinion describe different processes: 
\begin{itemize}
\item the competence evolution depends on a social background in which individuals grow and on the possibility for less competent agents to learn from the more competent ones; 
\item the opinion dynamics is based on a competence based compromise process including an equality bias effect and change of opinion through self-thinking;  
\item agents are driven toward an a-priori correct choice in dependence on their competence grade;
\end{itemize}
%The model is partially related by some recent works in which the  dynamics is influenced by external factor like knowledge or conviction \cite{BT,DJT,DMPW,PT1}.

%In \cite{BT} the authors introduced a kinetic model for conviction formation by assuming that the way in which conviction is formed is independent of the personal opinion. Then, the conviction parameter of each individuals entered into the microscopic binary interactions for opinion formation \cite{T}, to modify them in terms of compromise and self-thinking. We will consider in the present paper individuals characterized by two variables, representing competence and opinion respectively.
%
%\textcolor{red}{We assume that the individual competence is not dependent by the personal opinion, thus allowing formation of conviction without resorting to the distribution of opinions. First we introduce a microscopic model for opinion and competence.
% In this setting we take into account two relevant aspects of the equality bias mechanism: the \emph{responsibility for difficult group decisions} \cite{HF}, in which individuals tend to switch between their own opinions and the opinions of their partner, in order to sharing the responsibility for potential errors, and the \emph{social exclusion principle} \cite{Will}, where each individual tries not to be too capable or too unfit in order to stay relevant and socially included. Thanks to these we define then the concept of collective optimal decision, with which we will compare the biased decision of the model.}

In order to to derive a nonlinear equation of Boltzmann-type for the joint evolution of competence and opinion in the limit of a large number of interacting agents we resort to the principles of classical kinetic theory (we refer to the recent monograph \cite{PT2} for an introduction to the subject). Furthermore, to simplify the study of the asymptotic behavior of the model, we obtain a Fokker-Planck approximation of the dynamic in the so-called quasi-invariant scaling.

The rest of the paper is organized as follows, in Section \ref{sec:micro} we introduce the binary interaction model for competence and opinion. We discuss the competence-based interactions between agents formulating a definition of collective optimal decision which is coherent with the experimental setting of \cite{MPNAS}. Then the equality bias function is introduced acting as a modification of the effective competence into perceived competence. In Section \ref{sec:boltzmann} we derive the Boltzmann-type model and study the evolution of the moments under some specific assumptions. The Fokker-Planck approximation is then obtained in Section \ref{sec:FP}, and we derive the stationary marginal density of the competence variable. Finally in Section \ref{sec:numerics} we present several numerical experiments which show that the model is capable to describe correctly the decision making process based on agents competence and to include the equality bias effects. The latter, as expected, drive the system towards suboptimal decions.

\section{Modeling opinion and competence}\label{sec:micro}
In this section we discuss the modeling of opinion dynamics through binary exchanges, the analogous of dyadic interaction in the reference experimental literature \cite{BOLRRF,MPNAS}. The mathematical approach follows several recent works on alignment processes in socio-economical dynamics \cite{BT,DMPW,FHT,MT1,MT2,PT1,PT2,T}. 

\subsection{Evolution of competence}\label{sec:comp_evo}
%We start to agree on certain universal aspects about once we said \emph{competence}, without resorting its definition. Often we notice how more well-educated, more capable and more competent people are also those best disposed to dialogue. Then competence is generally associated to the predisposition to listen other people. The higher this quality, greater is the ability to value other opinions. Vice versa, a person unwilling to listen and dialogue is usually marked by a lower level of the described trait. Since we assume that this quality increases with competence, we will refer to it simply as competence.

It is evident that one of the main factors playing a role is the social background in which an individual grows and lives, then it is natural to assume that competence is, in part, inherited from the environment. Moreover, we clearly have the possibility to learn a specific competence during interactions with other, more competent, agents. 

%A consistent part of us is accustomed to rethink, and to have continuous afterthoughts on many aspects of our daily decisions, consciously or not obeying to the mechanisms of \emph{responsibility for difficult group decisions} and \emph{social exclusion}, mentioned above.

Similarly to the works \cite{BT,PT1,PT2} we describe the evolution of competence of an individual in terms of a scalar parameter $x\in X$ where $X\subset\RR^+$.
%is a closed interval. By convention we will identify $X$ with the unitary closed interval $[0,1]$.
 Let $z\in\RR^+$ be the degree of competence achieved from the background in each interaction; in what follows we will always suppose that $C(z)$ is a distribution with bounded mean
\be
\int_{\RR^+}C(z)dz=1,\qquad \int_{\RR^+}zC(z)dz=m_B.
\ee
We define the new amount of competence after a binary interaction of agents with competence $x$ and $x_*$ as follows
\begin{equation}\begin{cases}\label{eq:comp_bin}
x^{\prime}&=(1-\lambda(x))x+\lambda_C(x)x_*+\lambda_B(x) z+\kappa x\\
x_*^{\prime}&=(1-\lambda(x_*))x_*+\lambda_C(x_*)x+\lambda_B(x_*) z+\kappa x_*,
\end{cases}\end{equation}
where $\lambda(\cdot),\lambda_C(\cdot)$ and $\lambda_B(\cdot)$ quantify the amounts of competence lost by individual by the natural process of forgetfullness, the competence gained thanks to the interaction with other agents and the expertise gained from the background respectively, while $\kappa$ is a zero-mean random variable with finite second order moment $\sigma_{\kappa}^2$, taking into account the possible unpredictable changes of the competence process. A possible choice for $\lambda_C(x)$ is $\lambda_C(x)=\lambda_C \chi(x\ge \bar{x})$, where $\chi(\cdot)$ is the indicator function and $\bar{x}\in X$ a minimum level of competence required to the agents’ for increasing their own skills by interacting with the other agents.

With the dynamics \eqref{eq:comp_bin} we introduced a general process in which agents respectively loose and gain competence interacting with the other agents and with the background. It is reasonable then to assume that both the processes of acquisition and loss, which are weighted by the coefficients $\lambda,\lambda_C$ and $\lambda_B$, are bounded by zero. Thus if $\lambda\in[\lambda_{-},\lambda_{+}]$, with $\lambda_{-}> 0$ and $\lambda_{+}< 1$, and $\lambda_C(x),\lambda_B(x)\in[0,1]$ the random part may be chosen to satisfy $\kappa\ge -1+\lambda_+$. For example, $\kappa$ may be uniformly distributed in $[-1+\lambda_+,1-\lambda_+]$.

%Thanks to the introduced assumptions it is possible to show that the evolution of the competence of each agents remains nonnegative. \\

%Assuming $\lambda(x)=\lambda_B(x)+\lambda_C(x)$ and without diffusion we can ensure that the competence is positive being \eqref{eq:comp_bin} a convex combination of $x,x_*\in X$ and $z\in\RR^+$. Thus if $\lambda\in[\lambda_{-},\lambda_{+}]$, with $\lambda_{-}> 0$ and $\lambda_{+}<1$, and $\lambda_B(x)\in[0,\bar\lambda]$, with $\bar\lambda<1$, the random part may be chosen to satisfy the following bound: $\kappa\ge-(1-\lambda_+)$. Thanks to the introduced assumptions it is possible to show that the evolution of the competence of each agents remains nonnegative. \\

Let $g(x,t)$ be the density function of individuals with competence $x\in X\subset \RR^+$ at time $t\ge 0$. Resorting to the standard Boltzmann-type setting, we refer to \cite{PT2} for an extensive treatment, it is possible to describe in weak form the evolution of such density function as follows
\be
\dfrac{d}{dt}\int_{X} \psi(x)g(x,t)dx = \left< \int_{\RR^+\times X}(\psi(x’)-\psi(x))g(x,t)C(z)dxdz \right>,
\ee
where $x’$ is the post-interaction competence given in \eqref{eq:comp_bin}, the brackets $<\cdot>$ indicate the expectation with respect to the random variable $\kappa$ and $\psi(\cdot)$ is a test function. By imposing $\psi(x)=x$ we obtain an equation for the evolution of the evolution of the mean-competence $m_x(t)$ 
\be\begin{split}
\dfrac{d}{dt}m_x(t) = \dfrac{1}{2}\left(\int_{X^2}(\lambda_C(x)-\lambda(x))(x+x_*)g(x,t)g(x_*,t)dxdx_*+ \right.\\
\left. \int_{\RR^+\times X}\lambda_B(x) zg(x,t)C(z)dxdz\right),
\end{split}\ee
which, for $\lambda_C(x)=\lambda_C,\lambda_B(x)=\lambda_B,\lambda(x)=\lambda$, yields
\be
\dfrac{d}{dt}m_x(t) = -(\lambda-\lambda_C)m_x(t)+\lambda_Bm_B,
\ee
whose solution is given by
\be
m_x(t) = m_x(0)e^{-(\lambda-\lambda_C)t}+\dfrac{\lambda_B m_B}{\lambda-\lambda_C}(1-e^{-(\lambda-\lambda_C)t}).
\ee
Therefore if $\lambda> \lambda_C$ we obtain the asymptotic exponential convergence of the mean competence $m_x$ toward $\lambda_B m_B/(\lambda-\lambda_C)$. Note that, if we assume $\lambda=\lambda_B+\lambda_C$ we have that the average competence of the system tends to the mean competence induced by the variable $z\in\RR^+$. Finally we remark that, compared to previous models where the notion of knowledge/convinction were introduced \cite{BT,PT1}, here we have a fully binary dynamic which includes also the possibility to increase the agent's competence as a result of the interaction with the other agents.  

\subsection{The dynamics of competence based decisions}
\label{sec:prev}
%The description of the influence of competences on a opinion dynamics may be done through the setting of alignment processes \cite{AHP,APZa,PT2,T}. 
Let us now consider a system of binary interacting agents, each of them endowed with two quantities $(w,x)$ representing its opinion concerning a certain decision and competence respectively. Let $I=[-1,1]$ be the set of possible opinions for each interaction where the two extremal points $\pm 1$ represent the two alternative decisions. 
%We describe from a microscopic point of view the evolution of the variable opinion in terms of binary interactions. \\

Two agents identified by the couples $(x,w)$ and $(x_*,w_*)$ modify their opinion after interaction according to the following rules
\begin{equation}\begin{cases}\label{eq:binary}
w’ = w-\alpha_S S(x)(w-w_d)-\alpha_P P(x,x_*;w,w_*)(w-w_*)+\eta D(x,w) \\
w_*’ = w_*-\alpha_S S(x_*)(w_*-w_d)- \alpha_P P(x_*,x;w_*,w)(w_*-w)+\eta_* D(x_*,w_*),
\end{cases}\end{equation}
where $w,w_*\in I$ denote the pre-interaction opinions and $w’,w_*’$ the opinions after the exchange of information between the two agents. The positive function $S(\cdot)$ drives the agent toward the correct  choice $w_d\in \{-1,1\}$ with a rate dependent on its competence level, representing an individual decision making strength. For example, a possible choice for the function $S(\cdot)$ is $S(x)=\textrm{const.}>0$ if $x>\bar{x}$ and $S(x)=0$ elsewhere. In this case the agent needs to achieve a competence threshold $\bar{x}$ in order to perceive the correct choice. Observe that in \eqref{eq:binary} we introduced an interaction function $ P(\cdot,\cdot;\cdot,\cdot)$ depending both on the competence and on the opinion of the pair of agents. The nonnegative parameters $\alpha_S,\alpha_P$ characterize the drift toward the target opinion $w_d\in I$ and the interaction rate, respectively. The random variables $\eta,\eta_*$ are centered and with the same distribution $\Theta$ with finite variance $\sigma^2$.
% taking value in the Borel set $\mathcal{B}$.%, they are assumed to be independent from $\kappa$. 
The function $D(\cdot,\cdot)\ge 0$ represents the local relevance of the diffusion for a given opinion and competence, whereas the evolution of the competences $x,x_*$ are given by \eqref{eq:comp_bin}. 

In absence of diffusion $\eta,\eta_*=0$ and for a constant drift $S(\cdot)=S\le 1$ we have
\be\begin{split}
|w’-w_*’| & \le |1-\alpha_S S-\alpha_P (P(x,x_*;w,w_*)+P(x_*,x;w_*,w))| |w-w_*|.
\end{split}\ee
Then the post-exchange decisions are still in the reference interval $[-1,1]$ if we assume $0\le P(\cdot,\cdot;\cdot,\cdot)\le 1$ and $\alpha_S +2\alpha_P \le 2$. We can state the following result which identifies the condition on the noise term in order to ensure that the post-interaction opinions do not leave the reference interval. 
\begin{proposition}
We assume $0<P(x,x_*;w,w_*)\le 1,~ S(\cdot)=S\le 1$ and 
\[
|\eta|<(1-\alpha_P+\alpha_S)d, \qquad |\eta_*|<(1-\alpha_P+\alpha_S)d,\qquad \alpha_P \leq 1,
\]
where 
\[
d = \min_{(x,w)\in X\times I}\left\{ \dfrac{(1-w)}{D(x,w)}, D(x,w)\ne 0  \right\}.
\]
Then the binary interaction rule \eqref{eq:binary} preserves the bounds being $w’,w_*’\in [-1,1]$.
\end{proposition}
\begin{proof}
The proof is analogous to those in \cite{AHP,APZa,DMPW}, therefore we omit the details.
\end{proof}

The introduction of a drift term towards the correct decision in equation \eqref{eq:binary} follows from the set-up of the two-alternative forced choice tasks (TAFC) addressed in \cite{BBMHC,PSL}. The (TAFC) task is a classic behavioral experiment: at each trial agents are called to recognize a noisy visual stimulus choosing between two alternatives. The individual decision making involves in its mathematical description a drift term representing the fact that increasing experience/competence drives the agent towards the correct decision. 

The compromise function $0\le P(\cdot,\cdot;\cdot,\cdot)\le 1$ depends on both the particles’ opinion and competence. 
%This term, in the introduced binary model, indicates that each agent adjusts its opinion by adding a weighted average of the differences of its opinion with those of the other agents. 
As an example, a possible structure for the interaction function is given by the following
\begin{equation}\label{eq:P}
 P(w,w_*;x,x_*)=Q(w,w_*)R(x,x_*),
\end{equation}
where $0\le Q(\cdot,\cdot)\le 1$ represents the positive compromise propensity and $0\le R(\cdot,\cdot)\le 1$ is a function taking into account the competence of two interacting agents. 

Empirical experiments have been done in order to measure the impact of competence on collective decision-making tasks \cite{BOLRRF,KD,MPNAS}. Here the participants are organized in dyads which are called to make a decision about a visual stimulus. Among others two opposite models taking into account the influence of the competence, or \emph{sensitivity} in the psychology literature, might be considered. The first proposes that nothing except the opinions is communicated between individuals, and in case of disagreement we randomly select an opinion which coincides with the decision of the studied dyad. We will refer to this model as coin-flip model (CF). The second model takes strongly into account the competence of individuals: the decisions are communicated and in case of disagreement the opinion of the most competent prevails and then coincides with the decision of the dyad. This second model is called maximum competence model (MC). 

The aforementioned models can be described in our mathematical setting by considering a compromise propensity $Q(w,w_*)=1$. The CF model is a model in which at each interaction is associated a Bernoulli random variable with typical parameter $p=1/2$. In other words, in the CF case, the competence does not play any role in the decision process. We will consider an averaged version of this model (aCF), which corresponds to $R_{aCF}(\cdot,\cdot) = 1/2$, . For the MC model, in our setting corresponds to the Heaviside-type function
\begin{equation}\label{eq:R_Heaviside}
R_{MC}(x,x_*)=
\begin{cases}
1 & x<x_* \\
1/2 & x=x_* \\
0 & x>x_*.
\end{cases}
\end{equation}
In terms of the competence gap $x-x_*$ we can approximate the discontinuous function \eqref{eq:R_Heaviside} of the MC model with a smoothed continuous version (cMC) 
\begin{equation}\label{eq:R}
R_{cMC}(x,x_*)=\dfrac{1}{1+e^{c(x-x_*)}},
\end{equation}
with $c\gg 1$. We depict in Figure \ref{fig:R}, left plot, the behavior of the function $R(\cdot,\cdot)$ for different choices of the constant $c>0$. We can observe how in the half-plane $x>x_*$ the most competent agent is scarcely influenced by the less skilled one, while in the half-plane $x< x_*$ the situation is inverted and we see how an agent with less competence is influenced by the more competent one. 
\begin{figure}[tb]
\centering
\includegraphics[scale=0.37]{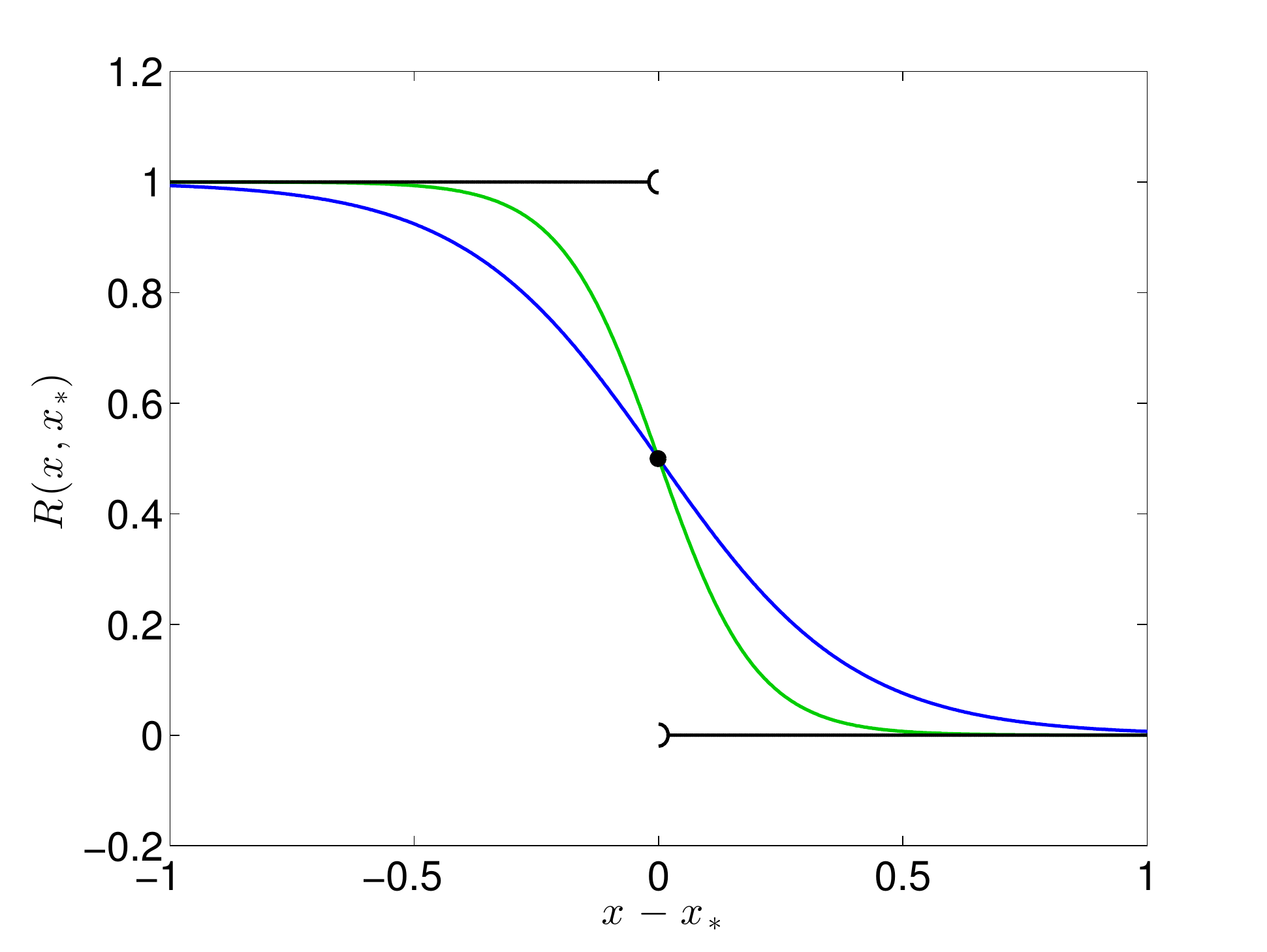}
\includegraphics[scale=0.37]{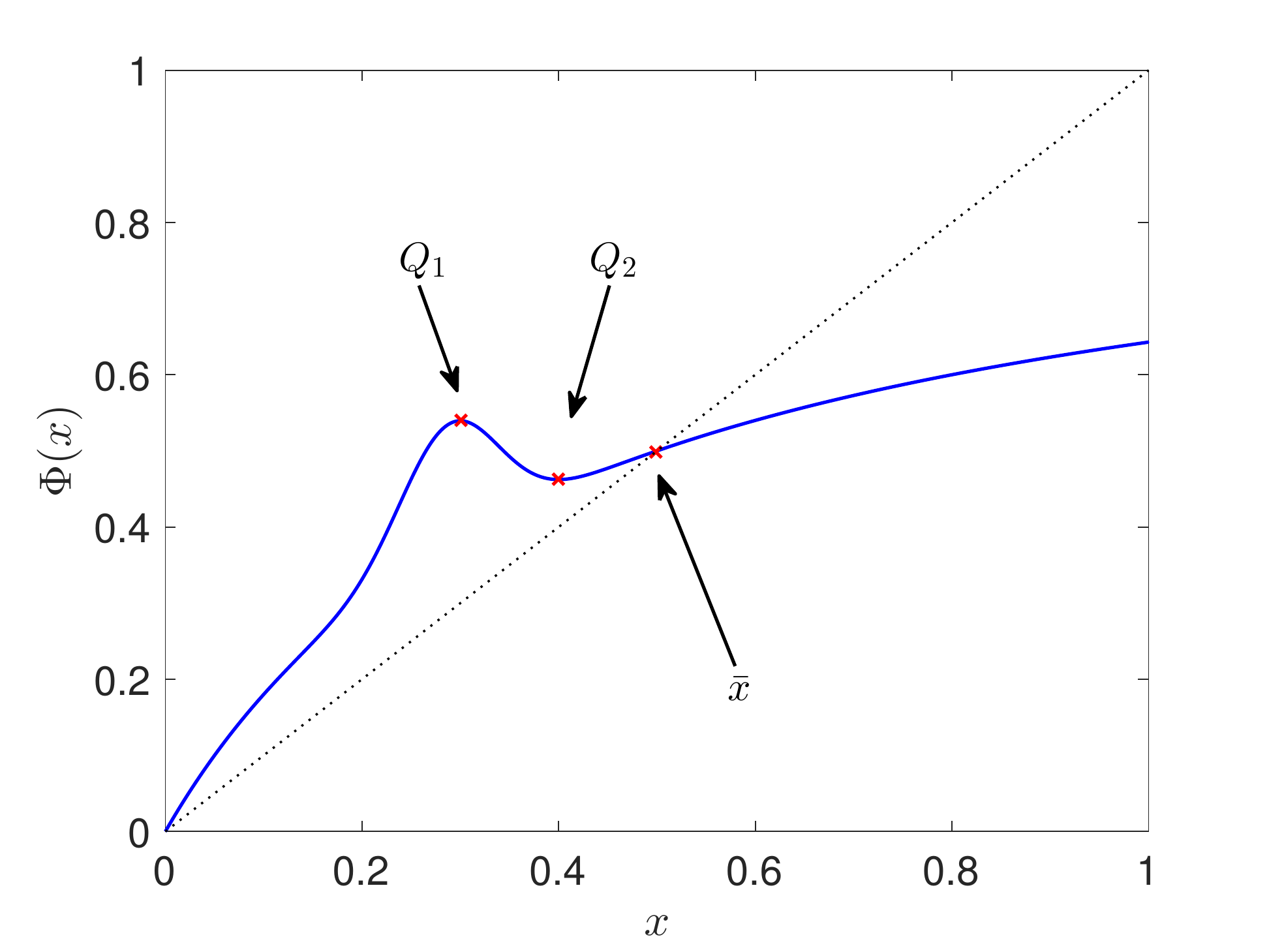}
\caption{Left: Heaviside-type competence-based interaction function $R_{MC}(x,x_*)$ and its continuous version defined in \eqref{eq:R} for $c=5$ and $c=10$.  Right: equality bias function $\Phi(x)$ of \eqref{eq:phi_psyco}, the choice of coefficients $a_1,a_2,a_3,a_4$ has been done such that $x_{Q_1}=0.3,x_{Q_2}=0.4$ and $\bar{x}=0.5$.}
\label{fig:R}
\end{figure}

\subsection{Dynamics of decisions under equality bias}\label{sec:decision_EB}
On the basis of the decision process built in section \ref{sec:prev} we modify the microscopic model \eqref{eq:binary}--\eqref{eq:P} in order to mathematically describe the phenomenon called \emph{equality bias} in collective decision--making communities. Our set-up is essentially inspired by some recent works of the experimental psychology literature \cite{BOLRRF,MPNAS} and the references therein. Their findings are consistent with the well-known cognitive bias called Dunning-Kruger effect regarding a misjudgment of personal competence of unskilled people, which overestimate their own ability. At the same time the most skilled individuals tends to underestimate their competence, implicitly believing that their knowledge is accessible to everyone \cite{KD}.
%Anyway the results of \cite{MPNAS} go far beyond the Dunning-Kruger effect, showing through real experiments on test subjects that people implicitly believe on relative competence.  \\

Let us consider an agents with competence $x\in X$. We introduce the equality bias function
\be\begin{split}
\Phi:  X&\mapsto X \\
  x& \rightarrow \Phi(x)
\end{split}\ee
 which measures the agent’s perceived level of competence. The choice of the equality bias function $\Phi(\cdot)$ is related to the subjective self-confidence in making decisions and experimental results show it is strictly dependent on the degree of competence of the interacting agents. 

%In the present setting a bias phenomenon emerges in collective decisions as a result of a miscalibration in competence-based interactions. 
Following the mathematical setting introduced in section \ref{sec:prev} we modify the function $P(\cdot,\cdot;\cdot,\cdot)$ as
\begin{equation}
P(\Phi(x),\Phi(x_*);w,w_*).
\end{equation}
In particular in the case \eqref{eq:P} the competence based interaction function modifies in $R(\Phi(x),\Phi(x_*))$ whose shape depends now on the perceived competence. Therefore the equality bias emerges as a result of the binary interactions of a multi-agent system with lack  of metacognition. In agreement with the experimental measurements of \cite{KD}, and in order to exemplify the action of the equality bias function, we depict in Figure \ref{fig:R} the function
\be\label{eq:phi_psyco}
\Phi(x) = x\left(\dfrac{a_1}{x+a_2}+a_3 \exp{\{-h(x-a_4)^2\}}\right),
\ee
with coefficients $a_1,a_2,a_3,a_4>0$. In the presented equality bias function we fixed the value ${x_d}\in X$ measuring the average competence of a population, the agents with competence $x>{x_d}$ may be identified with the most competent. Moreover, as experimentally suggested, it exists a local maximum $Q_1$ in $x=x_{Q_1}$ for the overestimation of the competence of unskilled agents, together with the local minimum $Q_2$ in $x=x_{Q_2}$. The asymptotic perceived competence depends on $a_1$, whose speed is weighted by $a_2$. Therefore we choose the coefficients $a_2,a_3,a_4$ in such a way that $\Phi( x_d)={x_d}$ and for $x=x_{Q_1},x=x_{Q_2}$ we find a local maximum and minimum respectively. The coefficient $h>0$ is a scaling parameter. In the rest of the paper we will refer to the dynamics of decisions under the action of the equality bias with (EB).

\section{A Boltzmann model for opinion and competence}\label{sec:boltzmann}
In this section we derive a kinetic model for opinion and competence reflecting the behavior introduced in the binary interaction model for opinion and competence in \eqref{eq:comp_bin}--\eqref{eq:binary}. 

Let $f=f(x,w,t)$ be the density of individuals with competence $x\in X\subset\RR^+$ and opinion $w\in I=[-1,1]$ at time $t\ge 0$. We derive the kinetic description for the evolution of the density function $f=f(x,w,t)$ through classic methods of kinetic theory \cite{PT2}. Let $g(x,t)$ be the marginal density of the competence competence variable $x\in X$
\be
g(x,t) = \int_I f(x,w,t)dw,
\ee
and $\int_{X} g(x,t) = 1$ for each $t\ge 0$.
% We explicit the evolution of the opinions by a binary interaction rule inspired by \eqref{eq:micro}. We derive from the differential system \eqref{eq:micro} an associated binary dynamics describing the evolution of the opinions depending on competence and opinions itself. \\

The evolution in time of the introduced density function is then given by the integro-differential equation of the Boltzmann-type
\be\label{eq:Bnonlin}
\dfrac{\partial }{\partial t}f(x,w,t)=Q(f,f)(x,w,t),
\ee
where $f(x,v,0)=f_0(x,w)$ and $Q(\cdot,\cdot)$ is defined as follows
\be\label{eq:Q}
Q(f,f)(x,w,t)= \int_{\mathcal{B}^2}\int_{X\times I}\left(’B \dfrac{1}{J}f(’x,’w,t)f(’x_*,’w_*,t) - Bf(x,w,t)f(x_*,w_*,t) \right)dw_*dx_*d\eta d\eta_*,
\ee
indicating with $(’w,’w_*)$ the pre-interaction opinions given by $(w,w_*)$ after the interaction and $(’x,’x_*)$ the pre-interaction competences. The term $J=J(x,w;x_*,w_*)$ denotes as usual the Jacobian of the transformations $(w,w_*)\rightarrow (w’,w_*’),(x,x_*)\rightarrow (x’,x_*’)$ via \eqref{eq:comp_bin} and \eqref{eq:binary}. The kernels $’B,B$ characterize the binary interaction and in the following will be considered of the form
\be
B_{(w,w_*)\rightarrow (w’,w_*’)}^{(x,x_*)\rightarrow (x’,x_*’)}=\beta \Theta(\eta)\Theta(\eta_*)C(z)\chi(|w’|\le 1)\chi(|w_*’|\le 1)\chi(x’\in X)\chi(x_*’\in X),
\ee
where $\beta>0$ is a scaling constant.

The presence of the Jacobian in the definition of the binary operator \eqref{eq:Q} may be avoided by considering its weak formulation. Let us consider a test function $\psi(x,w)$, we get
\be\begin{split}
&\int_{X\times I} Q(f,f)(x,w,t)\psi(x,w)dwdx \\
&\qquad= \beta \Bigg<\int_{\RR^+} \int_{X^2\times I^2} (\psi(x’,w’)-\psi(x,w))f(x,w,t)f(x_*,w_*,t)C(z)dwdw_*dxdx_*dz \Bigg> \\
&\qquad = \dfrac{\beta}{2}\Bigg<\int_{\RR^+} \int_{X^2\times I^2} (\psi(x’,w’)+\psi(x_*’,w_*’)-\psi(x,w)-\psi(x_*,w_*))\\
&\qquad\qquad f(x,w,t)f(x_*,w_*,t)C(z)dwdw_*dxdx_*dz \Bigg>
\end{split}\ee
where the brackets $<\cdot>$ denotes the expectation with respect to the random variables $\eta,\eta_*$. The weak formulation of the initial value problem \eqref{eq:Bnonlin} for the initial density $f_0(x,w)$ is given for each $t\ge 0$ by
\be\begin{split}\label{eq:weak}
&\dfrac{\partial}{\partial t}\int_{X\times I}\psi(x,w)f(x,w,t)dwdx \\
&\qquad =   \dfrac{\beta}{2}\Bigg<\int_{\RR^+} \int_{X^2\times I^2} (\psi(x’,w’)+\psi(x_*’,w_*’)-\psi(x,w)-\psi(x_*,w_*))	\\	
& \qquad \qquad f(x,w,t)f(x_*,w_*,t)C(z)dwdw_*dxdx_*dz \Bigg >.
\end{split}\ee
%for all test functions such that the following compatibility condition is satisfied
%\be
%\lim_{t\rightarrow 0^+}\int_{X\times I}\psi(x,w)f(x,w,t)dwdx=\int_{X\times I}\psi(x,w)f_0(x,w)dwdx.
%\ee
From the weak formulation in \eqref{eq:weak} we can derive the evolution of the macroscopic quantities like the moments for the opinion which may be obtained choosing as a test function $\psi(x,w)=\psi(w)=1,w,w^2$. 

\subsection{Collective decision and variance}
It is straightforward to observe that setting $\psi=1$ we obtain the conservation of the total number of agents. The mean opinion of the overall agents is defined as
\be\label{eq:u_def}
U(t) = \int_{X\times I} w f(x,w,t)dwdx,
\ee
which represents the \emph{collective decision} of the system at time $t$, see \cite{CG,G97,GZ}. 
%The decision of the agents with competence $x\in X$ is defined as 
%\be\label{eq:u_xt}
%u(x,t) =\int_{I}wf(x,w,t)dw.
%\ee
The evolution of the collective decision is derived as marginal quantity from equation \eqref{eq:weak} for $\psi(x,w)=w$
\be\begin{split}\label{eq:u_weak}
&\dfrac{\partial }{\partial t}\int_{X\times I} w f(x,w,t)dwdx= \alpha_S\beta \int_{X\times I}S(x)(w_d-w)f(x,w,t)dwdx \\
&\qquad\qquad+ \dfrac{\alpha_P\beta}{2} \int_{X^2\times I^2} \left(P(x,x_*;w,w_*)-P(x_*,x;w_*,w)\right) (w-w_*) \\
&\qquad\qquad f(x,w,t)f(x_*,w_*,t)dwdw_*dxdx_*.
\end{split}\ee
If the interaction function $P(\cdot,\cdot;\cdot,\cdot)$ is a symmetric function and $S(x)=s\in[0,1]$ 
equation \eqref{eq:u_weak} reduces to
\be
\dfrac{d}{dt} U(t) = \alpha_S \beta s\left(w_d-U(t)\right),
\ee
whose solution at each $t\ge 0$ is $U(t)= w_d+(U_0-w_d)\exp\{-\alpha_S \beta s t\}$, with $U_0=U(0)$ the initial collective decision. Therefore, in the limit $t\rightarrow +\infty$, the asymptotic collective decision converges exponentially toward $w_d\in\{-1,1\}$, i.e. $U_{\infty}=w_d$. 

In the case of the aCF model $P\equiv 1/2$ and $S(\cdot)\equiv 0$, the mean opinion of the overall system is then conserved, which implies $U(t)=U_0$ for all $t\ge 0$. Further if we consider a MC model, with interactions between agents based only on the competence variable and for $S\equiv 0$, from \eqref{eq:u_weak} we have
\be\begin{split}
\dfrac{d}{dt}\int_{X\times I} wf(x,w,t)dxdw =& {\alpha_P\beta} \int_{X^2\times I^2} R_{MC}(x,x_*)(w_*-w)f(x,w,t)f(x_*,w_*,t)dwdw_*dxdx_* \\
&- \dfrac{\alpha_P\beta}{2}\int_{X^2\times I^2}(w_*-w)f(x,w,t)f(x_*,w_*,t)dwdw_*dxdx_*,
\end{split}\ee
being $R_{MC}(x,x_*)-R_{MC}(x_*,x)=2R_{MC}(x,x_*)-1$ with $R_{MC}\equiv 0$ in the half space $x>x_*$, which leads to 
\be\begin{split}\label{eq:u_MC}
\dfrac{d}{dt}\int_{X} u(x,t)dx =& \alpha_P\beta \int_X  \left[ \dfrac{U^+(x,t)-U^-(x,t)}{2}\right]\rho(x,t) dx \\
&+ \alpha_P\beta \int_X  \left[\dfrac{\rho^-(x,t)-\rho^+(x,t)}{2}\right] u(x,t) dx\\
%& -\dfrac{\alpha_P\beta}{2} \left( \rho(x,t) U(t) - u(x,t)\right),
\end{split}\ee
where $u(x,t) = \int_I wf(x,w,t)dw$ is the mean opinion relative to the competence level $x\in X$ and
\be
U^+(x,t) =\int_{x< x_*}  \int_I w_*f(x_*,w_*,t)dw_*dx_*; \qquad \rho^+(x,t) = \int_{x<x_*}\int_{I}f(x_*,w_*,t)dw_*dx_*.
\ee
In particular $U^+(x,t)$ and $\rho^+(x,t)$ indicate the average opinion and the numerical density of agents with competence greater than $x\in X$. We define also $U^-(x,t) = U(t)-U^+(x,t)$ and $\rho^-(x,t) = 1-\rho^+(x,t)$ for each $x\in X$ and $t\ge 0$.
Therefore the variations of the mean opinion of agents with fixed competence $u(x,t)$ follow the choice of the most competent agents. 
%In order to exemplify, let us consider $U^+(x,t)>U^-(x,t)$, hence the variation of $U(t)$ is positive driving the collective decision toward the choice of the most competent agents. 
%Further we observe how from \eqref{eq:u_xt} follows $|u(x,t)|\le 1$ and 
%\be
%\dfrac{d}{dt}|u(x,t)|\le \dfrac{\alpha_P\beta}{2}(\rho(x,t)+1)\le 1\quad \textrm{if} \quad \alpha_P\beta\le 1
%\ee

%In a more general situation and under the assumption \eqref{eq:P} for a nonsymmetric function $R(x,x_*)$, if $Q(w,w_*)=Q(w_*,w)$ we get
%\begin{equation}\begin{split}\label{eq:mean_evo}
%&\dfrac{d}{dt}U(t) = \alpha_S\beta \int_{X\times I}S(x)(w_d-w)f(x,w,t)dwdx \\
%& \qquad\qquad\dfrac{\alpha_P\beta}{2}\int_{X^2\times I^2} Q(w,w_*)[R(x,x_*)-R(x_*,x)](w_*-w)\\
%&\qquad\qquad\qquad f(x,w,t)f(x_*,w_*,t)dwdw_*dxdx_*
%\end{split}\end{equation}
%and the evolution of the mean opinion is no more given in an explicit form due to the nonlinear influence induced by the competence. 

The second order moment for the opinion of the overall system is defined as
\be
E(t) = \int_{X\times I} w^2 f(x,w,t)dwdx
\ee
and its evolution may be obtained from  \eqref{eq:weak} with $\psi(x,w)=w^2$
\be\begin{split}\label{eq:evo_variance}
&\dfrac{d}{dt} E(t) = \alpha_S^2\beta \int_{X\times I} S^2(x)(w-w_d)^2f(x,w,t)dwdx \\
& + \dfrac{\alpha_P^2\beta}{2}\int_{X^2\times I^2}[P^2(x,x_*;w,w_*)+P^2(x_*,x;w_*,w)](w-w_*)^2f(x,w,t)f(x_*,w_*,t)dwdw_*dxdx_* \\
& -\alpha_P\beta \int_{X^2\times I^2}(w-w_*)[wP(x,x_*;w,w_*)-w_*P(x_*,x;w_*,w)]f(x,w,t)f(x_*,w_*,t)dwdw_*dxdx_* \\
& -2\alpha_S\beta \int_{X\times I}S(x)w(w-w_d)f(x,w,,t)dwdx \\
& + \alpha_S\alpha_P\beta\int_{X^2\times I^2}(w-w_*)[S(x)P(x,x_*;w,w_*)(w-w_d)-S(x_*)P(x_*,x;w_*,w)(w_*-w_d)]\\
& \qquad\qquad \qquad\qquad f(x,w,t)f(x_*,w_*,t)dwdw_*dxdx_*+\alpha_P\beta\sigma^2 \int_{X\times I} D^2(x,w)f(x,w,t)dwdx.
\end{split}\ee
In the simplified situation $P(x,x_*;w,w_*)=p\in[0,1]$, $S(x)=s\in [0,1]$, $\alpha_S=\alpha_P=\alpha$ and in absence of diffusion, the equation \eqref{eq:evo_variance} assumes the following form
\be\begin{split}
& \dfrac{d}{dt}E(t) = \alpha^2\beta s^2\int_{X\times I} (w-w_d)^2f(x,w,t)dwdx-2\alpha\beta s\int_{X\times I}(w^2-ww_d)f(x,w,t)dwdx \\
&\qquad +\alpha\beta p(\alpha p -1+\alpha s)\int_{X^2\times I^2}(w-w_*)^2f(x,w,t)f(x_*,w_*,t)dw_*dx_*dwdx
%
%+\dfrac{\alpha^2\beta}{4}\int_{X^2\times I^2}(w-w_*)^2f(x,w,t)f(x_*,w_*,t)dw_*dx_*dwdx \\
%&\qquad -2\alpha\beta s\int_{X\times I}(w^2-ww_d)f(x,w,t)dwdx-\dfrac{\alpha\beta}{2}\int_{X^2\times I^2}(w-w_*)^2f(x,w,t)f(x_*,w_*,t)dw_*dx_*dwdx\\
%&\qquad \dfrac{\alpha^2\beta s}{2}\int_{X^2\times I^2}(w-w_*)^2f(x,w,t)f(x_*,w_*,t)dw_*dx_*dwdx
\end{split}\ee
Since under the same conditions $U(t)$ converges for large time toward the correct choice $w_d\in\{-1,1\}$, we obtain that $E(t)$ converges exponentially to $w_d^2$ if
\be
\alpha \le\min\Big\{1,\dfrac{2(s+p)}{(s+p)^2+p^2}\Big\}.
\ee
Therefore the quantity
\be
\int_{X\times I}(w-w_d)^2f(x,w,t)dwdx,
\ee
under the above assumptions, converges toward zero for large times and the steady state solution is the Dirac delta $\delta(w-w_d)$ centered in the right decision.
%\begin{remark}
%Since under the same conditions $u(t)$ converges for large time toward the correct choice $w_d\in\{-1,1\}$, we obtain that $E(t)$ converges exponentially to $w_d^2$.
%\be
%\dfrac{d}{dt}E(t) = \left(\dfrac{\alpha^2\beta}{2}-\alpha\beta\right)E(t)-\left(\dfrac{\alpha^2\beta}{2}-\alpha\beta\right)\bar u^2,
%\ee
%and being the coefficient $\left(\dfrac{\alpha^2\beta}{2}-\alpha\beta\right)$ negative we have that $E(t)$ converges asymptotically to $\bar u^2$ and the quantity
%\be
%\int_{X\times I}(w-\bar u)^2f(x,w,t)dw=E(t)-\bar u^2
%\ee
%converges to zero for $t\rightarrow +\infty$. In other words, under the mentioned assumptions, the steady state solution in case of homogeneous competence is given by a Dirac delta $f_{\infty}=\delta(w-\bar u)$ centered in the optimal opinion $\bar u$.
%\end{remark}

%%%%%%%%%%%%%
\section{Fokker-Planck approximation}\label{sec:FP}
In order to obtain analytic results on the large-time behavior from Boltzmann-type models a classical mathematical tool is given by the derivation of approximated Fokker-Planck models through scaling techniques \cite{AHP,APZa,APZc,BT,FPTT,PT1,PT2}. In what follows we apply this approach, also known as quasi-invariant limit \cite{T, PT2}, to the model derived in the latter section.

We introduce a scaling parameter $\epsilon>0$ and the following scaled quantities
\be\begin{split}\label{eq:scaling_x}
\lambda = \epsilon \lambda,\qquad \lambda_B=\epsilon\lambda_B, \qquad \lambda_C=\epsilon\lambda_C,\qquad \sigma_{\kappa}=\sqrt{\epsilon}\sigma_{\kappa},
\end{split}\ee
where we assumed $\lambda_B(x)=\lambda_B>0,\lambda_C(x)=\lambda_C>0$ and $\lambda(x)=\lambda>0$. This corresponds to the situation where each interaction produces a small variation of the competence. The same strategy may be applied to the binary opinion model rescaling the interaction frequency $\beta$, the interaction propensity $\alpha$ and the diffusion variance $\sigma^2$ as follows
\be\label{eq:scaling_w}
\alpha_S=\epsilon\alpha_S,\qquad \alpha_P=\epsilon\alpha_P, \qquad \beta=\dfrac{1}{\epsilon},\qquad \sigma=\sqrt{\epsilon}\sigma.
\ee

The scaled equation \eqref{eq:weak} reads
\be\begin{split}\label{eq:scaled}
&\dfrac{\partial}{\partial t}\int_{X\times I}\psi(x,w)f(x,w,t)dwdx = \\
& \qquad \dfrac{1}{\epsilon}\Bigg< \int_{\RR^+}\int_{X^2\times I^2}(\psi(x’,w’)-\psi(x,w))f(x,w,t)f(x_*,w_*,t)C(z)dwdw_*dxdx_*dz\Bigg>.
\end{split}\ee
Under the assumptions on the random variables involved in the binary exchanges $\kappa,\eta,\eta_*$ we define the following mean quantities
\be \begin{split}
<x’-x> &= -\lambda x+\lambda_C x_*+\lambda_B z = G(x,x_*,z), \\
<w’-w> &=-\alpha_S S(x)(w-w_d) -\alpha_P P(x,x_*;w,w_*)(w-w_*) \\
 & =H_S(x,w)+H_P(x,x_*;w,w_*),\\
<(x’-x)^2> &= G^2(x,x_*,z)+\sigma^2_{\kappa}x^2, \\
<(w’-w)^2>& = H^2(x,x_*;w,w_*)+\sigma^2D^2(x,w), \\
<(x’-x)(w’-w)>&= G(x,x_*,z)H(x,x_*;w,w_*).
\end{split}\ee
Then, we have
\be\begin{split}
&<\psi(x’,w’)-\psi(x,w)> = G(x,x_*,z)\dfrac{\partial \psi}{\partial x}(x,w)+(H_S(x,w)+H_P(x,x_*;w,w_*))\dfrac{\partial \psi}{\partial w}(x,w) \\
&\qquad \dfrac{1}{2}\Bigg[(G^2(x,x_*,z)+\sigma^2_{\kappa}x^2)\dfrac{\partial^2\psi}{\partial x^2}(x,w)+((H_S(x,w)+H_P(x,x_*;w,w_*))^2+\sigma^2D^2(x,w)\dfrac{\partial^2\psi}{\partial w^2}(x,w)) \\
&\qquad G(x,x_*,z)(H_S(x,w)+H_P(x,x_*;w,w_*))\dfrac{\partial^2\psi}{\partial x\partial w}(x,w)\Bigg]+R(x,x_*;w,w_*),
\end{split}\ee
where $R(x,x_*;w,w_*)$ denotes the higher order terms of the Taylor expansion. From the quasi-invariant scalings introduced in \eqref{eq:scaling_x}-\eqref{eq:scaling_w} it follows that
\be
G(x,x_*,z) = \epsilon \tilde G(x,x_*,z), \quad H_S(x,w)=\epsilon \tilde H_S(x,w),\quad H_P(x,x_*;w,w_*) = \epsilon \tilde H_P(x,x_*;w,w_*).
\ee
Then equation \eqref{eq:scaled} takes the form
\be\begin{split}
&\dfrac{\partial}{\partial t}\int_{X\times I}\psi(x,w)f(x,w,t)dwdx = \\
&\qquad \int_{\RR^+}\int_{X^2\times I^2} \Bigg[\tilde G(x,x_*,z)\dfrac{\partial \psi}{\partial x}(x,w)+\tilde H_S(x,w)\dfrac{\partial \psi}{\partial w}(x,w)+\tilde H_P(x,x_*;w,w_*)\dfrac{\partial \psi}{\partial w}(x,w)\\
&\qquad+\dfrac{\sigma^2_{\kappa}}{2}x^2 \dfrac{\partial\psi}{\partial x^2}(x,w)+\dfrac{\sigma^2}{2}D^2(x,w)\dfrac{\partial ^2\psi}{\partial w^2}(x,w)\Bigg]f(x,w,t)f(x_*,w_*,t)C(z)dwdw_*dxdx_*dz \\
&\qquad + \tilde R(\epsilon)+O(\epsilon),
\end{split}\ee
with
\be\begin{split}\label{eq:reminder}
&\tilde R(\epsilon) = \dfrac{1}{2\epsilon}\int_{\RR^+}\int_{X^2\times I^2} \Bigg[\epsilon^2\Big( \tilde G^2(x,x_*,z)\dfrac{\partial^2\psi}{\partial x^2}(x,w)\\
&\qquad+(\tilde H_S(x,w) +\tilde H_P(x,x_*;w,w_*))^2\dfrac{\partial^2\psi}{\partial w^2}(x,w)\\
&\qquad+\tilde G(x,x_*,z)(\tilde H_S(x,w) +\tilde H_P(x,x_*;w,w_*))\dfrac{\partial^2\psi}{\partial x\partial w}(x,w)\Big)\\
&\qquad +R(x,x_*;w,w_*)\Bigg]f(x,w,t)f(x_*,w_*,t)C(z)dwdw_*dxdx_*dz .
\end{split}\ee
By similar arguments to \cite{PT1,T} it can be shown that the term $R(\epsilon)$ defined in \eqref{eq:reminder} decays to zero in the limit $\epsilon\rightarrow 0$. Finally for $\epsilon\rightarrow 0$ we obtain
\be\begin{split}\label{eq:FP_weak}
&\dfrac{\partial}{\partial t} \int_{X\times I}\psi(x,w)f(x,w,t)dwdx=  \int_{\RR^+}\int_{X^2\times I^2} \Bigg[\tilde G(x,x_*,z)\dfrac{\partial \psi}{\partial x}(x,w)\\
&\qquad +\tilde H_S(x,w)\dfrac{\partial \psi}{\partial w}(x,w)+\tilde H_P(x,x_*;w,w_*)\dfrac{\partial \psi}{\partial w}(x,w)\Bigg]f(x,w,t)f(x_*,w_*,t)C(z)dwdw_*dxdx_*dz\\
&\qquad+\dfrac{\sigma^2_{\kappa}}{2}\int_{X^2\times I^2} x^2 \dfrac{\partial\psi}{\partial x^2}(x,w)f(x,w,t)dwdx \\
&\qquad+\dfrac{\sigma^2}{2}\int_{X^2\times I^2}D^2(x,w)\dfrac{\partial ^2\psi}{\partial w^2}(x,w)f(x,w,t)dwdx.
\end{split}\ee
Integrating back by parts equation \eqref{eq:FP_weak} we have the following nonlinear Fokker-Planck equation
\be\begin{split}\label{eq:FP}
&\dfrac{\partial}{\partial t}f(x,w,t)=\dfrac{\partial}{\partial x}\mathcal{G}[f](x,t)f(x,w,t)+\dfrac{\partial}{\partial w}\mathcal{H}[f]f(x,w,t)+\dfrac{\partial}{\partial w}\mathcal{K}[f]f(x,w,t)\\
&\qquad\dfrac{\sigma_{\kappa}^2}{2}\dfrac{\partial^2}{\partial x^2}(x^2 f(x,w,t))+\dfrac{\sigma^2}{2}\dfrac{\partial^2}{\partial w^2}(D^2(x,w)f(x,w,t)),
\end{split}\ee
where
\be\begin{split}
\mathcal{G}[f](x,t) &=\int_{\RR^+} \int_{X\times I}(\lambda x-\lambda_Cx_*-\lambda_B z) f(x_*,w_*,t)C(z)dw_*dx_*dz \\
&=\lambda x-\lambda_C m_x-\lambda_B m_B,
\end{split}\ee
being $m_x(t)=\int_{X\times I}xf(x,w,t)dwdx$, and where the functionals $\mathcal{H}[f],\mathcal{K}[f]$ are defined as follows
\be\begin{split}\label{eq:GH_func}
\mathcal{H}[f](x,w,t) &= \int_{X\times I} \alpha_P P(x,x_*;w,w_*)(w-w_*)f(x_*,w_*,t)dw_*dx_* \\
\mathcal{K}[f](w_d) &=\int_{X\times I} \alpha_S S(x)(w-w_d)f(x,w,t)dwdx.
\end{split}\ee

\subsection{Stationary states for the marginal density}\label{sec:stationary}
We introduce the marginal density function of the competence
\be
g(x,t) = \int_{I}f(x,w,t)dw.
\ee
%Let $m_x(t)=\int_{I}xg(x,t)dx$ be its mean value, that is bounded under the assumptions described in Section \ref{sec:micro}.
%Supposing that $\lambda,\lambda_B$ are given constants we have
%\be
%\mathcal{G}[f](x,w,t)=-\lambda m_x+\lambda_B m_z.
%\ee
By direct integration of the Fokker-Planck equation \eqref{eq:FP} with respect to the variable $w$ and considering a local diffusion function $D(x,w)$ decaying at the boundaries of the reference interval for the opinions, and taking into account the convergence of the mean competence showed in Section \ref{sec:comp_evo}, we obtain
\begin{equation}\begin{split}\label{eq:FP_marginalx}
&\dfrac{\partial}{\partial t}g(x,t)=\dfrac{\partial}{\partial x}\left(\lambda x-\dfrac{\lambda\lambda_B}{\lambda-\lambda_C} m_B \right)g(x,t)+\int_{I}\mathcal{H}[f]f(x,w,t)dw + \dfrac{\sigma^2_{\kappa}}{2}\dfrac{\partial^2}{\partial x^2} (x^2 g(x,t)).
\end{split}\end{equation}
The above equation simplifies by imposing the Dirichelet boundary conditions $f(x,-1,t)=f(x,1,t) = 0$ for each $x\in X, t\ge 0$. Therefore it is possible to give the analytic formulation of the stationary solution of \eqref{eq:FP_marginalx}. The solution of
\be
\dfrac{\partial}{\partial x}\left(-\lambda x+\dfrac{\lambda\lambda_B}{\lambda-\lambda_C} m_B \right)g^{\infty}(x)=\dfrac{\sigma_{\kappa}^2}{2}\dfrac{\partial^2}{\partial x^2}(x^2 g^{\infty}(x))
\ee
is given by
\be
g^{\infty}(x) = \dfrac{c_{\lambda,\lambda_B,\lambda_C,\sigma_{\kappa}^2}}{x^{2+2\lambda/\sigma_{\kappa}^2}}\exp\Big\{ -\dfrac{2}{\sigma^2_{\kappa}x}\cdot\dfrac{\lambda\lambda_Bm_B}{\lambda-\lambda_C} ) \Big\},
\ee
where $c_{\lambda,\lambda_B,\lambda_C,\sigma_{\kappa}^2}$ is a constant chosen such that the total mass of $g^{\infty}$ is equal to one.

Unlike the usual method for determining the stationary density developed in  \cite{BT,PT1}, here the asymptotic competence has been derived from the complete Fokker-Planck equation \eqref{eq:FP} after the integration of the opinion variable, and under specific boundary conditions.

%%%%%%%%%%%%%%
\section{Numerics}\label{sec:numerics}
%We introduce the concept of \emph{collective decision} as follows
%\begin{definition}[Collective decision]\label{def:coll_opt}
%We define collective decision for a multi-agent system the average opinion after interaction
%\[
%\bar{w} = \int_{X\times I} wf^{\infty}(x,w)dwdx.
%\]
%\end{definition}
%An analogous definition has been considered in a finite-dimensional setting in \cite{CG,G97,GZ}  in the TAFC case, i.e. each decision is a dichotomous variable whose possible values are $\pm 1$. \\
%When group members are similar in terms of their competence, we expect that the equality bias simplifies the decision process. However, as observed in \cite{MPNAS}, ``when a wide competence gap separates group members, the normative strategy requires that each opinion is weighted by its reliability. In such situations, an equality bias can be damaging for the group. Indeed, previous research has shown that group performance in the task described here depends critically on how similar group members are in terms of their competence’’. The Definition that we introduced takes into account all these issues.\\

In this section we propose several numerical tests for the Boltzmann-type model introduced in the previous paragraphs which show the emergence of suboptimal collective decisions under the hypothesis of equality bias compared to aCF and cMC models.  All the results presented have been obtained through direct simulation Monte Carlo methods for the Boltzmann equation (see \cite{PR,PT2}). 

%Starting with the initial condition $f(x,w,0)=f_0(x,w)$, with $f_0(x,w)$ a probability function over the domain $X\times I$, we numerically integrate  the initial value problem
%\begin{equation}\begin{cases}
%\dfrac{d}{dt}f(x,w,t)=Q(f,f)(x,w,t) , \\
%f(x,w,0)=f_0(x,w)
%\end{cases}\end{equation}
%through the time-discrete scheme 
%\be\label{eq:scheme}
%f^{n+1}(x,w) = \left( 1-\dfrac{\Delta t}{\epsilon} \right) f^n(x,w)+\dfrac{\Delta t}{\epsilon}Q^+_{\epsilon}(f^n,f^n)(x,w),
%\ee
%where $f^{n}=f(x,w,t^n)$ being $t^n=n\cdot \Delta t,n\ge 0$ a time discretization. In \eqref{eq:scheme} we explicit the dependence on the interaction frequency $1/\epsilon$ and we denoted with $Q^+_{\epsilon}(\cdot,\cdot)$ the gain part accounting the density of opinions $w$ with competence $x$ after interaction. Observe that \eqref{eq:scheme} is a convex combination of probability densities under the restriction on the time step $\Delta t\le \epsilon$. 
 The Fokker-Planck regime is obtained via the quasi-invariant scaling \eqref{eq:scaling_x}--\eqref{eq:scaling_w} with $\epsilon=0.01$. The local diffusion function has been considered in the case
\be\label{eq:local_dif}
D(x,w) = 1-w^2
\ee
which multiplies in the binary collision model \eqref{eq:binary} for the opinion, a uniform random variable with scaled variance $\sigma^2$. The choice \eqref{eq:local_dif} implies that the diffusion does not act on the agents with more extremal opinions.
In the competence dynamics we considered $\lambda(x)=\lambda$, $\lambda_C(x)=\lambda_C$ and $\lambda_B(x)=\lambda_B$ and a uniform random variable $\kappa$ with finite scaled variance $\sigma_{\kappa}^2$. The random variable $z$ in here uniform with mean $1/2$.

\subsection{Test 1: collective decision under equality bias}\label{test1}

%When group members are similar in terms of their competence, we expect that the equality bias simplifies the decision process. However, as observed in \cite{MPNAS}, ``when a wide competence gap separates group members, the normative strategy requires that each opinion is weighted by its reliability. In such situations, an equality bias can be damaging for the group. Indeed, previous research has shown that group performance in the task described here depends critically on how similar group members are in terms of their competence’’. \\

In this test we compute the collective decisions emerging from the Boltzmann-type model \eqref{eq:Bnonlin} with interaction function \eqref{eq:P} and $Q(w,w_*)=1$ and $R(x,x_*)$ defined in the reference cases aCF-cMC and EB. The binary interaction terms defined in \eqref{eq:comp_bin}--\eqref{eq:binary} have been considered for a choice of constants coherent with the bounds described in Section \ref{sec:micro}. In this test we considered a vanishing drift term $S(x)=0$, for each $x\in X$.

We compare the emerging collective decisions in the aCF and cMC cases considering the action of the equality bias function described in equation \eqref{eq:phi_psyco}. In this test we take into account the compromise behavior of two strongly polarized populations with equal sizes (case A) and with different sizes (case B).  In both cases the group of competent agents is characterized at $t=0$ by opinions in the interval $w\in [-1,-0.75]$, whereas the second population, composed by the less competent agents, expresses opinions $w \in [0.75,1]$. In the case of populations with different size we considered an initial distribution such that the number of competent agents is five time larger then the number of the incompetent ones. In Figure \ref{fig:initial} we exemplified the initial configurations for the two tests: Test 1A and Test 1B.  

\begin{figure}
\centering
\includegraphics[scale=0.44]{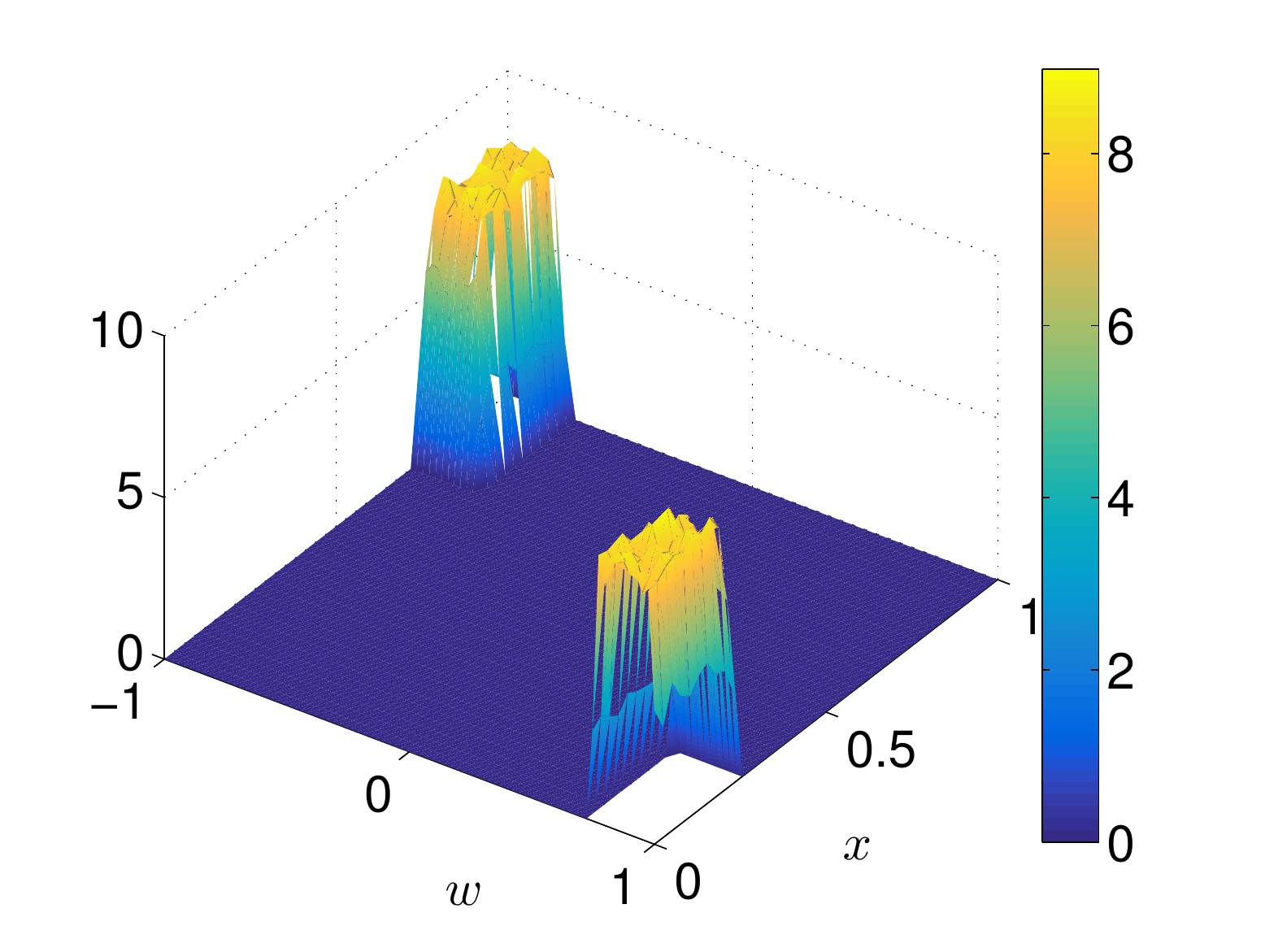}
\includegraphics[scale=0.44]{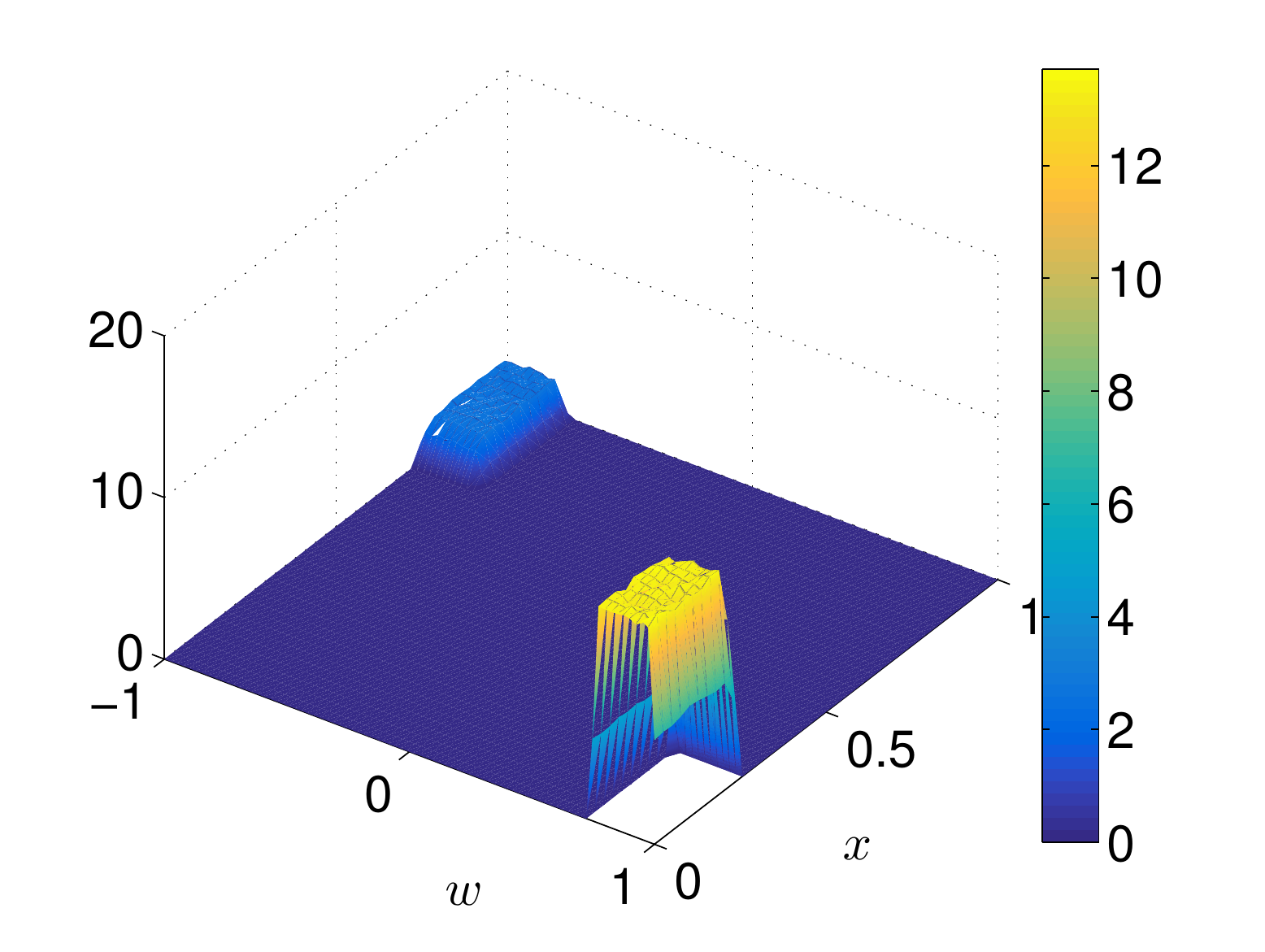}
\caption{Left: \textbf{Test 1A}, initial configuration of the multi-agent system with an equal size of competent and incompetent agents. Right: \textbf{Test 1B}, initial configuration of the multi-agent system with a different size of competent and incompetent agents, in particular the number of competent agents is five times smaller than the size of the incompetent agents. }
\label{fig:initial}
\end{figure}

\begin{figure}
\centering
\includegraphics[scale=0.35]{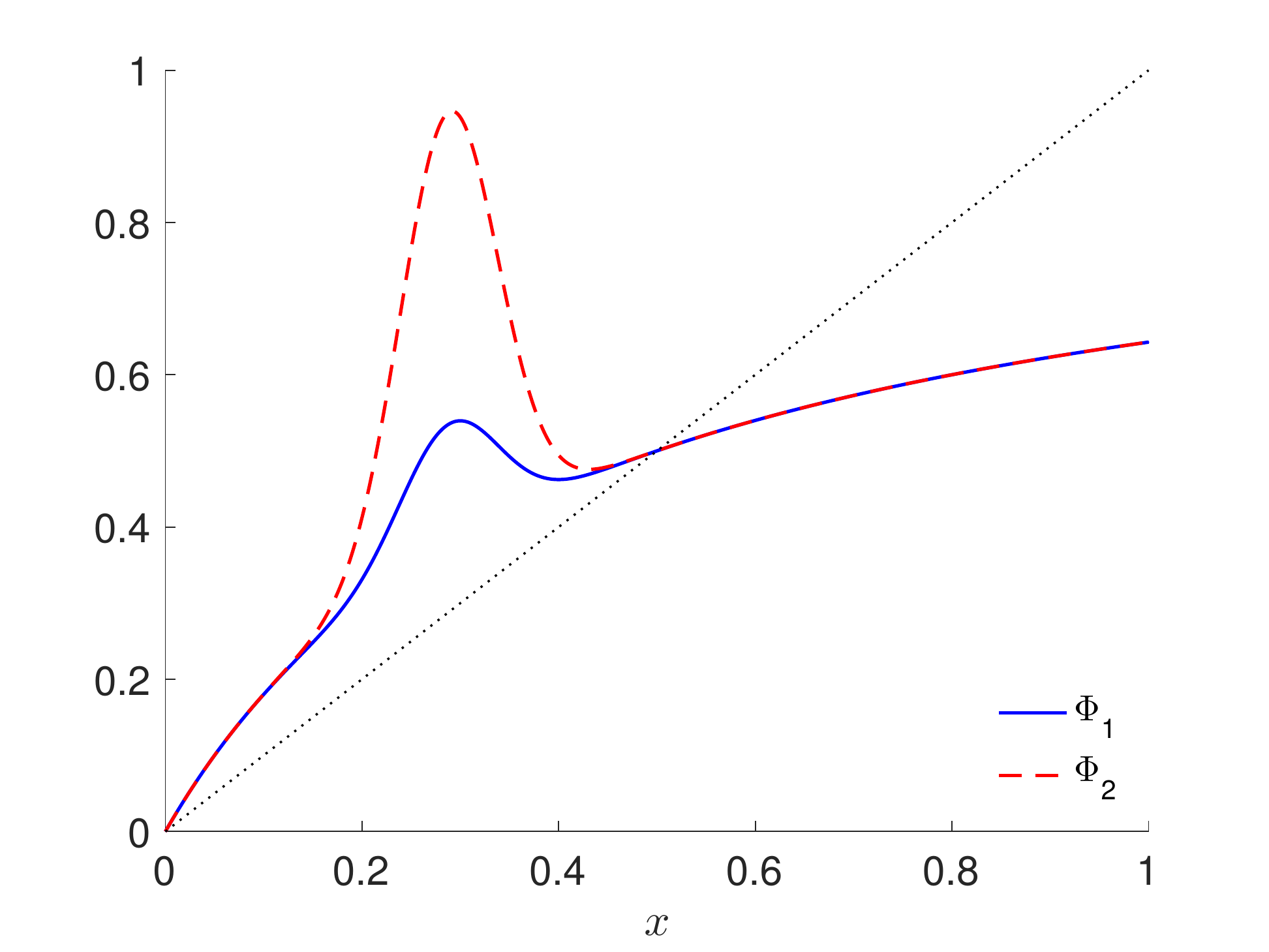}
\includegraphics[scale=0.35]{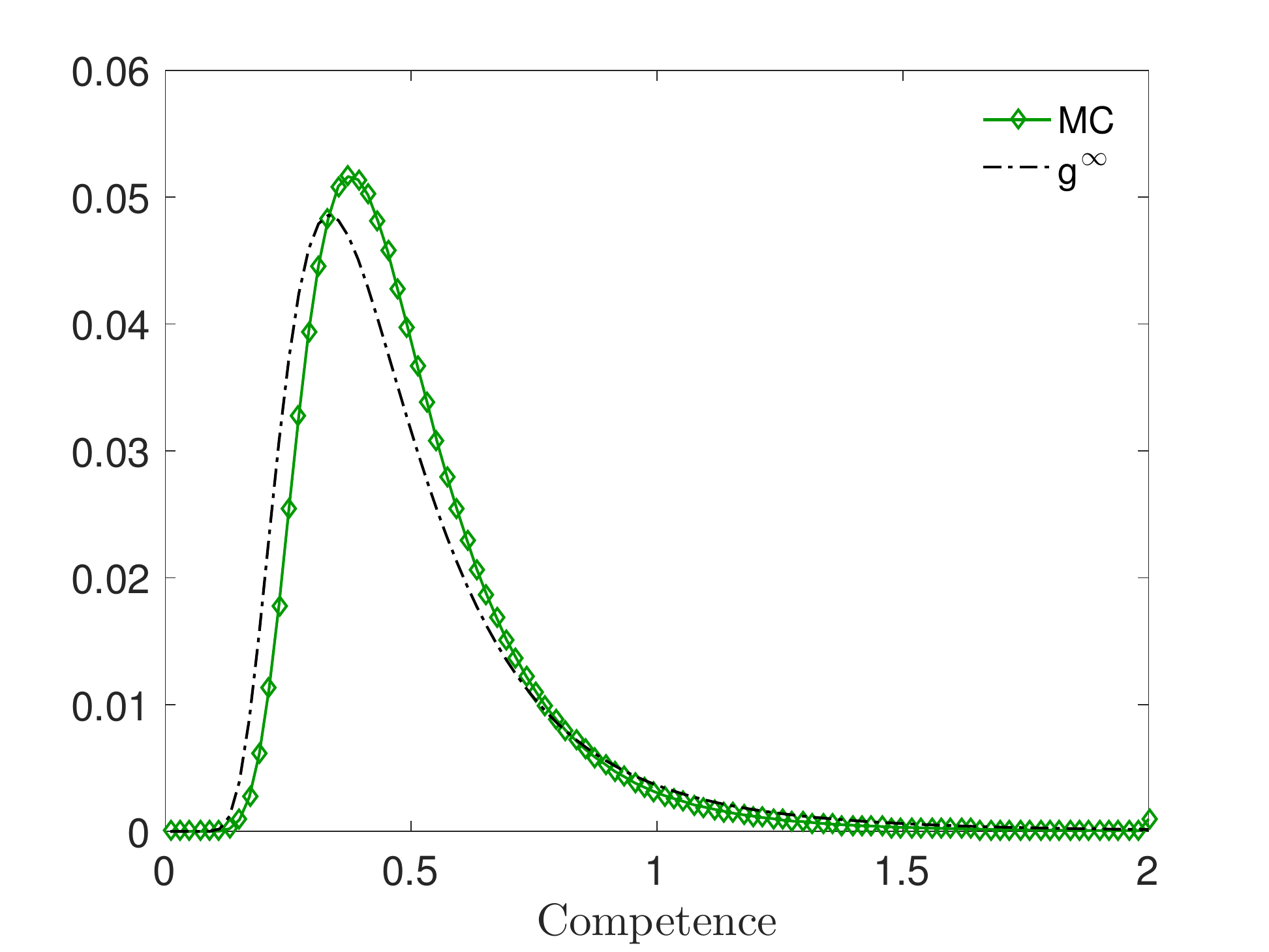}
\caption{\textbf{Test 1}. Left: two examples of equality bias functions $\Phi(x)$ if the form introduced in \eqref{eq:phi_psyco} and passing through $\bar{x}=0.5$ with $a_1=0.9$ and scaling parameter $h=10^2$. In the rest of the paper we will refer to them as $\Phi_1(x)$ and $\Phi_2(x)$. Right: convergence of the Monte Carlo method toward the analytic steady state for $\epsilon=0.01$.}
\label{fig:phi_competence}
\end{figure}

In Figure \ref{fig:test1} we show the kinetic and particle solutions of the Test 1A for a system of interacting agents in the aCF-cMC and EB1-EB2 models, where the case of equality bias has been considered for two examples of equality bias functions $\Phi_1(x)$ and $\Phi_2(x)$ of the type introduced in Section \ref{sec:decision_EB}, see equation \eqref{eq:phi_psyco}, and whose behaviors are reported in Figure \ref{fig:phi_competence}. Further in Figure \ref{fig:phi_competence} we compare the stationary distribution for the competence variable emerging form the Monte Carlo method and its analytic formulation obtained in Section \ref{sec:stationary}. In particular in Figure \ref{fig:test1} we can observe how the emerging collective decisions of the EB1 and EB2 cases are significantly shifted from the decision emerging in a model cMC.

\begin{figure}
\centering
\subfigure[aCF t=1]{
\includegraphics[scale=0.26]{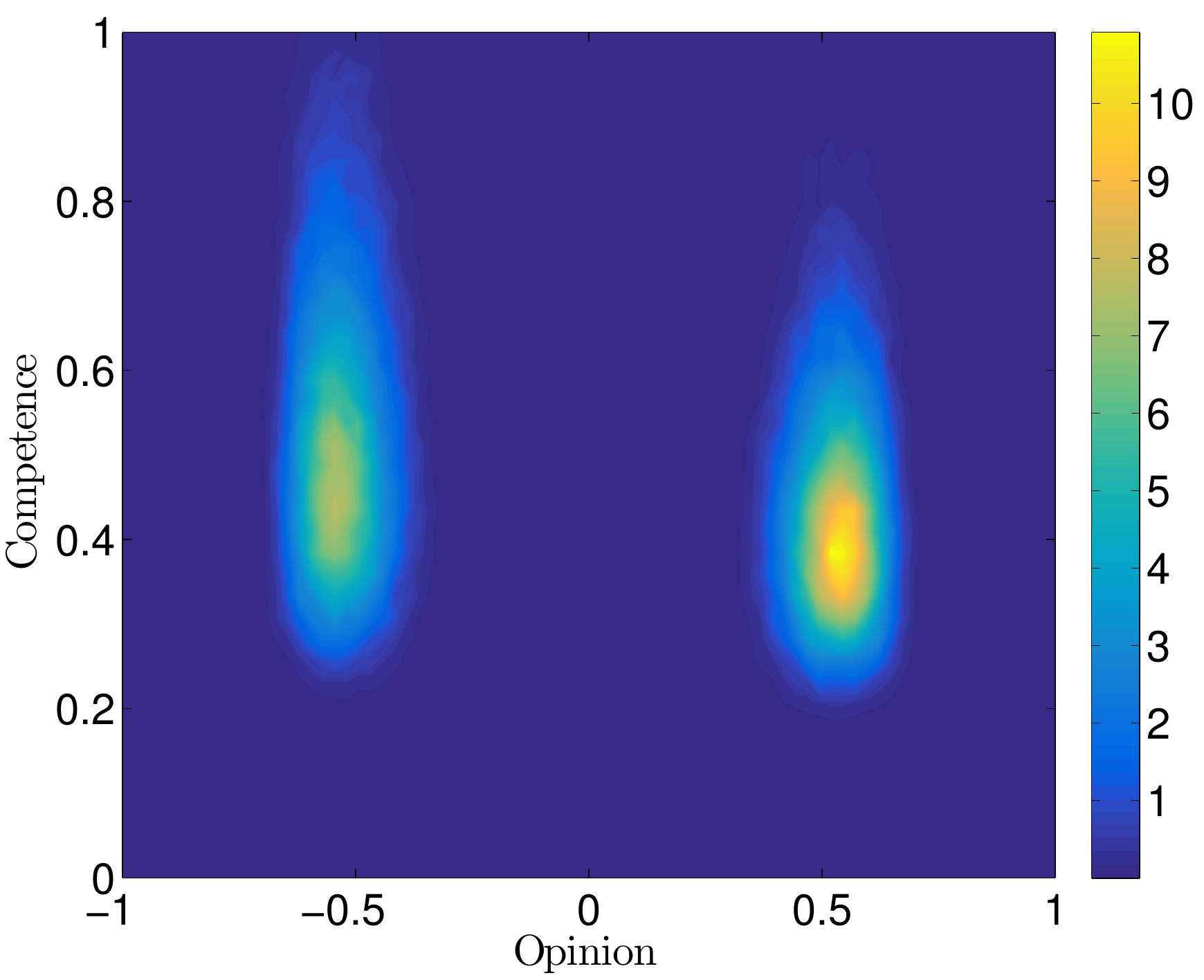}}
\subfigure[aCF t=3]{
\includegraphics[scale=0.26]{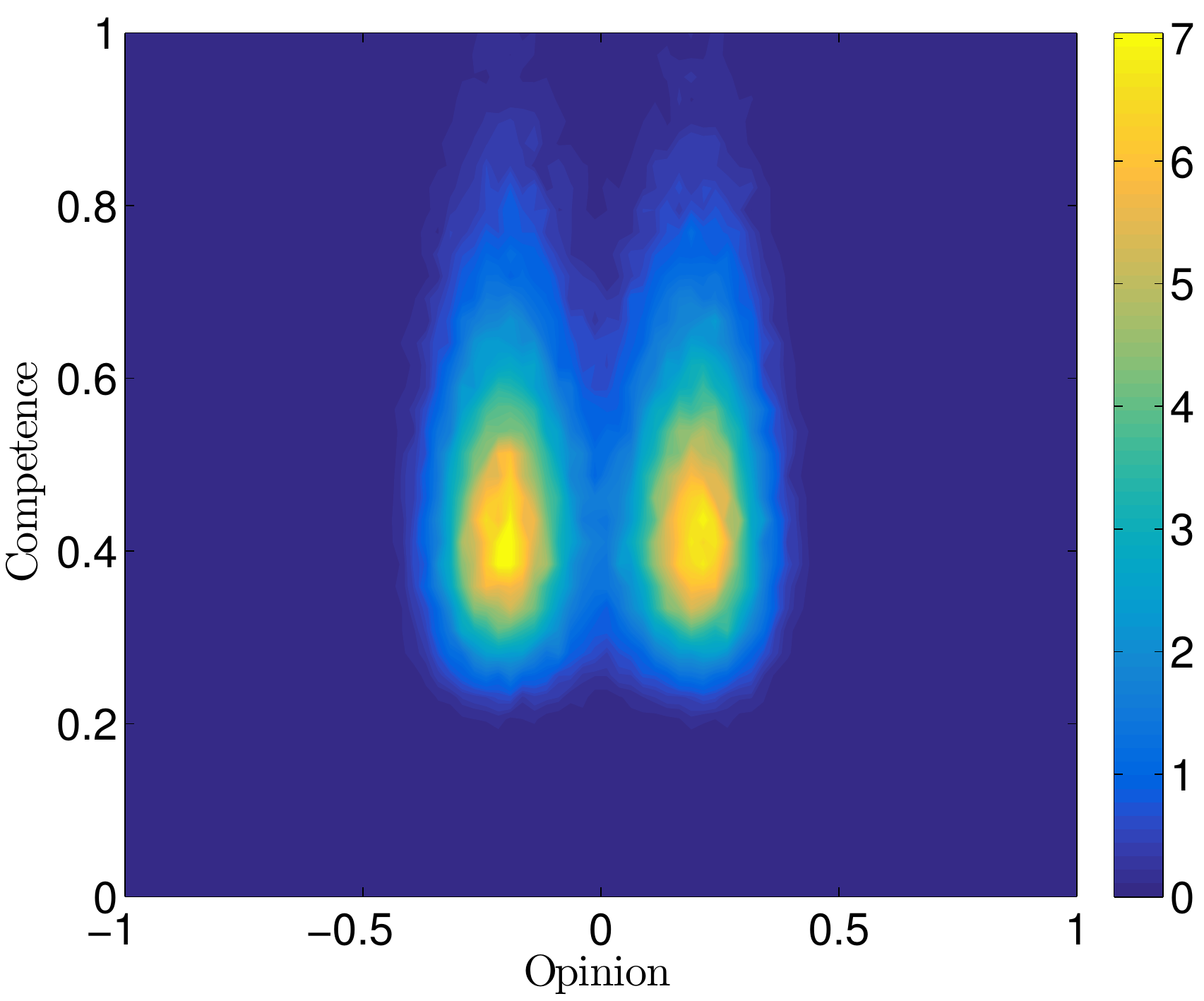}}
\subfigure[aCF t=10]{
\includegraphics[scale=0.26]{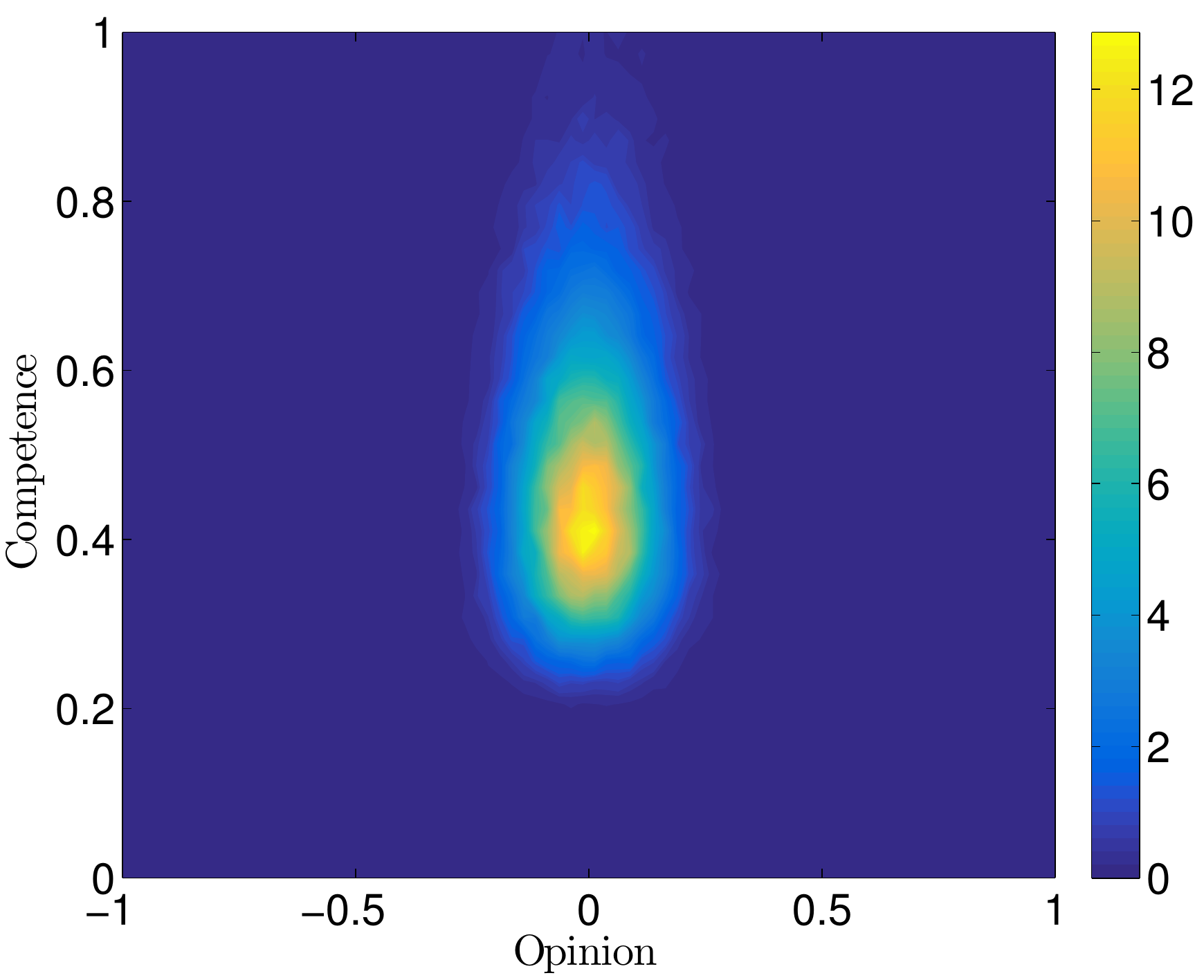}}\\
\subfigure[cMC t=1]{
\includegraphics[scale=0.26]{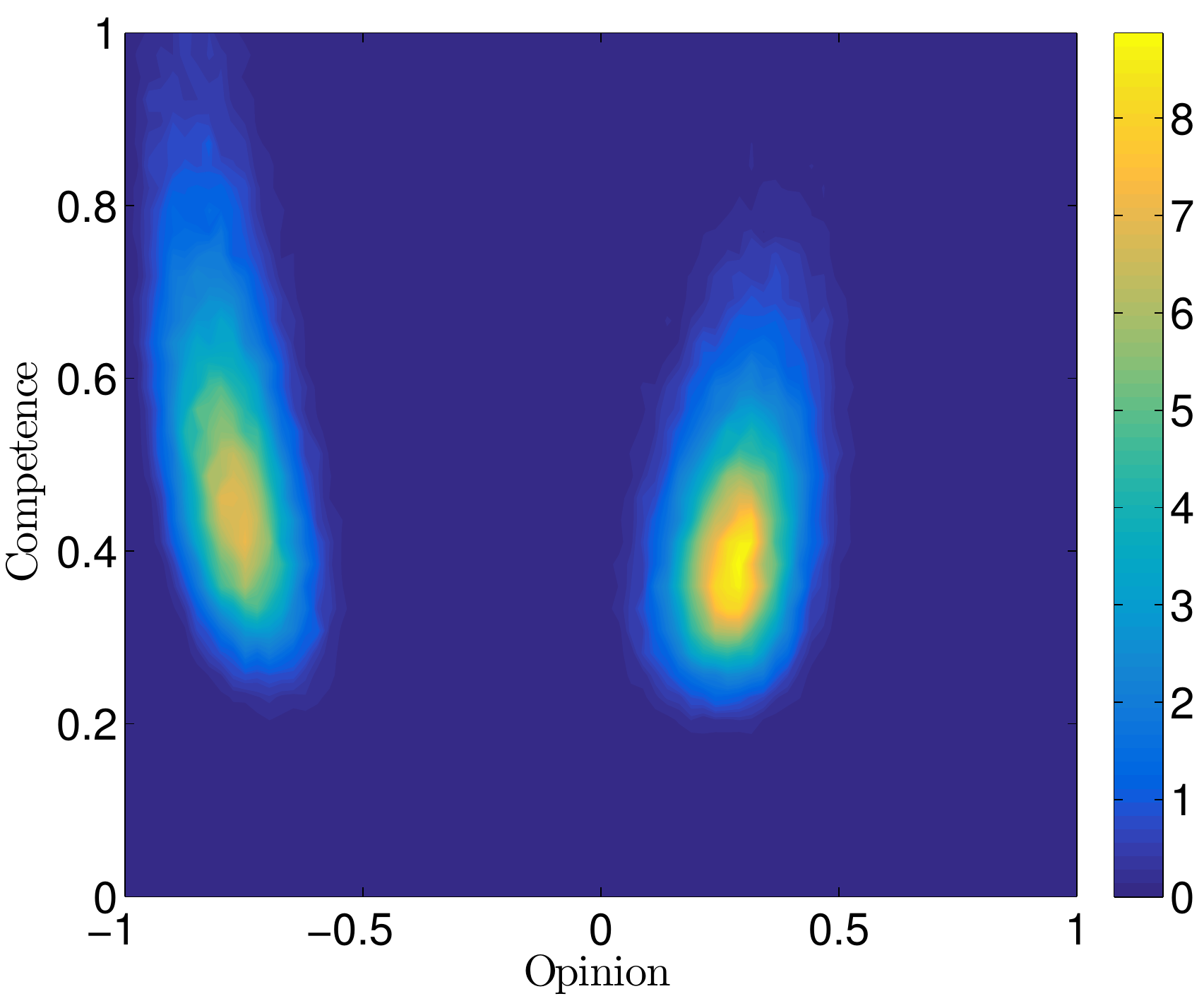}}
\subfigure[cMC t=3]{
\includegraphics[scale=0.26]{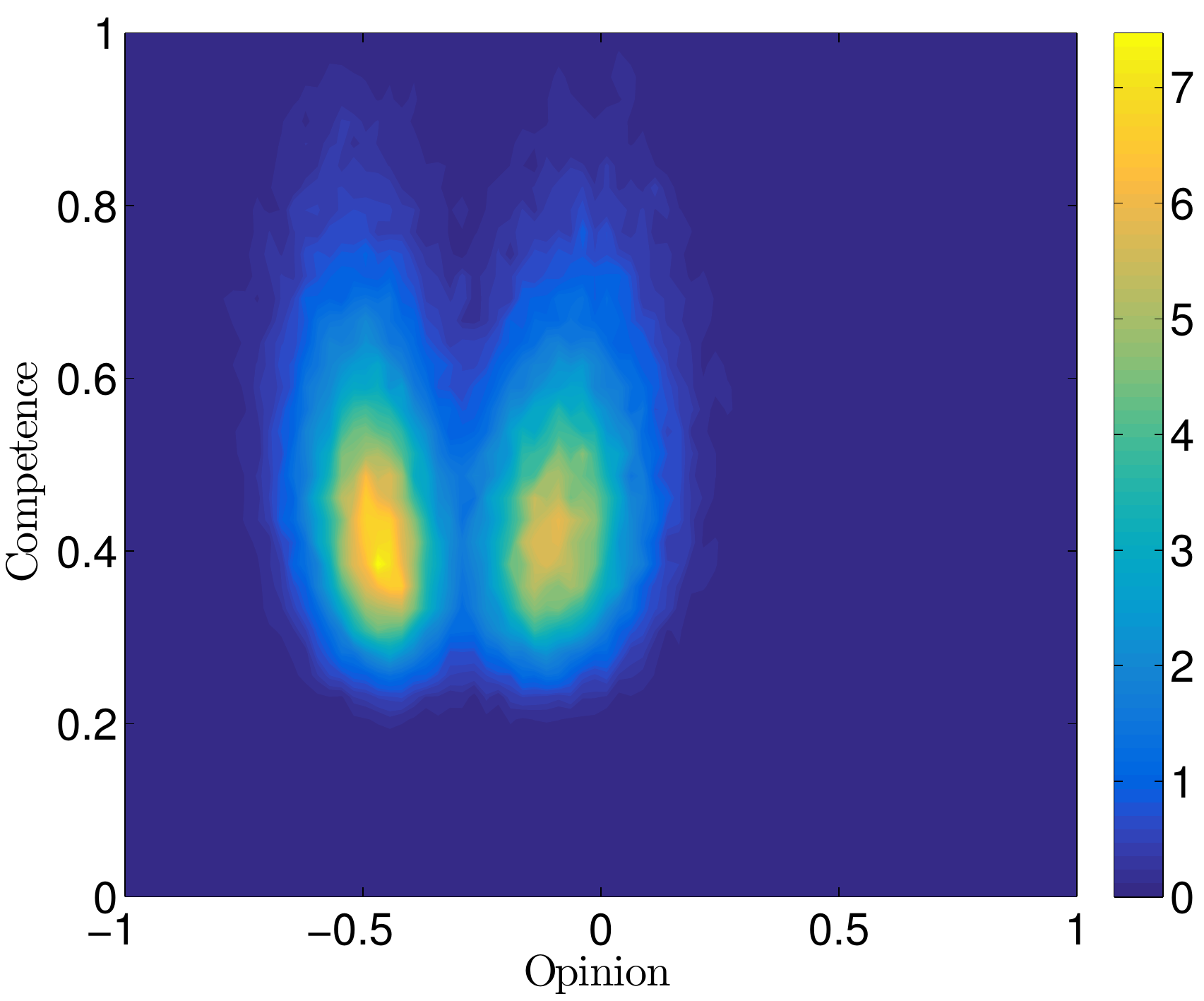}}
\subfigure[cMC t=10]{
\includegraphics[scale=0.26]{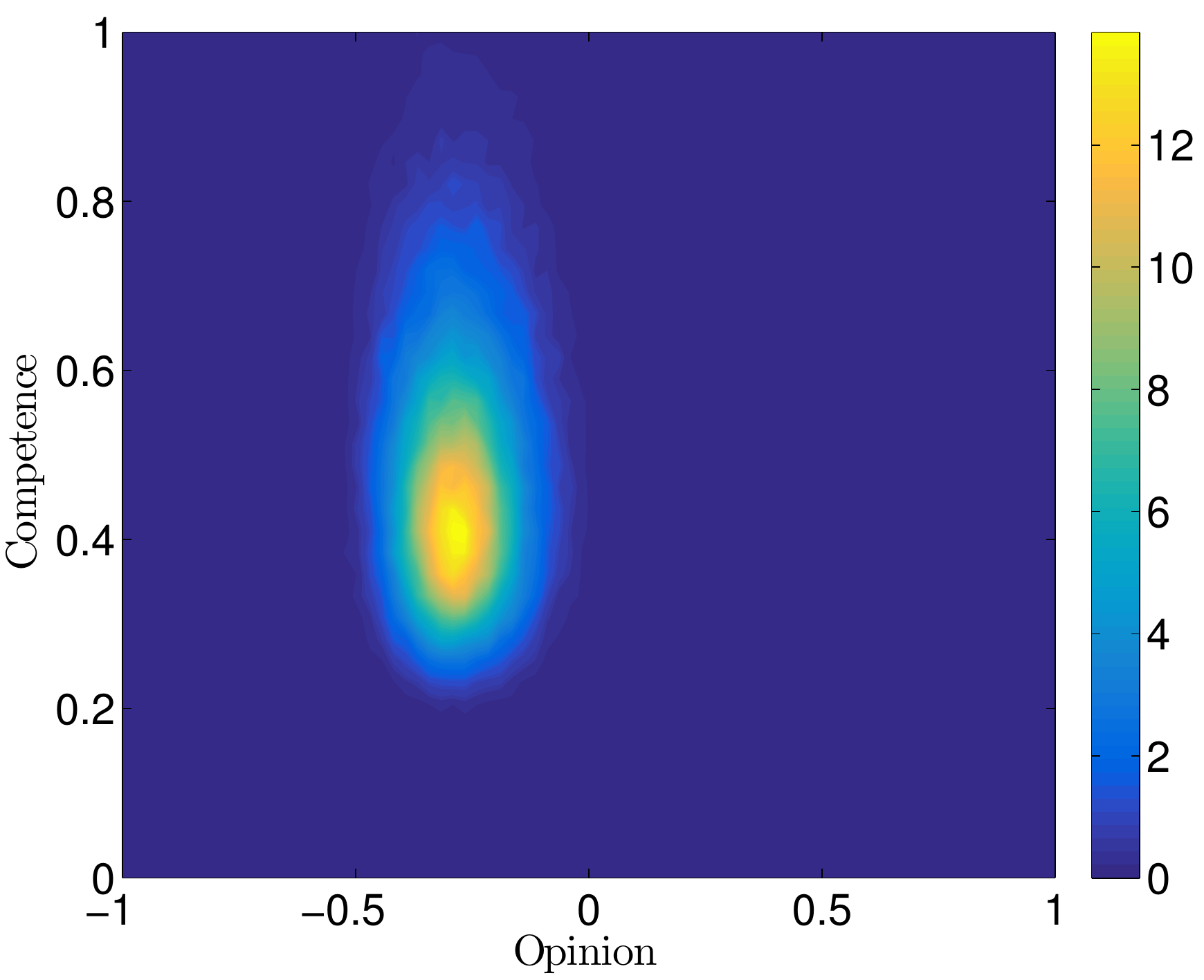}}\\
\subfigure[EB1 t=1]{
\includegraphics[scale=0.26]{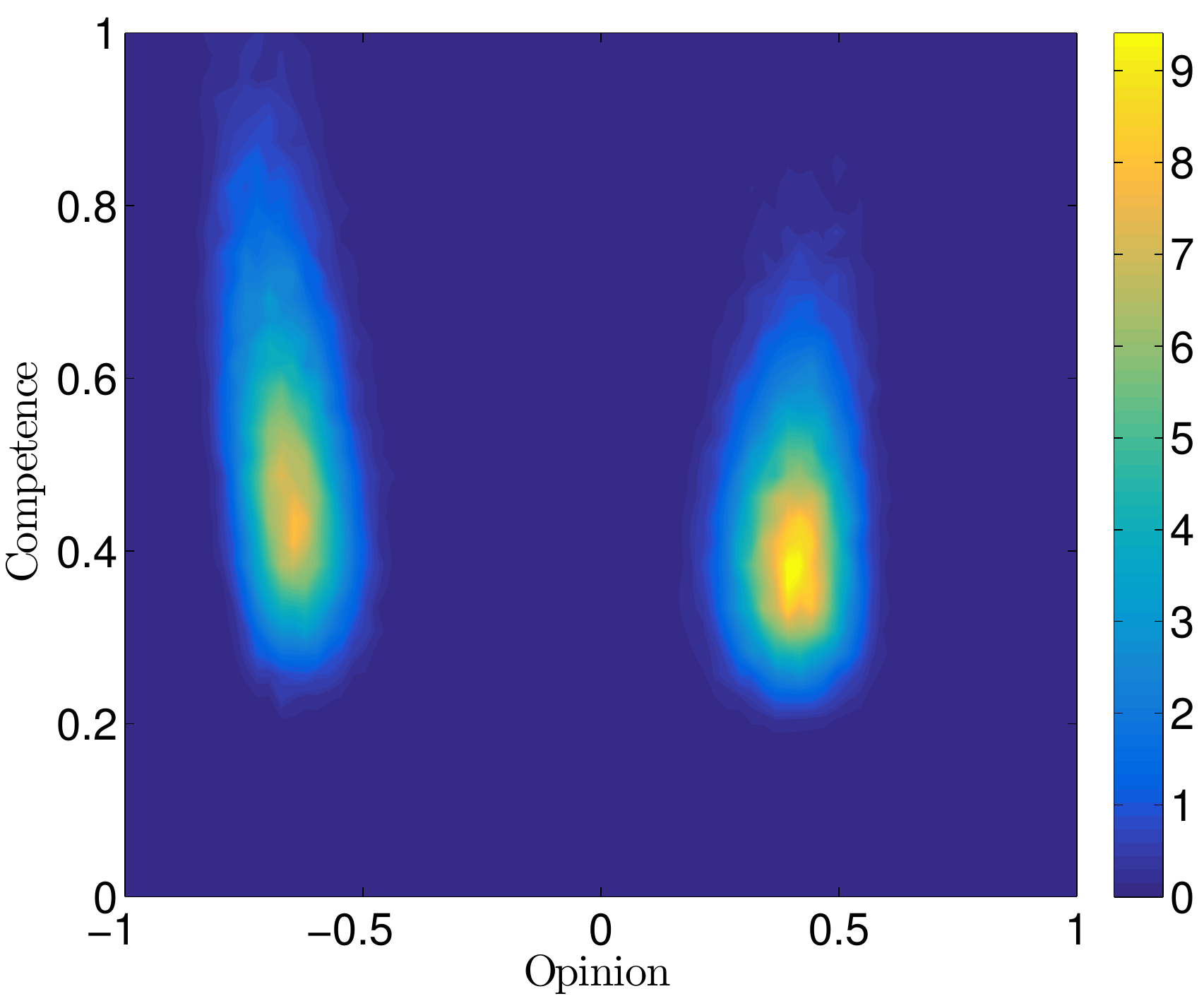}}
\subfigure[EB1 t=3]{
\includegraphics[scale=0.26]{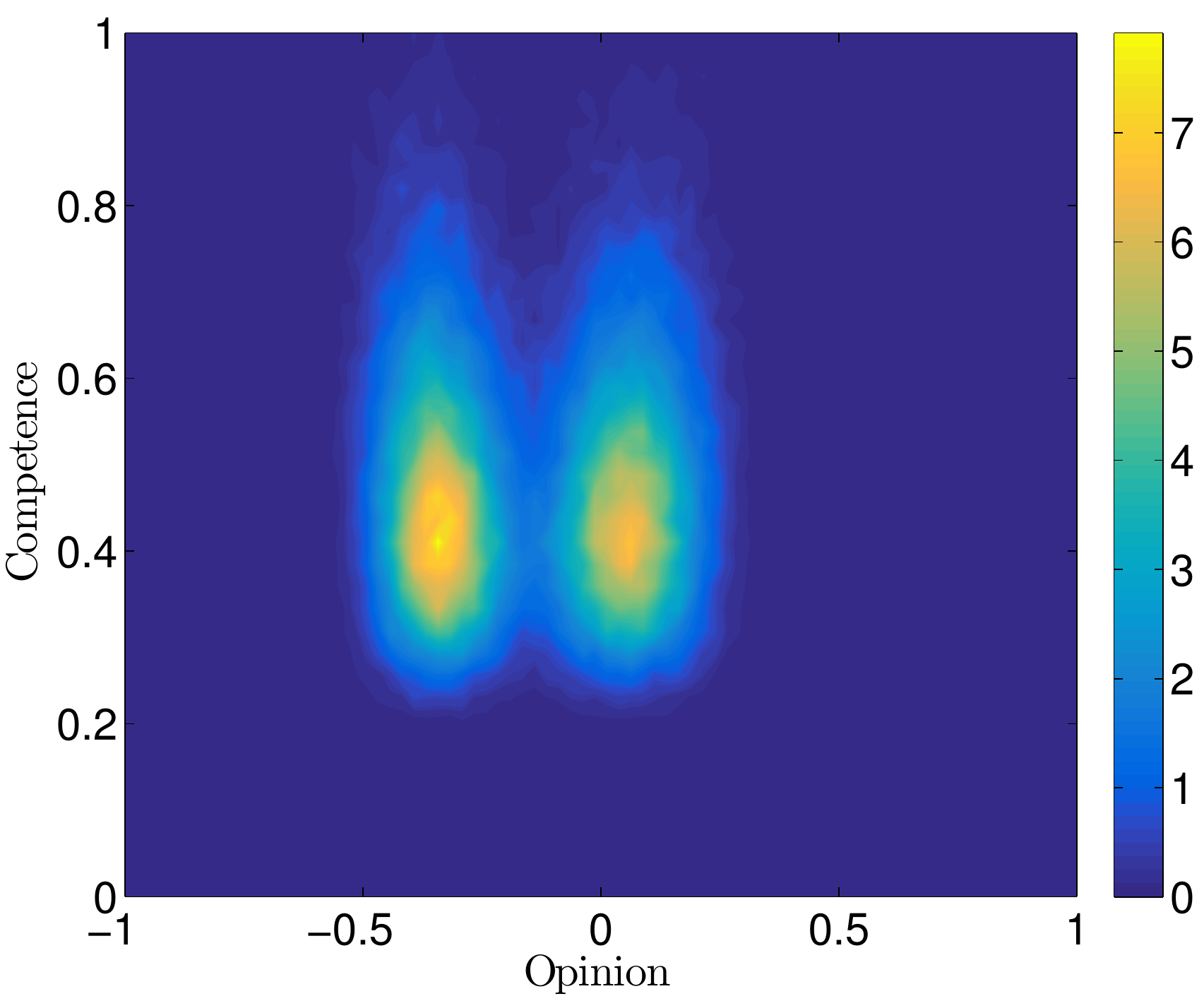}}
\subfigure[EB1 t=10]{
\includegraphics[scale=0.26]{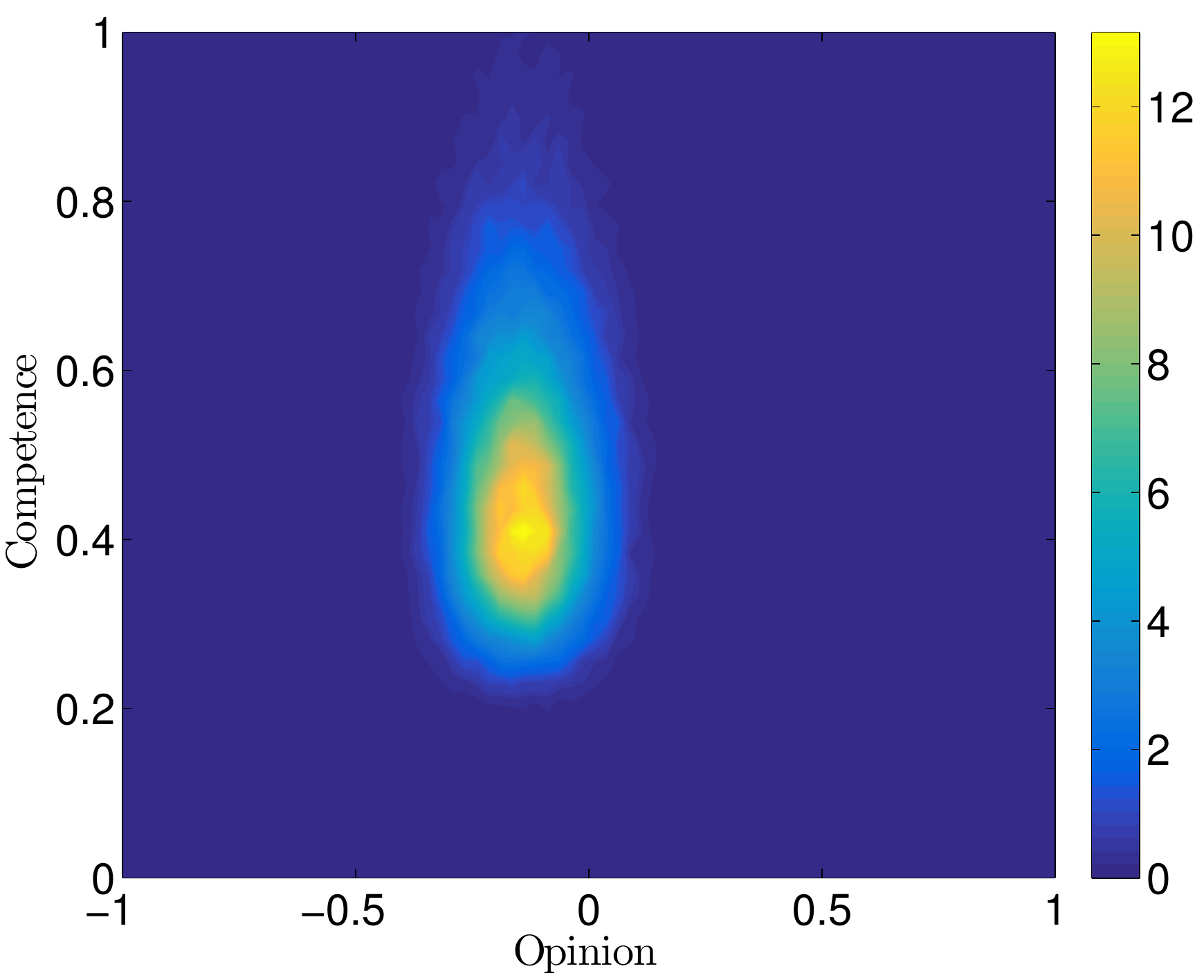}}\\
\subfigure[EB2 t=1]{
\includegraphics[scale=0.26]{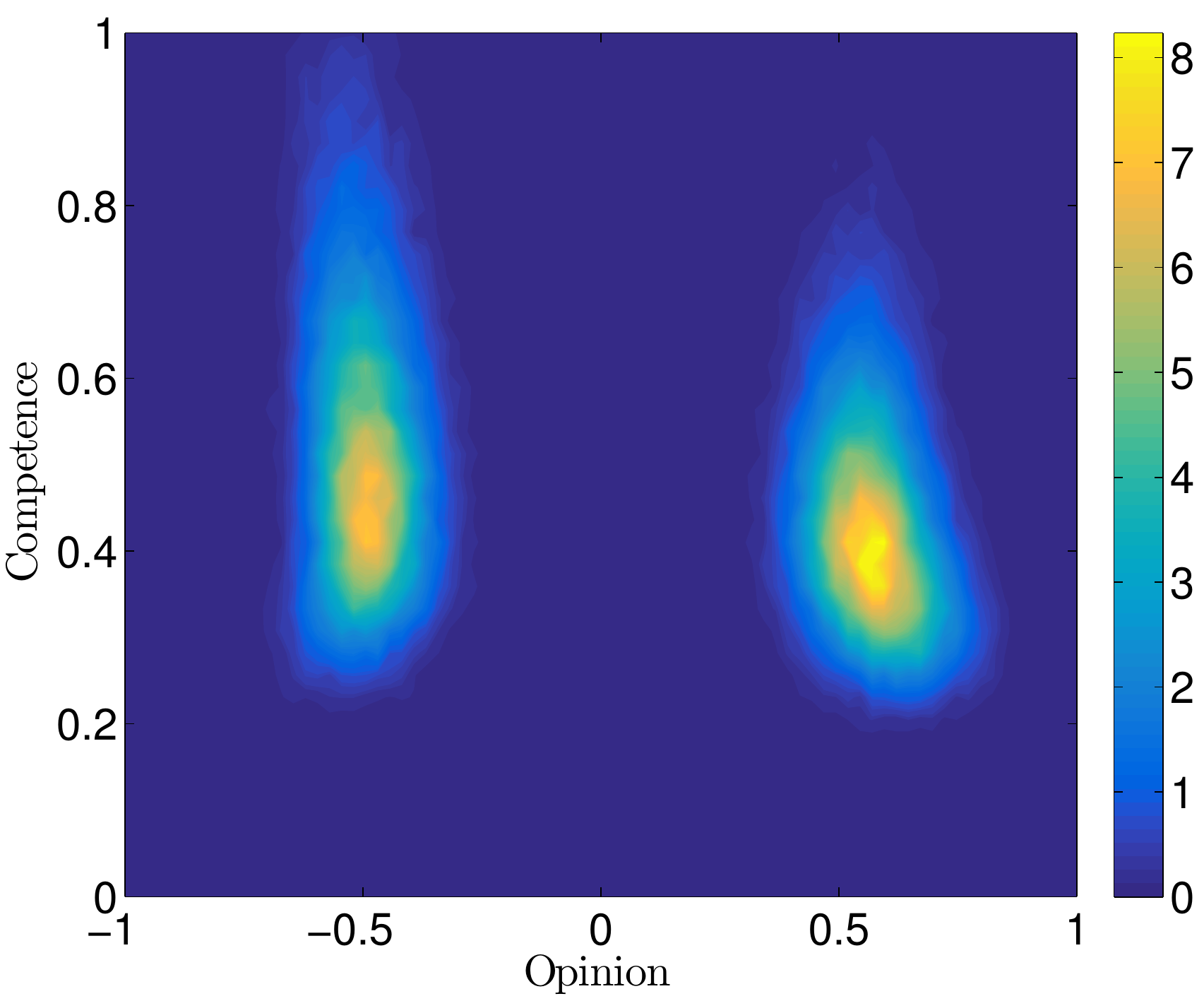}}
\subfigure[EB2 t=3]{
\includegraphics[scale=0.26]{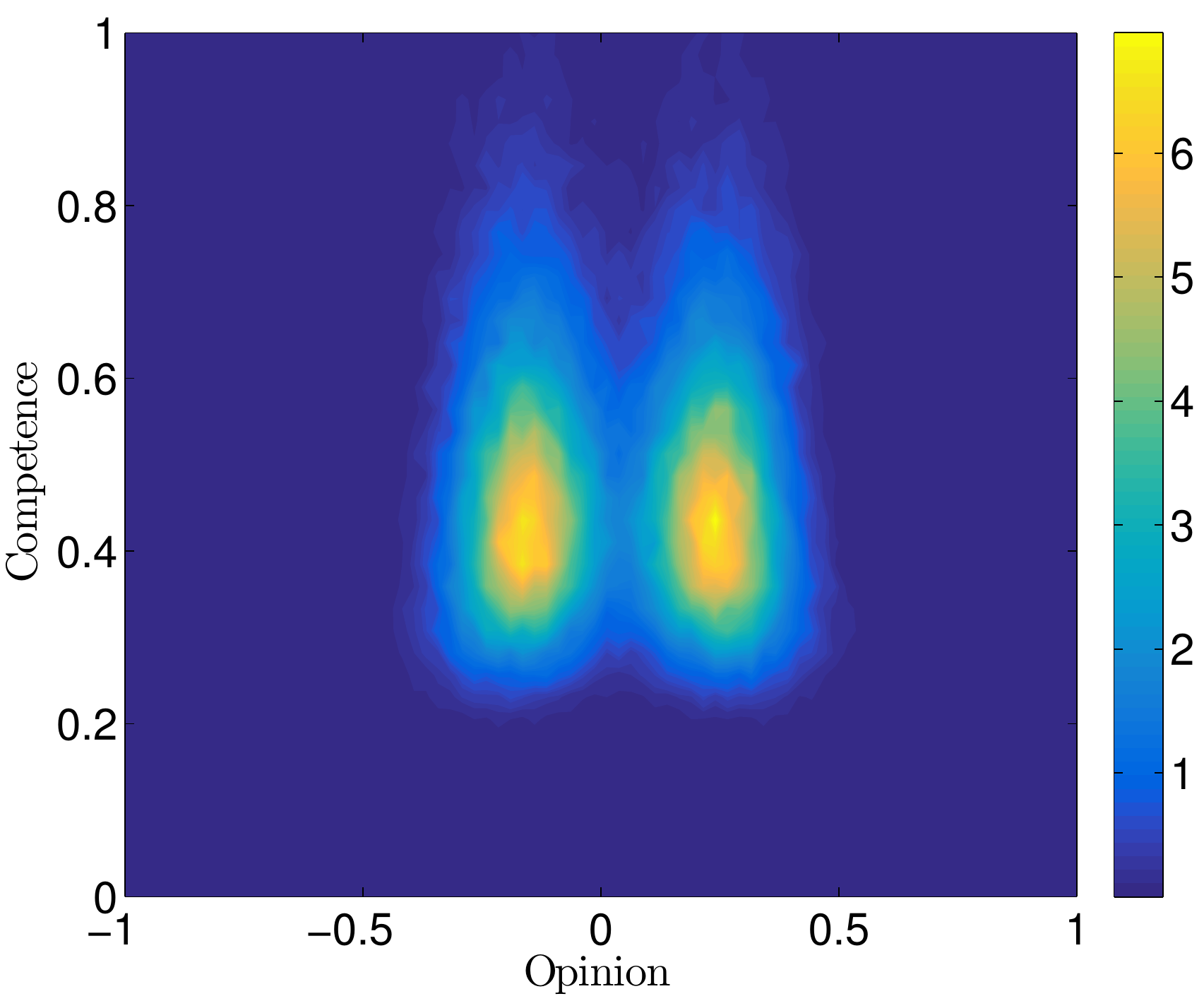}}
\subfigure[EB2 t=10]{
\includegraphics[scale=0.26]{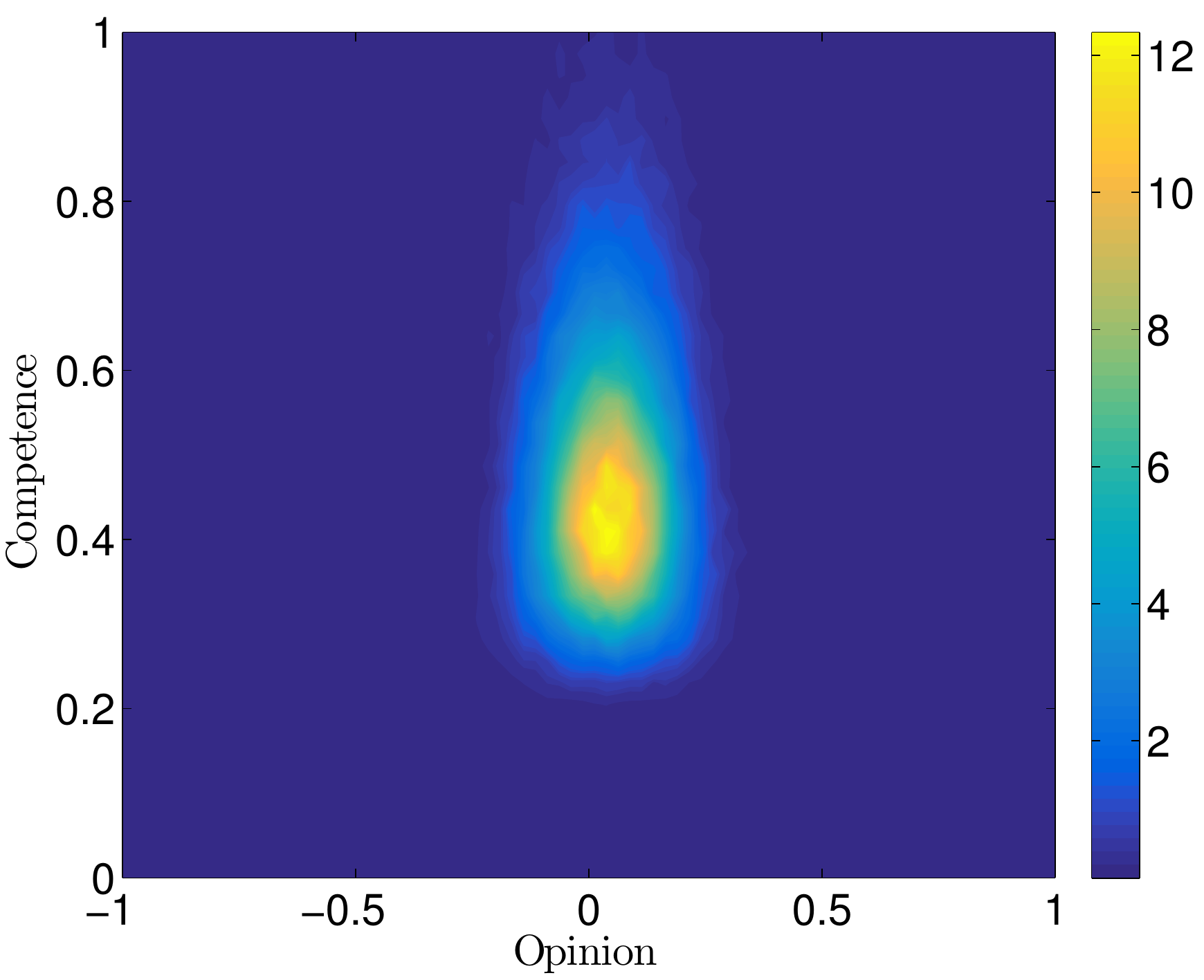}}\\
\caption{\textbf{Test 1A}: kinetic solution at different time steps in the aCF-cMC and EB1-EB2 model cases, respectively form the first row. The choice of constants in this test is: $c=10$, $\lambda_B=\lambda_C=10^{-2}$, $\sigma_\kappa^2=10^{-2}$ and $\lambda=\lambda_B+\lambda_C$. The mean competence of the multi-agent system is $\bar{x}=0.5$. We see how the equality bias, through the equality bias functions $\Phi_1(x)$ and $\Phi_2(x)$ presented in Figure \ref{fig:phi_competence}, influences the opinion dynamics, driving the collective decision toward suboptimal states with respect to the cMC model.}
\label{fig:test1}
\end{figure}

In Figure \ref{fig:test1_coll_op} we show the stationary distributions of the opinion variable in the four models of interactions both the Test 1A (left) and Test 1B (right). We observe how the EB1-EB2 models in general defines a decision which is suboptimal with respect to the competence based model cMC. Further in presence of a strong overestimation of the opinions of the less skilled agents, like in the EB2 model, we see how the emerging decision may perform worse than a aCF mode, i.e. a model which does not take into account the competence at all. In the case of the Test 1B we see how for asymmetric populations the emerging decision, also in the cMC model, may be deviated toward the positions of the less competent agents. The same behavior is highlighted in Figure \ref{fig:diffsize} where we find the evolution of the mean opinion of the multi-agent systems which lead to the formation of the collective decisions.  

%Finally, in Figure \ref{fig:diffsize} we observe the evolution of the mean opinion both for the Test 1A (left) and the Test 1B (right).

\begin{figure}
\centering
\includegraphics[scale=0.35]{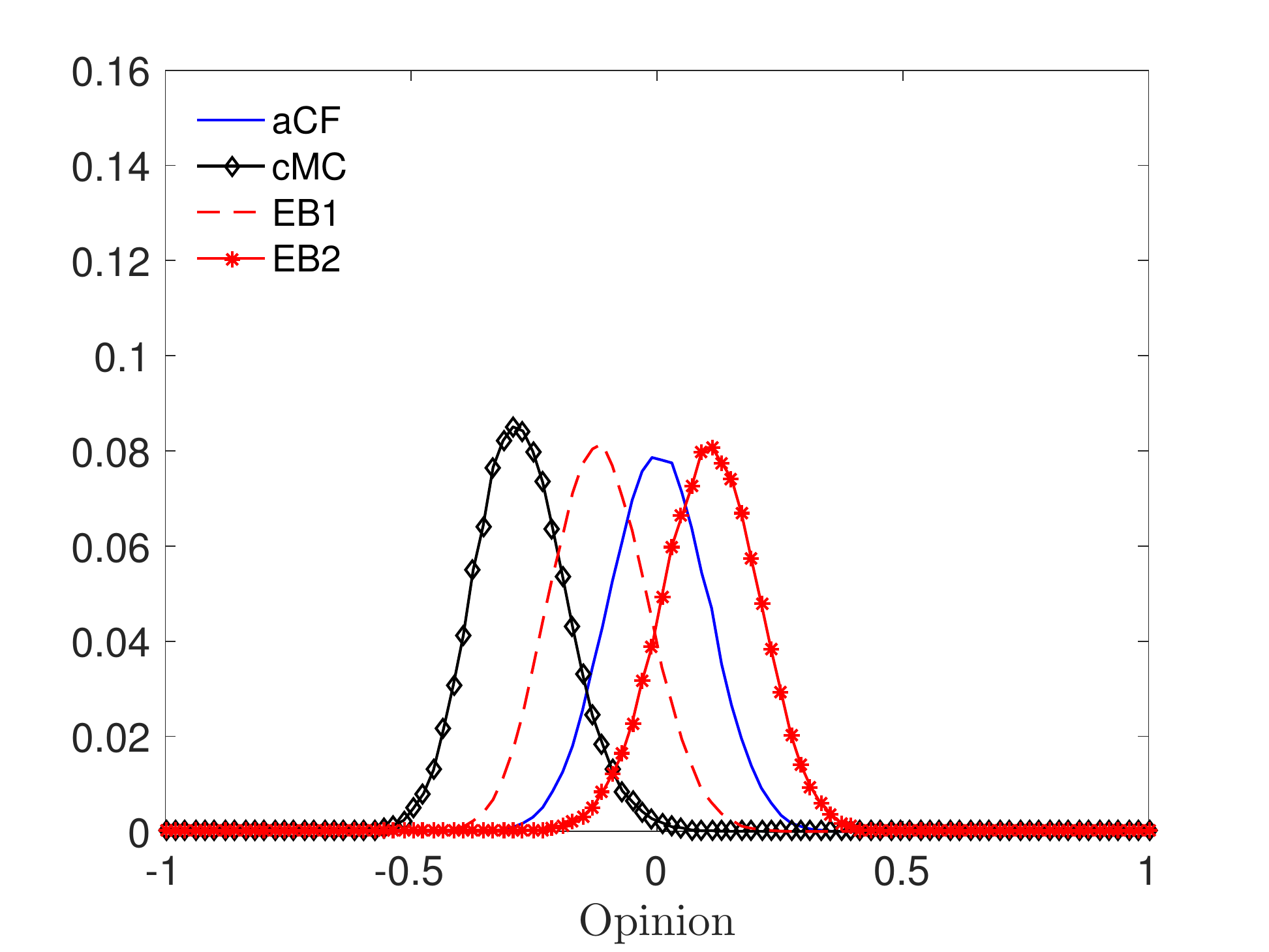}
\includegraphics[scale=0.35]{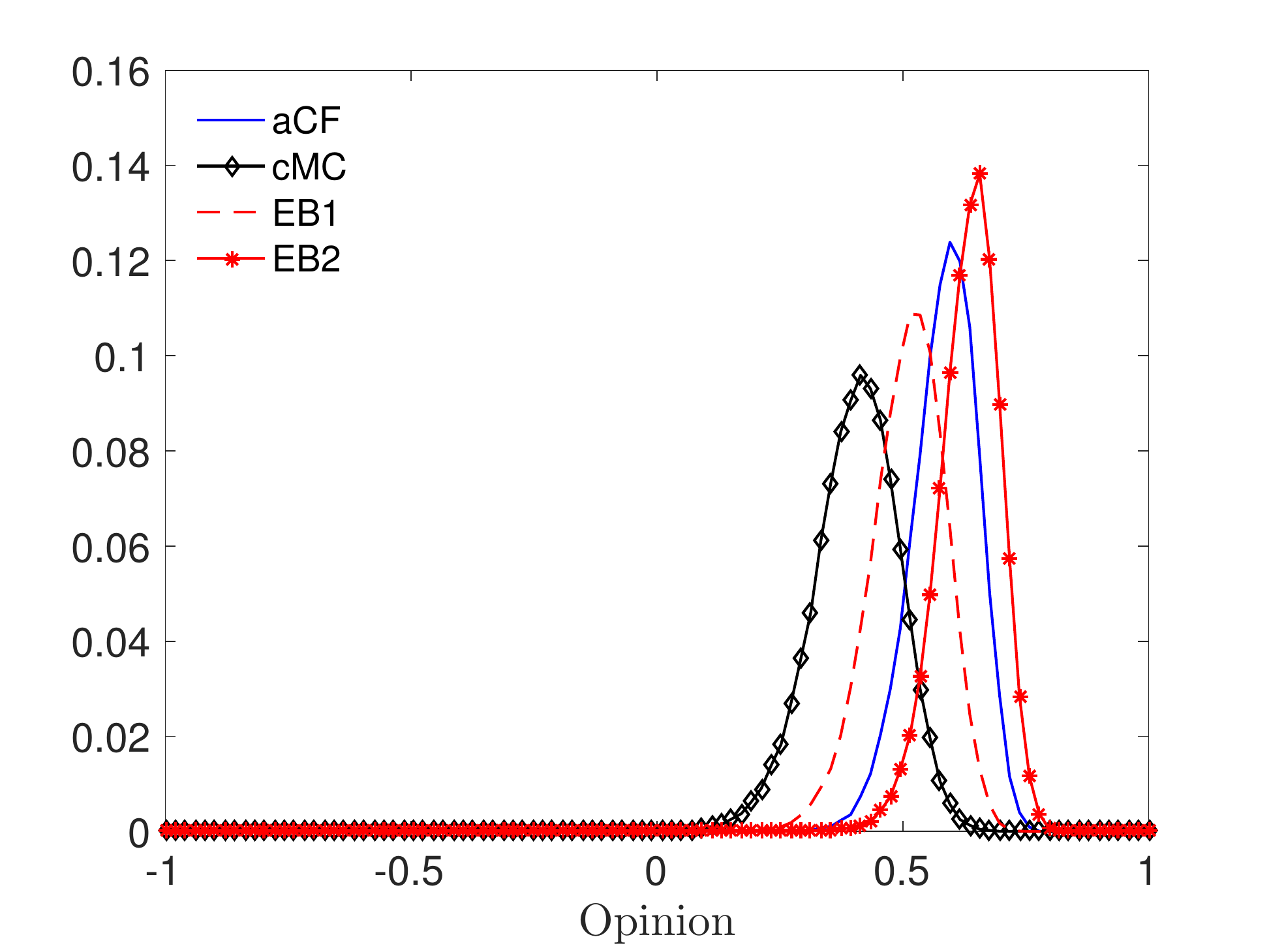}
\caption{Asymptotic distributions of the opinion variable for the models aCF-cMC-EB1-EB2
Left: \textbf{Test 1A}, case of initial configuration with equal size of competent and incompetent agents. Right: \textbf{Test1B}, case of asymmetric populations as in Figure \ref{fig:initial}. The choice of constants is $c=10$, $\lambda_B=\lambda_C=10^{-2}$, $\sigma_\kappa^2=10^{-2}$ and $\lambda=\lambda_B+\lambda_C$, the equality bias functions are sketched in Figure \ref{fig:phi_competence}. }
\label{fig:test1_coll_op}
\end{figure}

\begin{figure}
\centering
\includegraphics[scale=0.35]{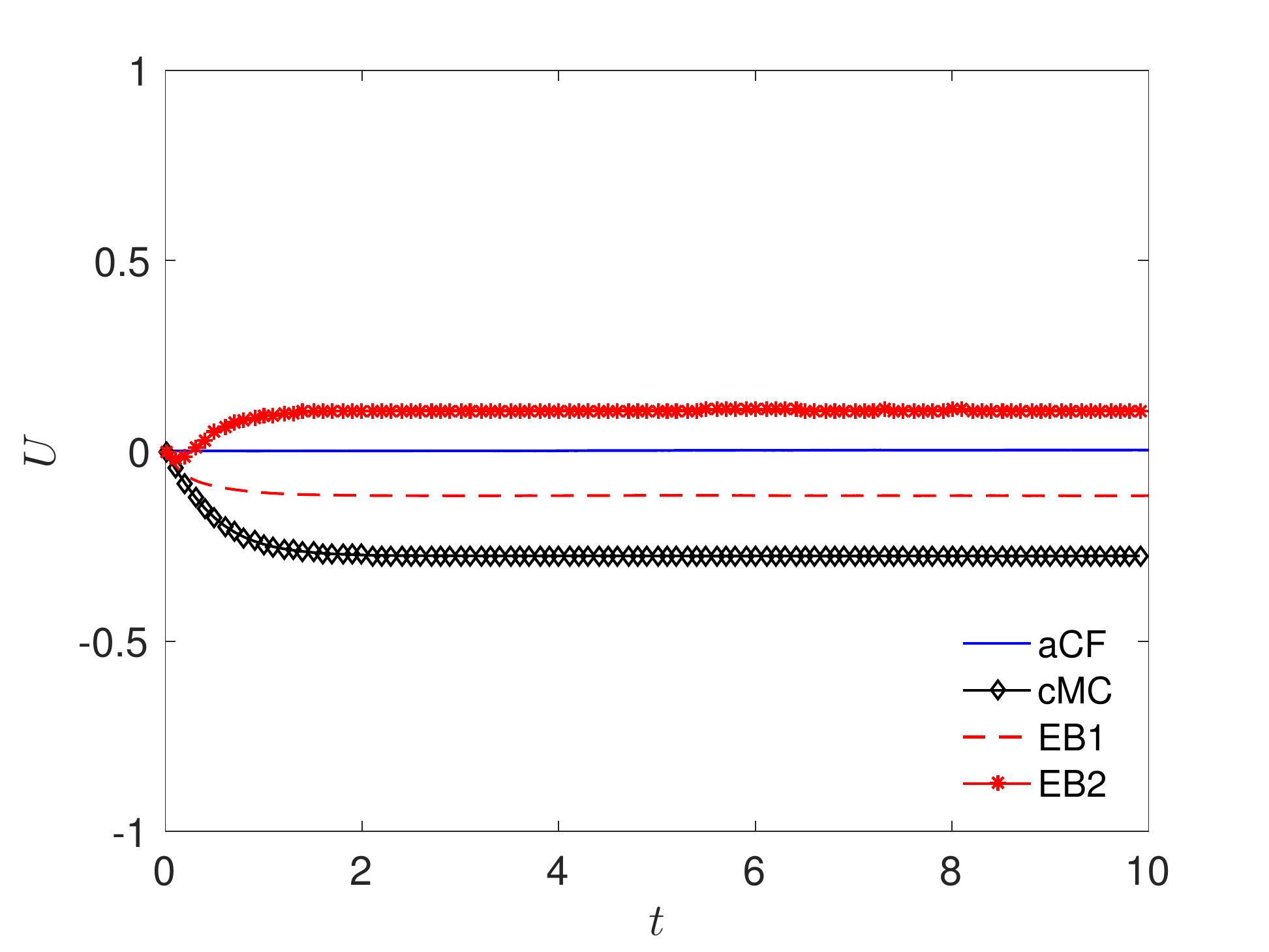}
\includegraphics[scale=0.35]{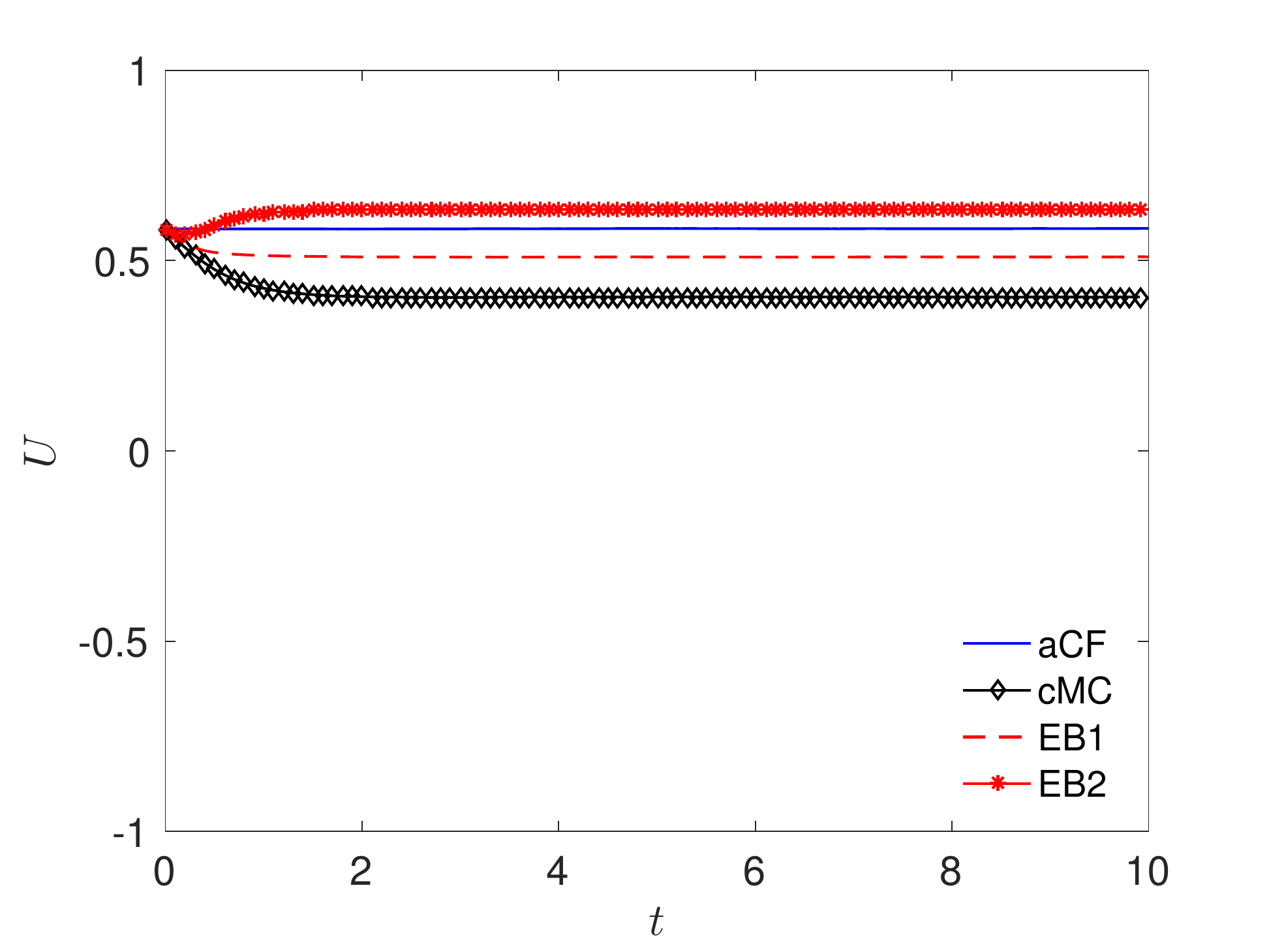}
\caption{ Collective decisions of the multi-agent system for the reference models aCF-cMC-EB1-EB2 with the choice of parameters $c=10$, $\lambda_B=\lambda_C=10^{-2}$, $\sigma_\kappa^2=10^{-2}$ and $\lambda=\lambda_B+\lambda_C$; the variable $z$ is . Left: \textbf{Test 1A}. Right: \textbf{Test1B}.
%we find the evolution of the collective decision of the system for the cited models; observe how the size gap between the two populations drive the overall system toward the opinions of the less competent agents. 
 }
\label{fig:diffsize}
\end{figure}

\subsection{Test 2:  competence driven optimal decision}
We consider here the action of the term driving the system to the correct decision introduced in the opinion dynamics \eqref{eq:binary}. There we included a competence-dependent force with a rate given by $S(\cdot)$ representing an increasing evidence in supporting the a-priori right choice $w_d\in \{-1,1 \}$. We consider here 
\be\label{eq:S_drift}
S(x)=
\begin{cases}
\min\{1,x\} & x\ge x_d \\
0 & x\le x_d
\end{cases}
\ee
and $\lambda_C=\lambda_C\chi(x\ge x_d)$. In the proposed set-up we intended to mimic the fact that extremely low skilled people $(x\leq x_d)$, in addition to make wrong choices, have not the ability to realize the inaccuracy of their decision, a phenomenon which follows form the Dunning-Kruger effect. In Figure \ref{fig:test2} we report the evolution of the multi-agent system with the same initial configuration as in Test 1B. The forcing term is characterized by \eqref{eq:S_drift} with $x_d=0.3$, which drives the opinions of sufficiently competent agents toward $w_d=-1$. Here the equality bias functions $\Phi_1(x)$ and $\Phi_2(x)$, introduced in \ref{fig:initial}, influences the speed of convergence of the opinions toward $w_d$. We compare in Figure \ref{fig:mean_selfprop} the convergences toward $w_d$ for the models aCF-cMC and EB1-EB2.

\begin{figure}
\centering
\subfigure[aCF t=1]{
\includegraphics[scale=0.26]{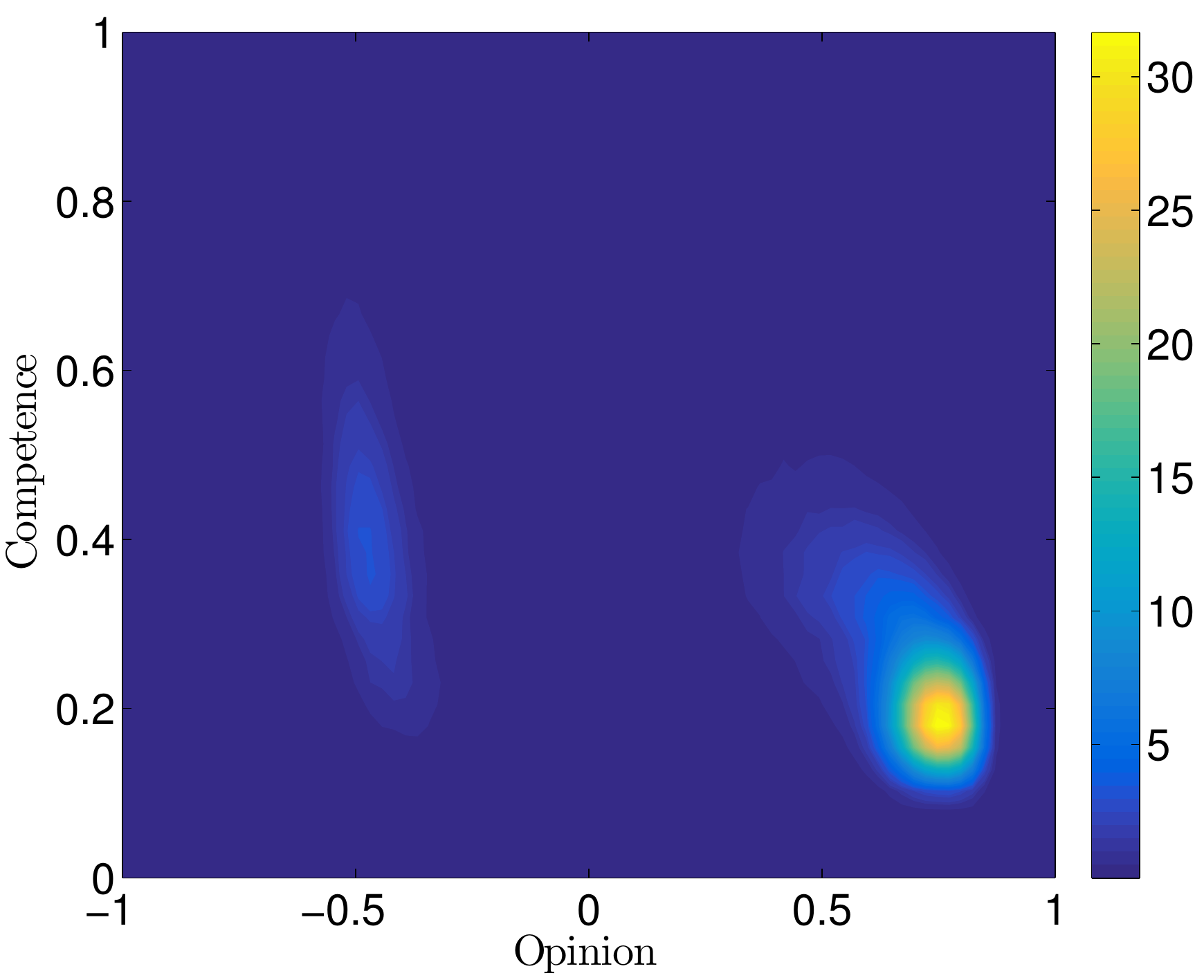}}
\subfigure[aCF t=5]{
\includegraphics[scale=0.26]{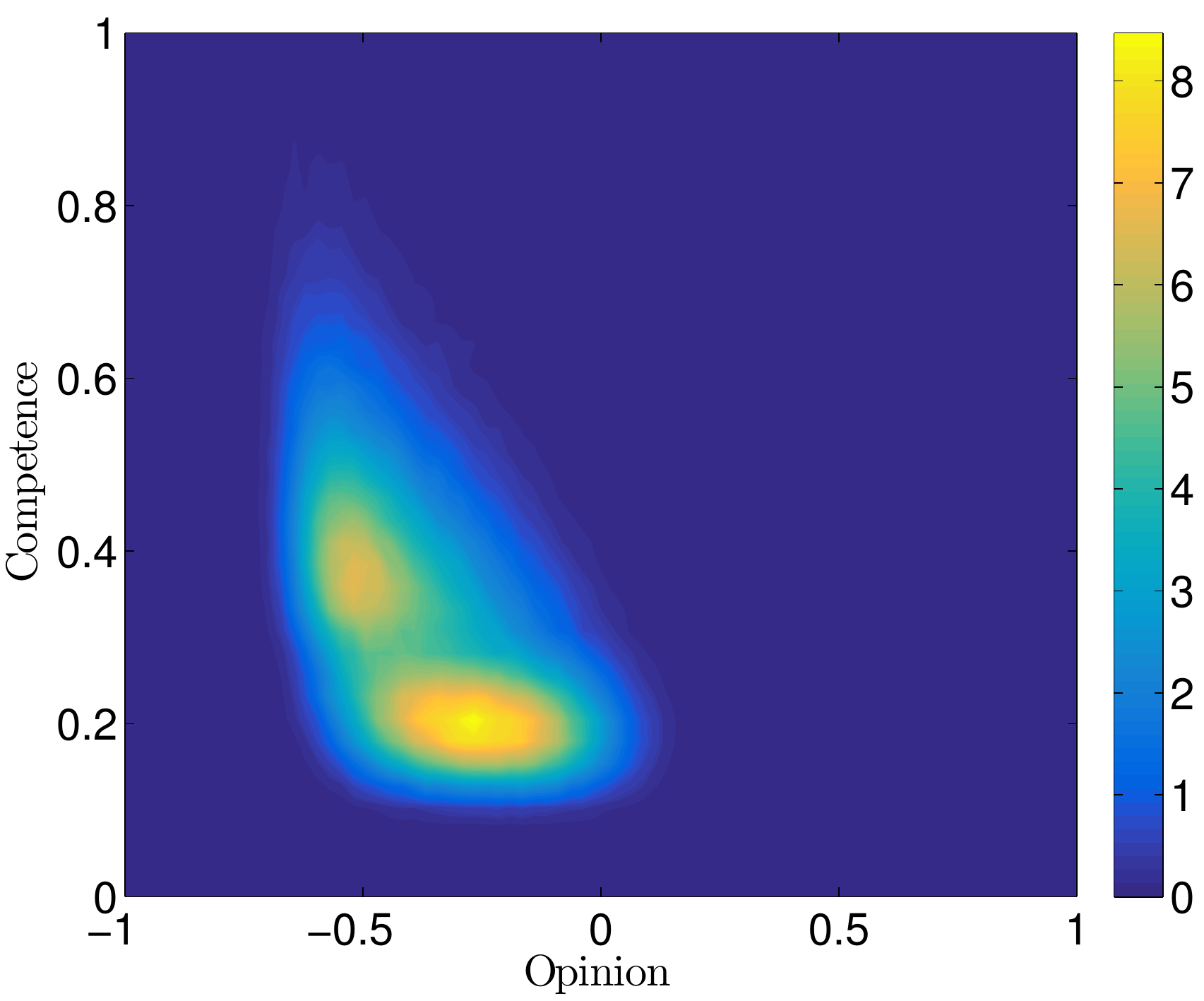}}
\subfigure[aCF t=10]{
\includegraphics[scale=0.26]{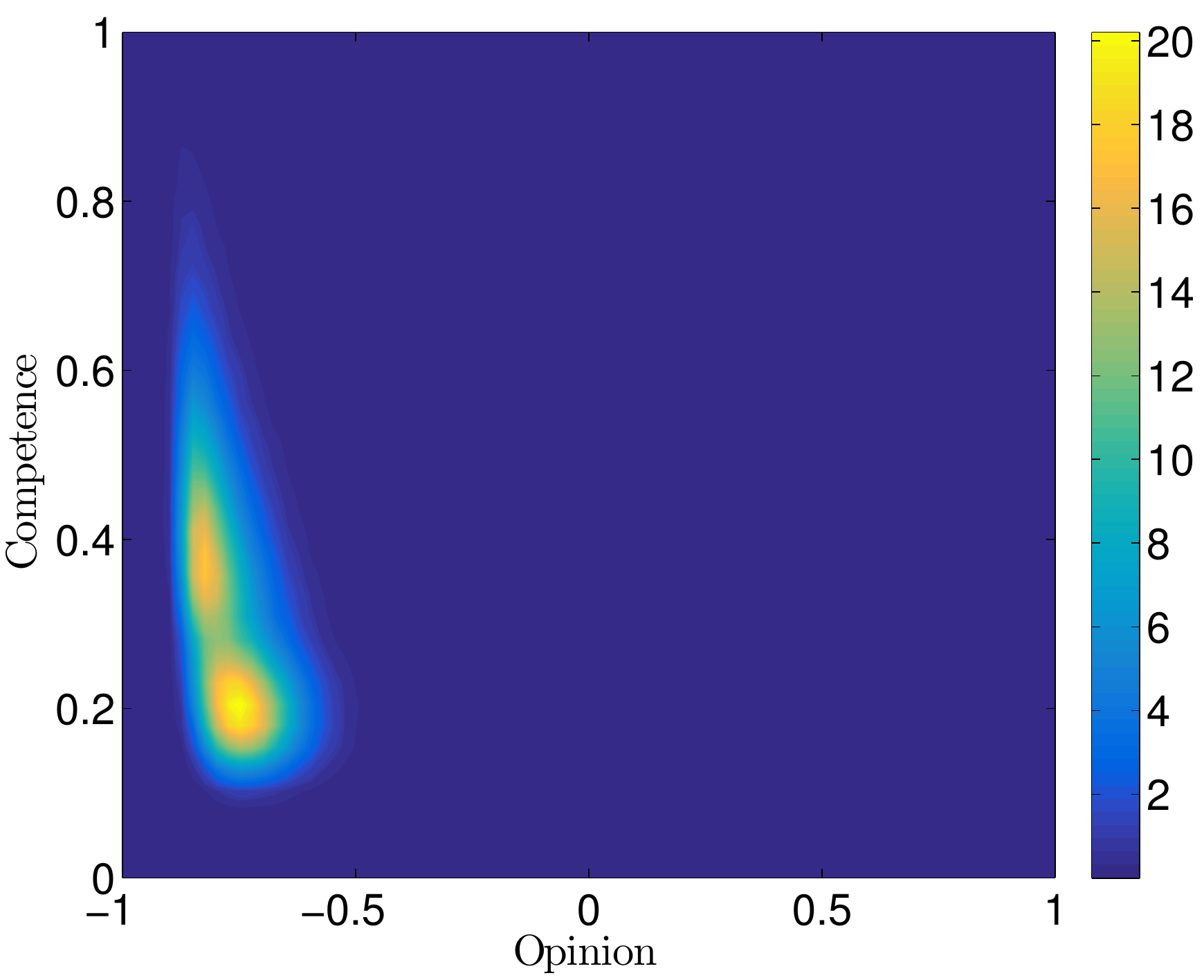}} \\
\subfigure[cMC t=1]{
\includegraphics[scale=0.26]{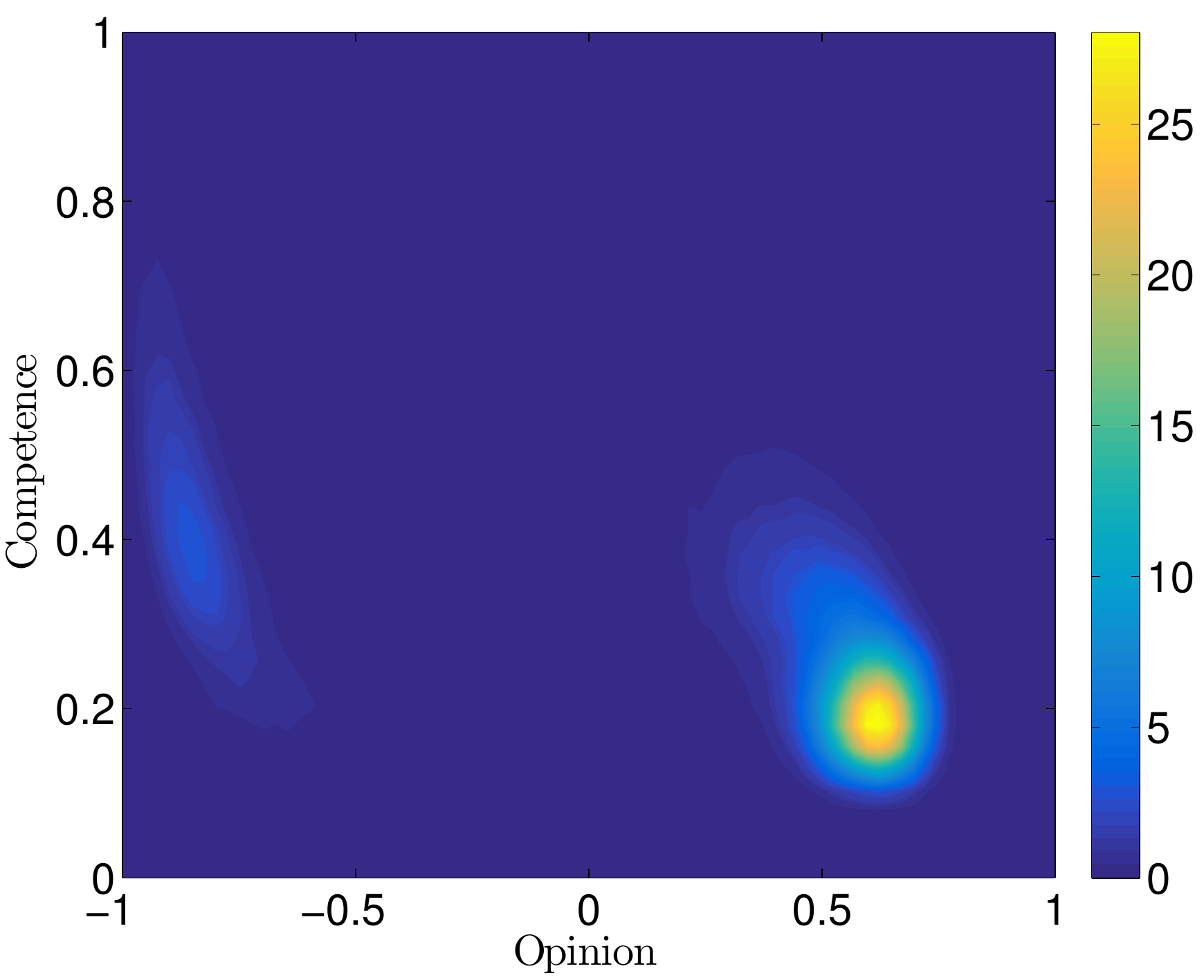}}
\subfigure[cMC t=5]{
\includegraphics[scale=0.26]{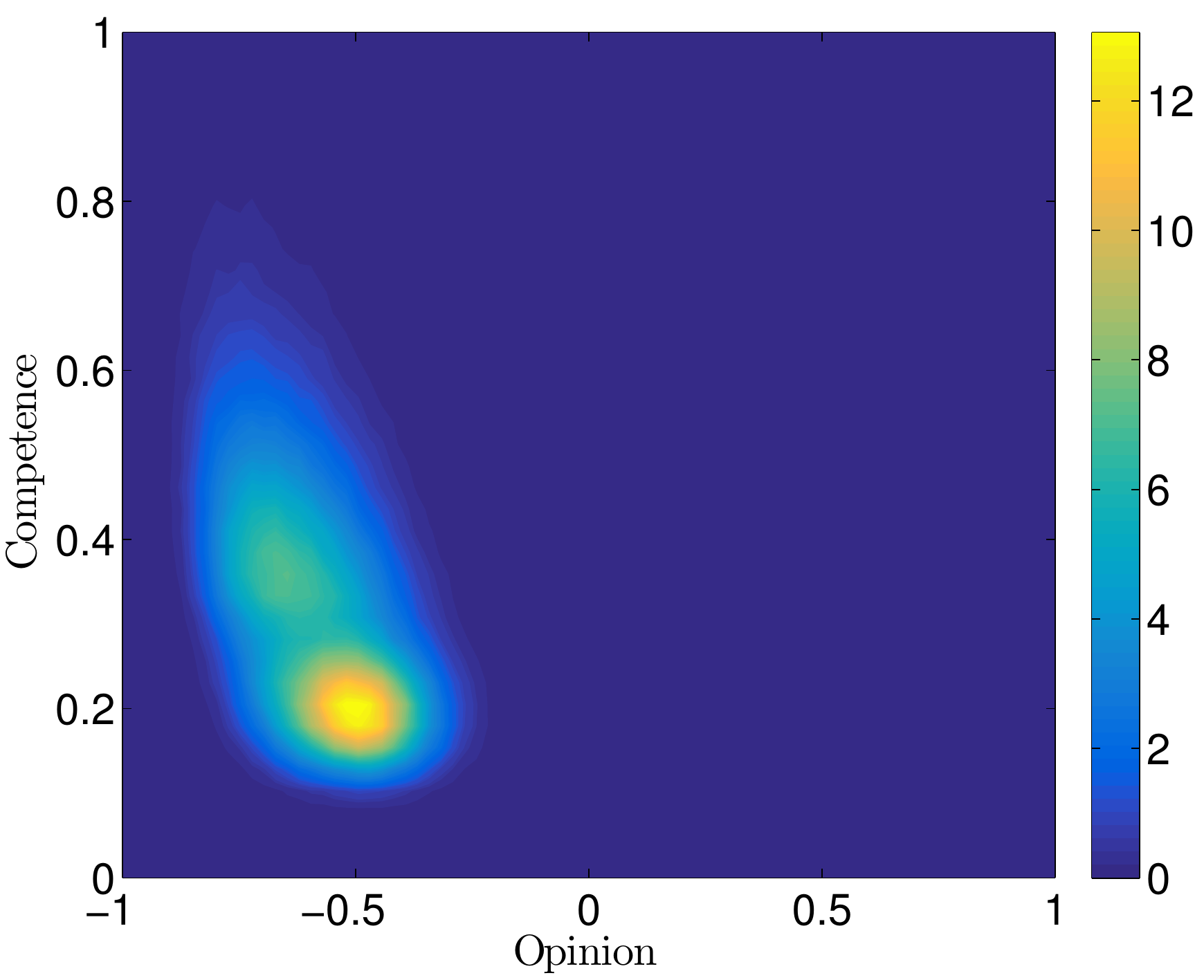}}
\subfigure[cMC t=10 ]{
\includegraphics[scale=0.26]{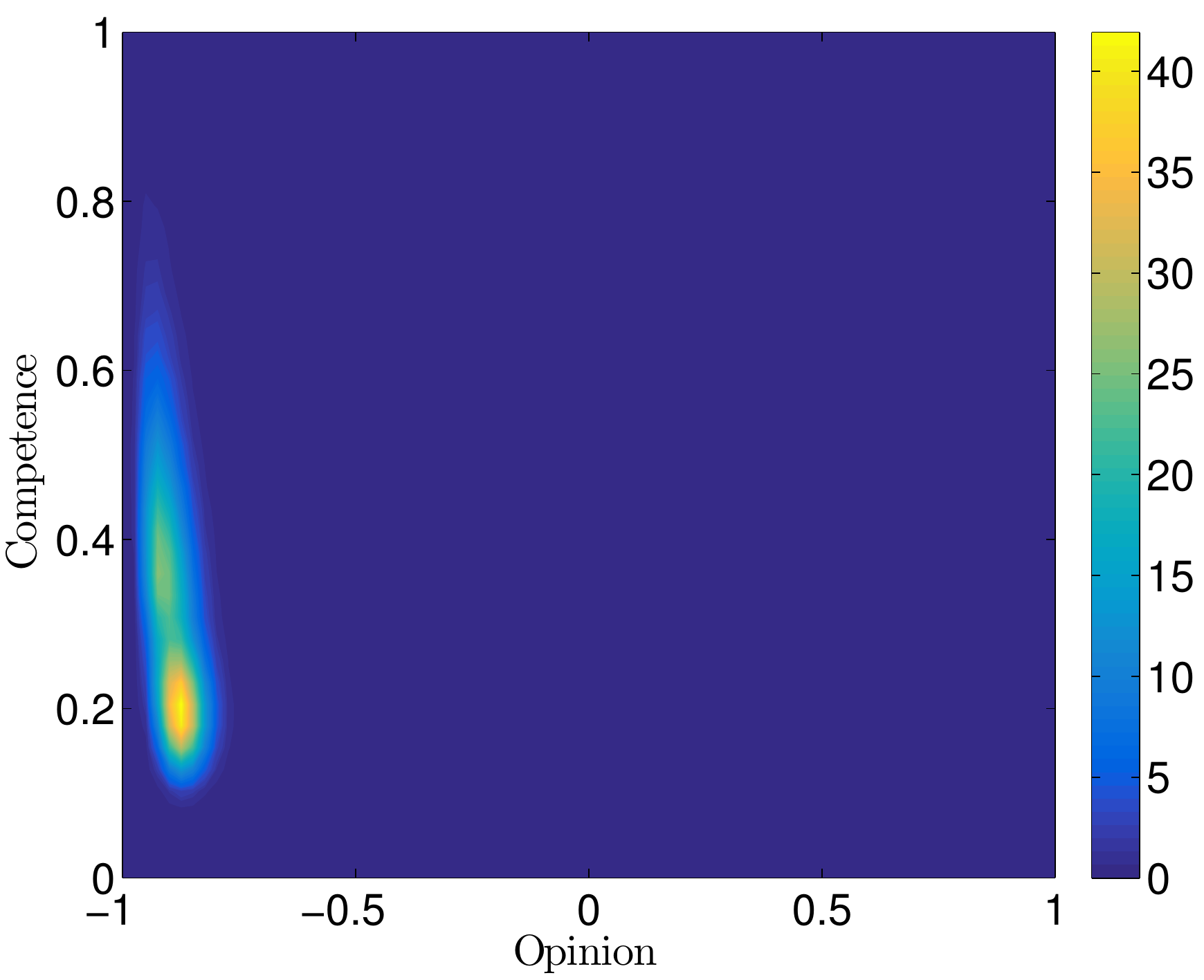}} \\
\subfigure[EB1 t=1]{
\includegraphics[scale=0.26]{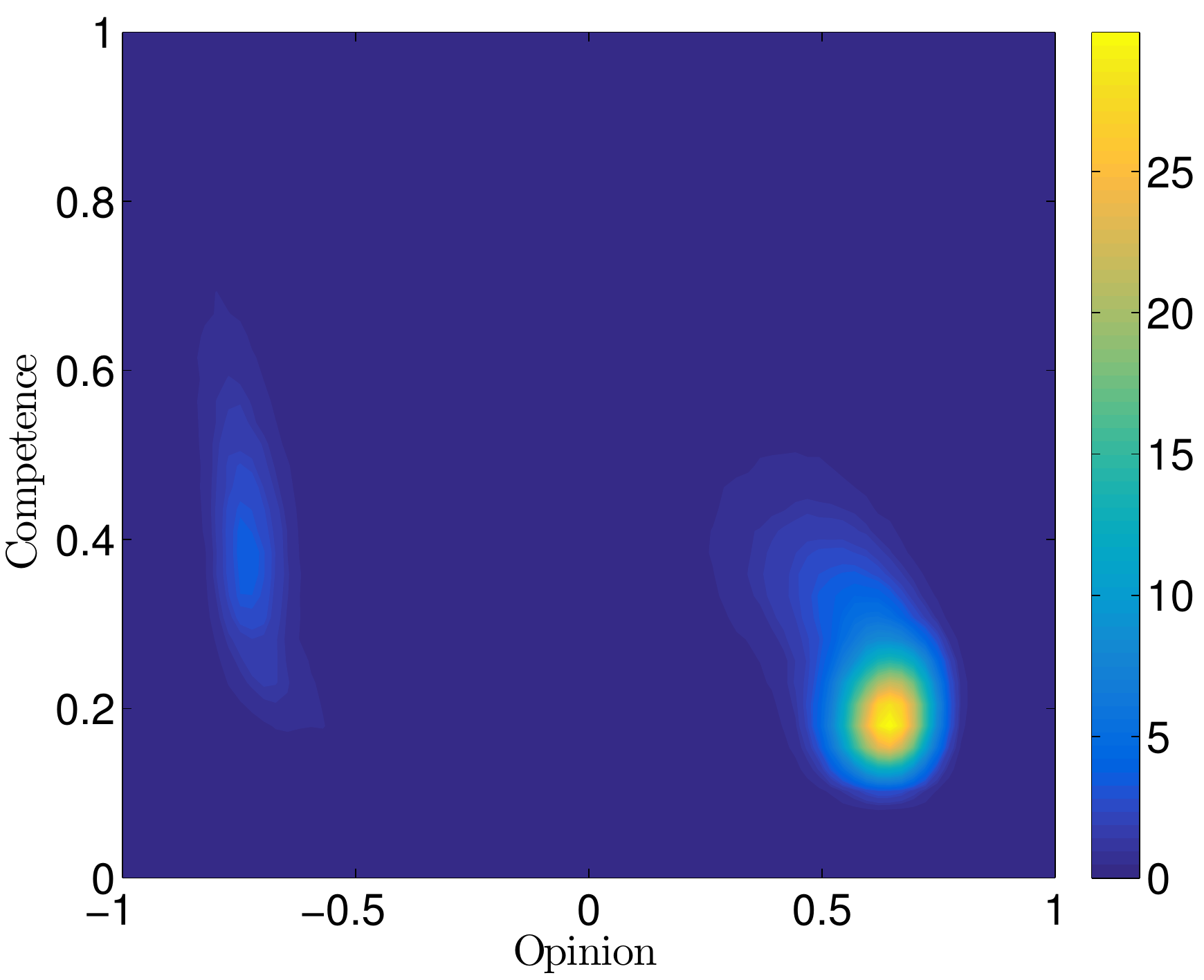}}
\subfigure[EB1 t=5]{
\includegraphics[scale=0.26]{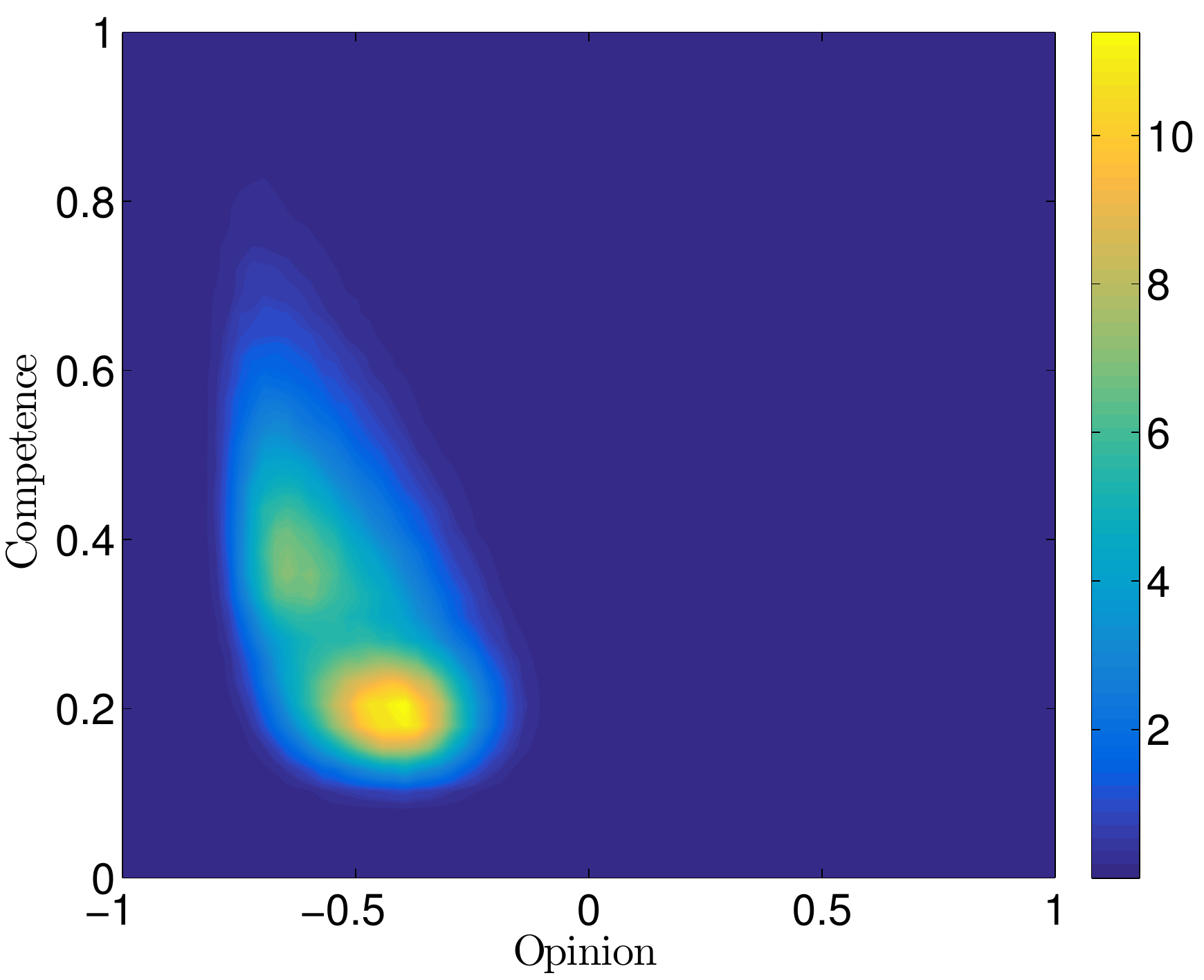}}
\subfigure[EB1 t=10]{
\includegraphics[scale=0.26]{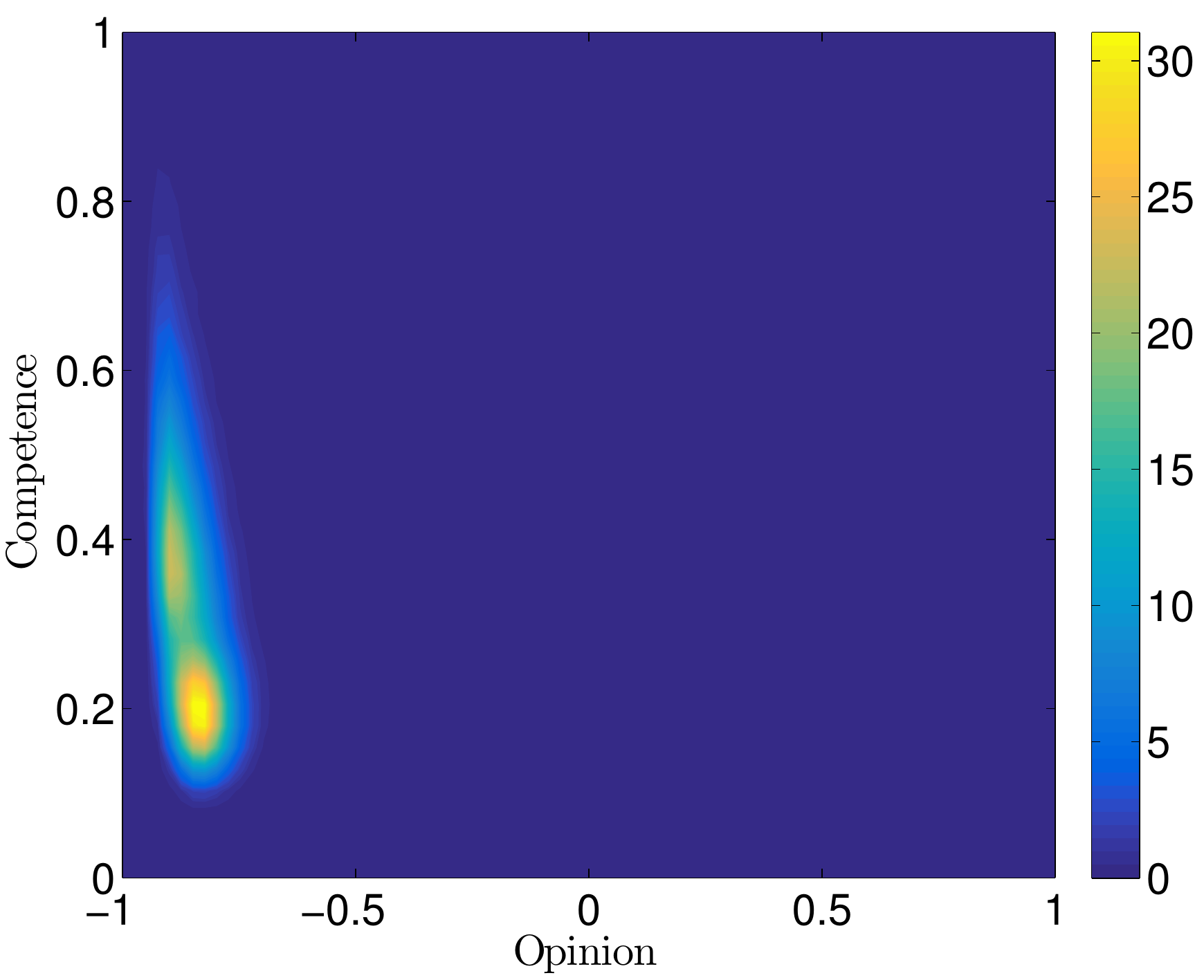}}\\
\subfigure[EB2 t=1]{
\includegraphics[scale=0.26]{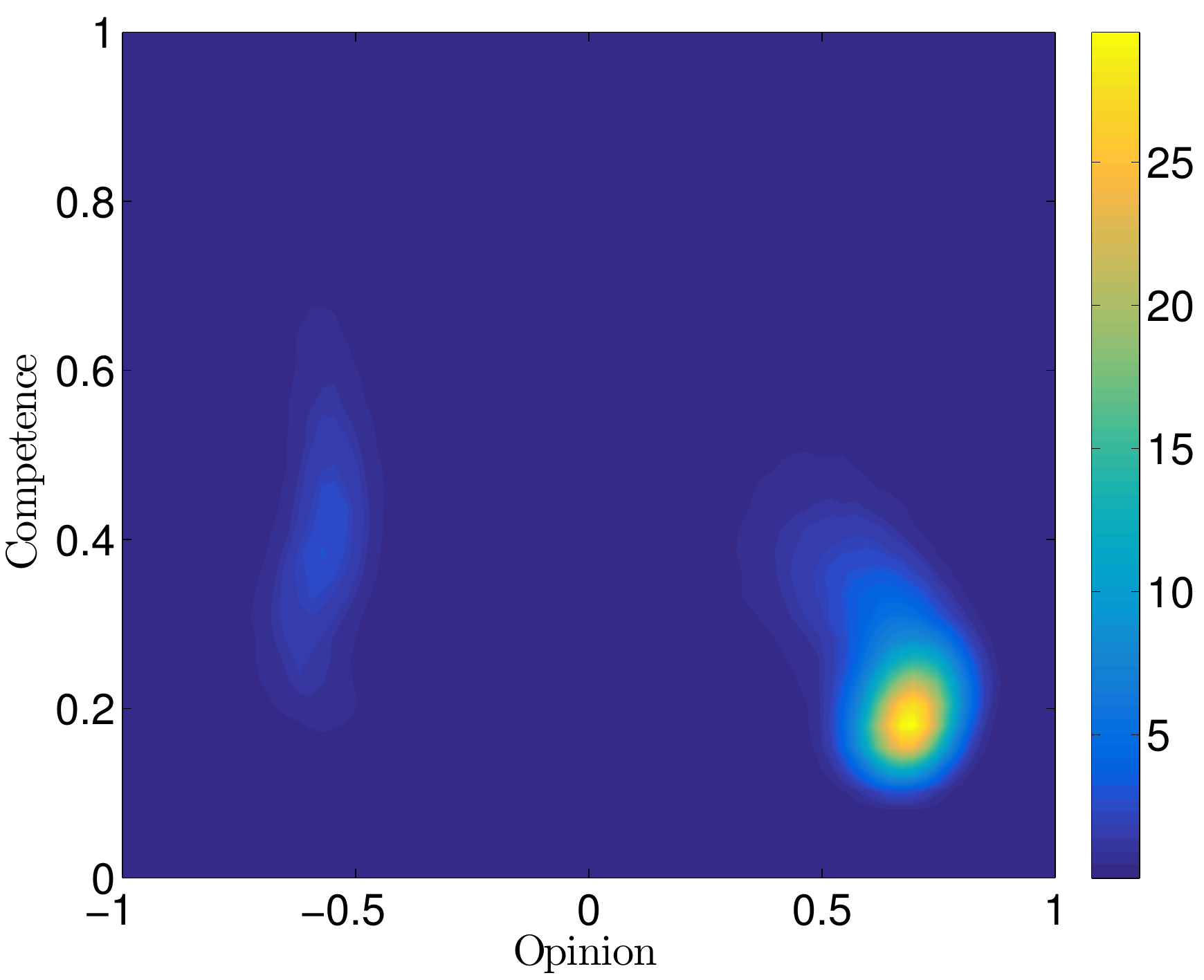}}
\subfigure[EB2 t=5]{
\includegraphics[scale=0.26]{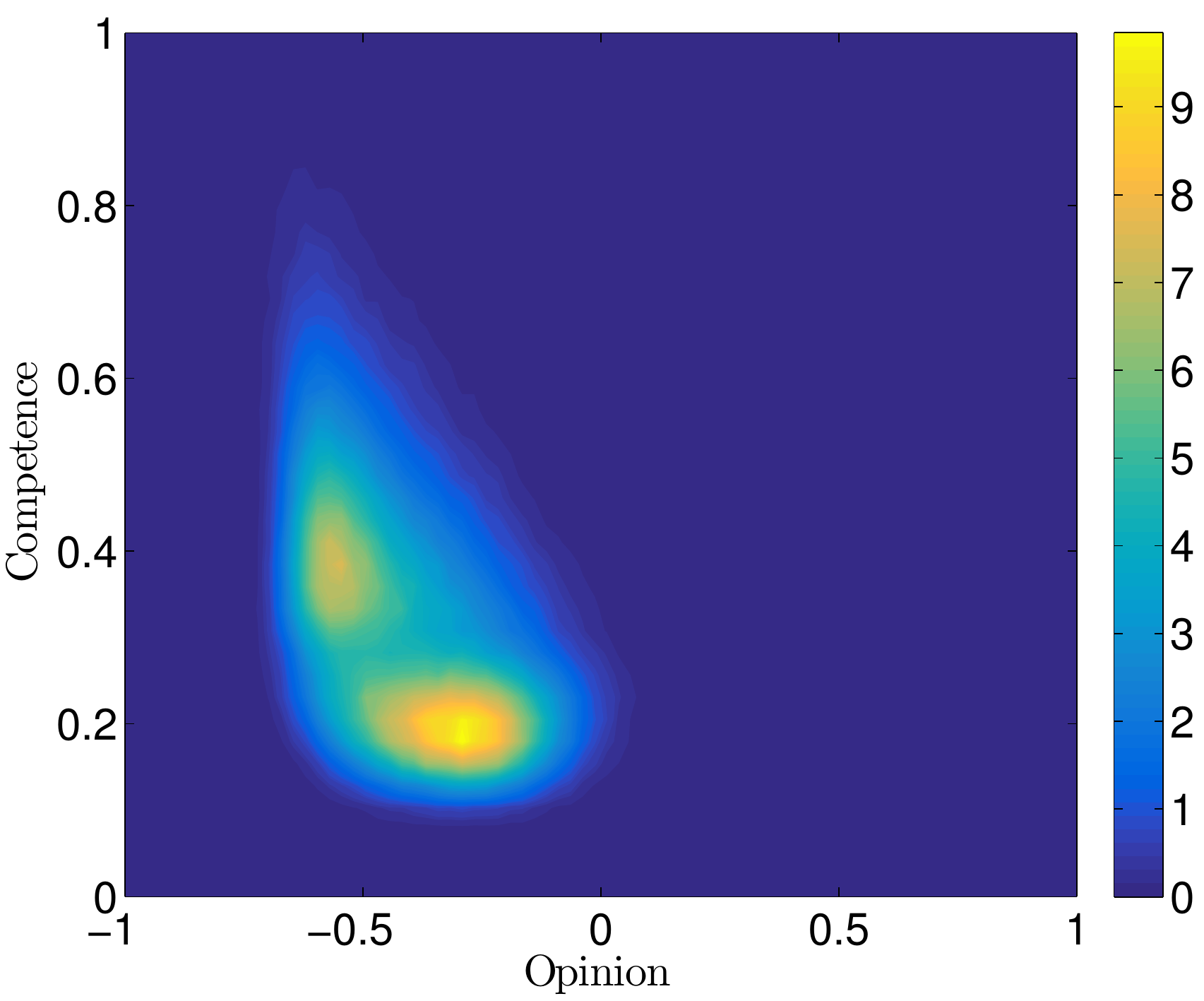}}
\subfigure[EB2 t=10]{
\includegraphics[scale=0.26]{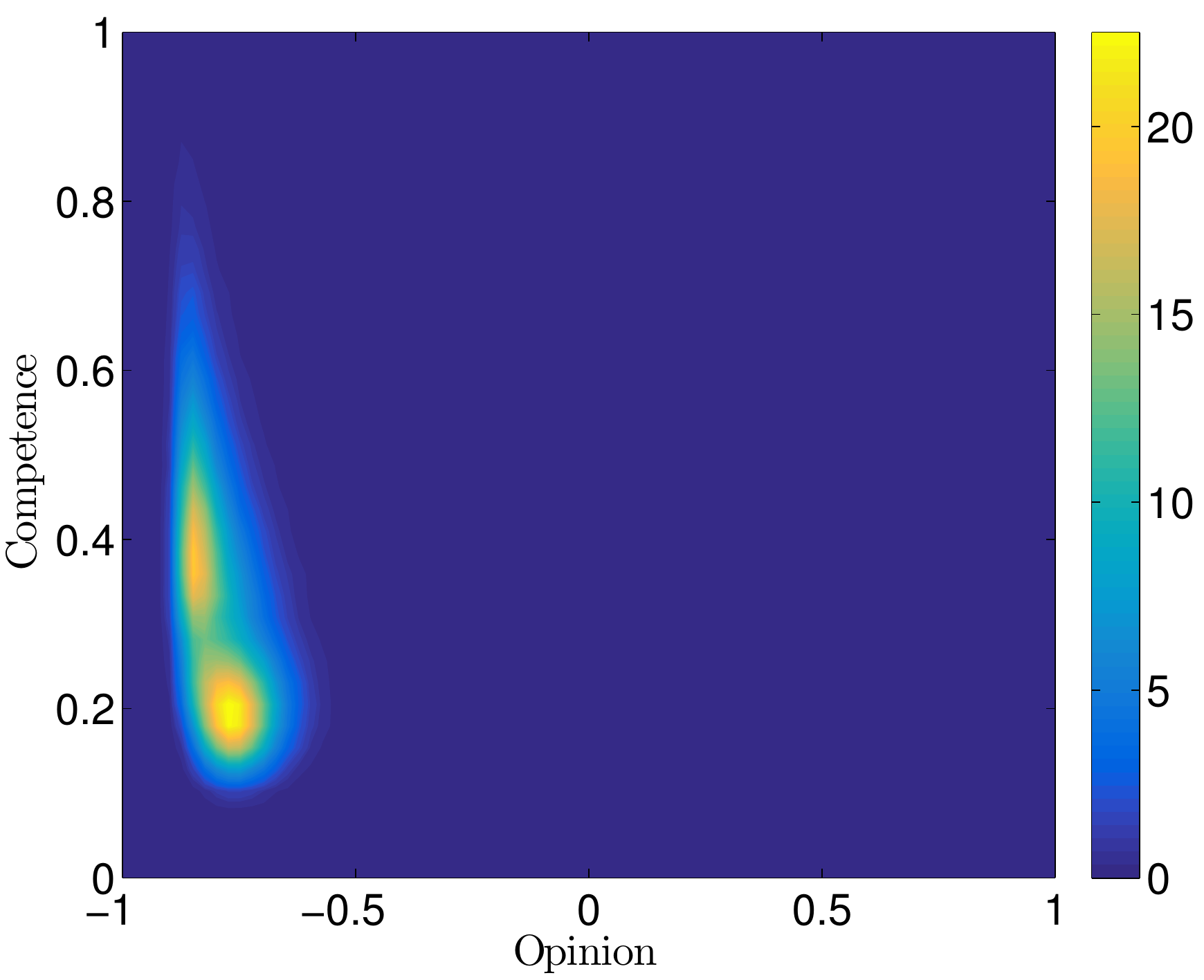}}\\
\caption{\textbf{Test 2}. Kinetic solution for the opinion dynamics with forcing term with rate \eqref{eq:S_drift} at different time steps. The evolution of the competence variable is given by $\lambda_B=10^{-2}$, $\lambda_C(x)=\lambda_C\chi(x\ge x_d)$ where $\lambda_C=10^{-2}$ and $x_d=0.3$, $\lambda=\lambda_B+\lambda_C$, and $\sigma_{\kappa}^2=10^{-2}$. We considered $z\sim U([0,1])$. We present the behavior of the reference models aCF-cMC and EB1-EB2 under the action of the equality bias functions $\Phi_1(x),\Phi_2(x)$ for three time steps. }
\label{fig:test2}
\end{figure}

\begin{figure}
\centering
\includegraphics[scale=0.4]{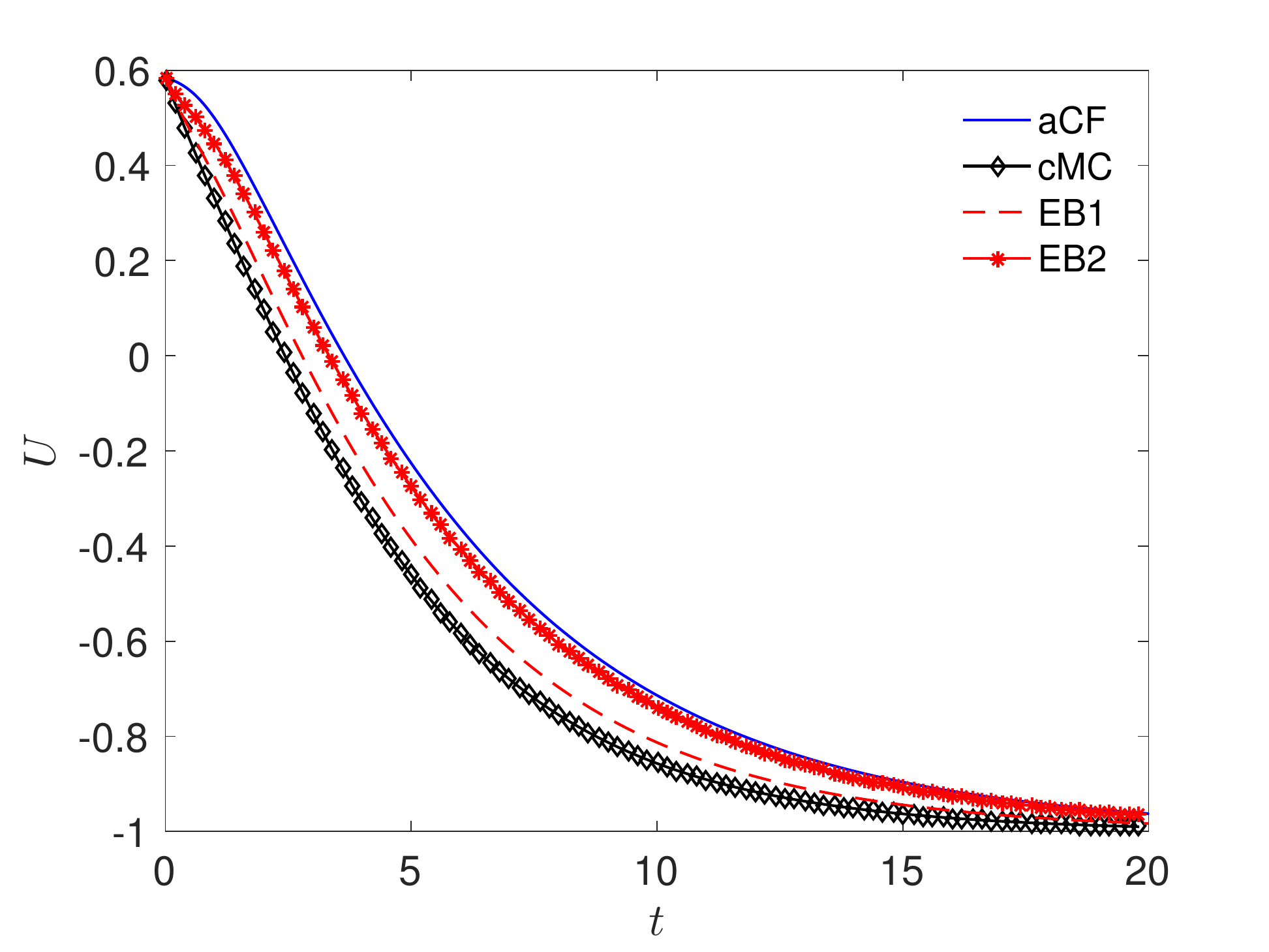}
\caption{\textbf{Test 2}. Convergence of the mean deicision in the reference models aCF-cMC and EB1-EB2 under the action of the forcing term with rate \eqref{eq:S_drift} and $x_d=0.3$.  The equality bias function $\Phi_1(x)$, which has been sketched in Figure \ref{fig:initial}, slows down the convergence speed of the collective decision toward $w_d=-1$. }
\label{fig:mean_selfprop}
\end{figure}

\subsection{Test 3: bounded confidence case}
In the last test case we consider the case of bounded confidence based interactions \cite{AHP,APZc,BL}. Often we notice how more well-educated, more capable and more competent people are also those best disposed to dialogue. Then competence is generally associated to the predisposition to listen other people. The higher this quality, greater is the ability to value other opinions. Vice versa, a person unwilling to listen and dialogue is usually marked by a lower level of the described trait. Therefore, it is natural to consider a bounded confidence model where the threshold on the exchanges of opinions depends on the degree of competence.

In particular, we consider a compromise function $P(x,x_*;w,w_*)$ with the following form
\be\label{eq:BC_interaction}
P(x,x_*;w,w_*)=\chi(|w-w_*|\le \gamma \Delta(x,x_*)),
\ee
where $\Delta(x,x_*)$ is a competence-dependent function that ranges in the closed interval $[0,1]$ taking into account the maximum distance under which the interactions are allowed, and $\gamma>0$ is a constant value.  
%In this test we consider
%\be
%\Delta(x,x_*) = \Delta(x-x_*).
%\ee
A possible choice of for the function $\Delta(\cdot,\cdot)$ is
\be\label{eq:Delta_BC}
\Delta(x,x_*) = R_{cMC}(x,x_*),
\ee
where the function $R(\cdot)$ has been defined in \eqref{eq:R}. We observe how for the choice $\gamma=2$ and $\Delta(x,x_*)=R_{MC}$ the bounded confidence interaction function reproduces the behavior of a maximum competence model.

% In the introduced model the agents with the highest competence are not influenced by the less competent but only by people with a comparable level of competence. On the other hand, the less skilled agents are always influenced by them.
 In the biased case the bounded confidence model becomes
\be\label{eq:BC_biased}
P(x,x_*;w,w_*)=\chi(|w-w_*|\le \gamma\Delta(\Phi(x),\Phi(x_*))).
\ee
 We perform the numerical test in the case of absence of driving force, i.e. $S(\cdot)=0$. In Figure \ref{fig:BC_initial_marginal} we report the initial configuration of the multi-agent system for the test and the asymptotic density function of the opinion variable taking into account the bounded confidence interaction function \eqref{eq:BC_interaction} and its biased version \eqref{eq:BC_biased}. We chose the parameters $\lambda_C=\lambda_B=10^{-3}$, $\lambda=\lambda_B+\lambda_C$ and $\sigma_{\kappa}=10^{-4}$ for the evolution of the competence variable, the function $\Delta(x-x_*)$ is \eqref{eq:Delta_BC} with $c=10^2$ and $\gamma=1/2$.

In Figure \ref{fig:BC_collective} it is possible to observe how the evolution of opinion and competence deeply changes under the effect of an equality bias function. In particular we considered the equality bias function $\Phi_2(x),x\in X$ with the choice of parameter introduced in Figure \ref{fig:initial}. In Figure \ref{fig:BC_initial_marginal} (right plot) we report the asymptotic marginal density for the opinion variable. The system evolves towards two clusters, characterizing two subpopulations with different decisions driven by the most competent agents. Finally, it is possible to observe how the equality bias drives the system toward a suboptimal collective decision for both populations where the influence of less competent agents become more relevant.

\begin{figure}
\centering
\includegraphics[scale=0.43]{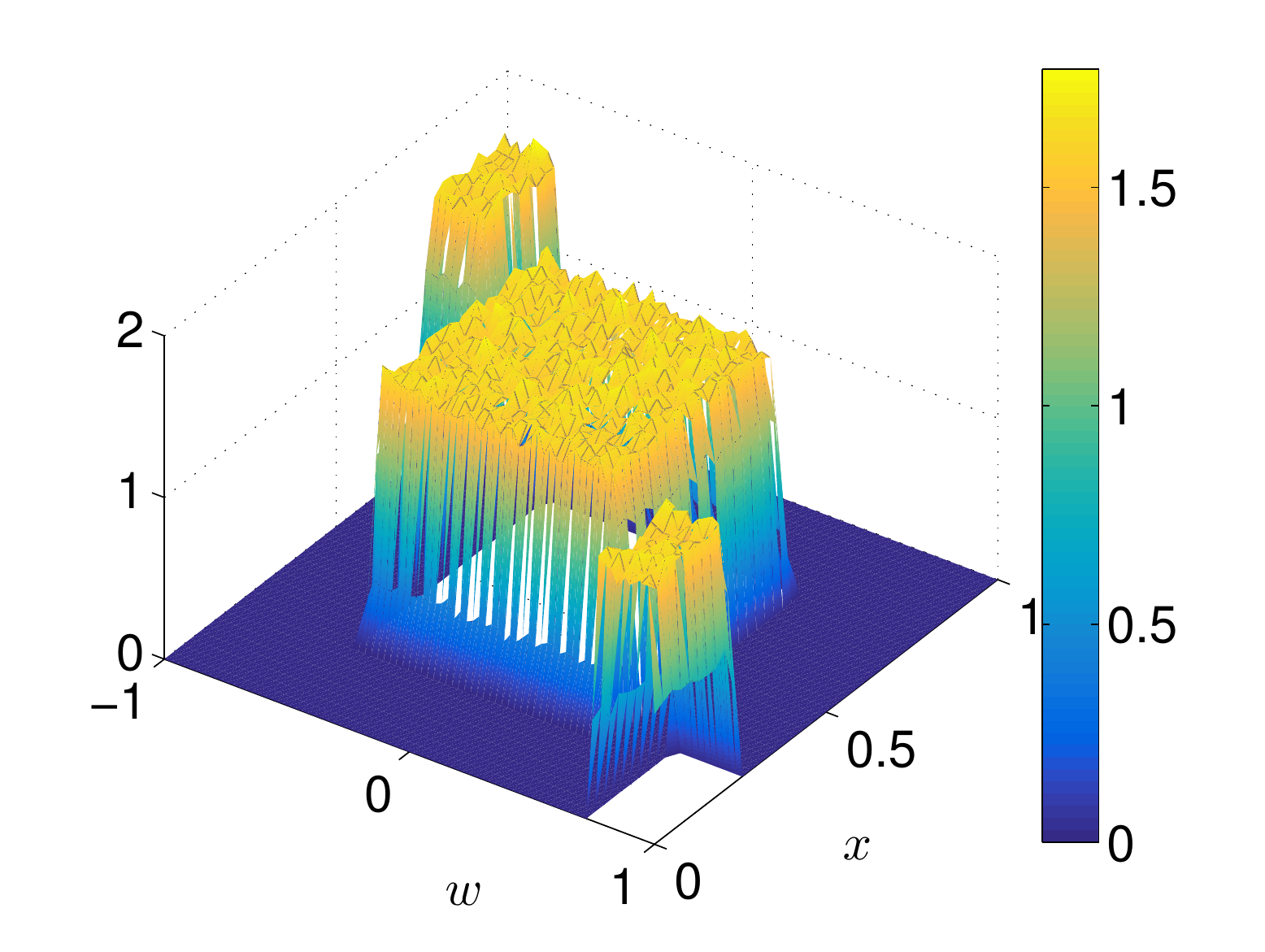}
\includegraphics[scale=0.34]{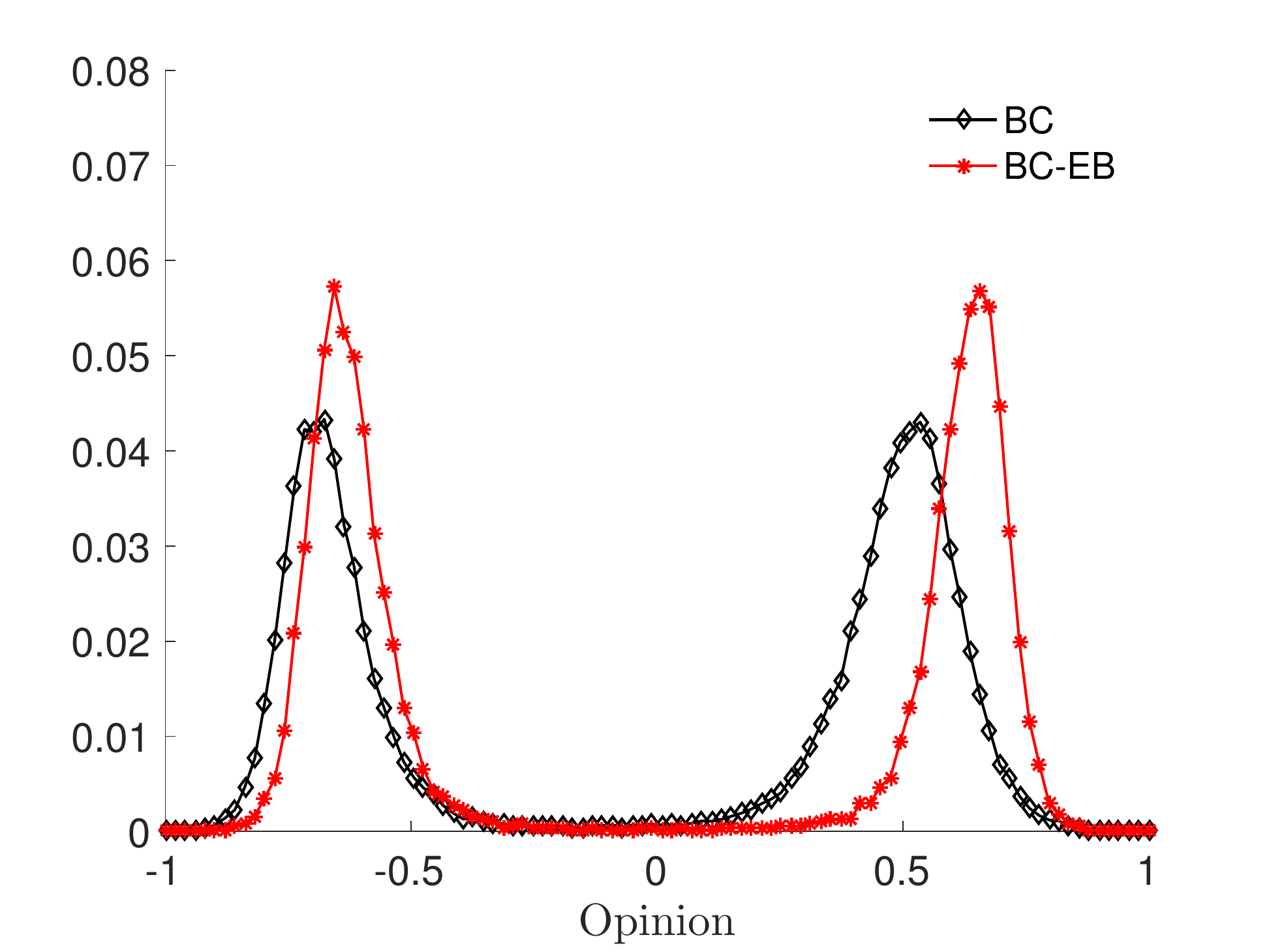}
\caption{\textbf{Test 3}. Left: initial configuration of the multi-agent system, we considered the case with three population. Right: asymptotic marginal density for the opinion variable, we observe how the bounded confidence interaction function introduced in \eqref{eq:BC_interaction}, $\gamma = 1/2$, splits the system in two populations for which the equality bias emerges separately as an action of the equality bias function $\Phi_2(x)$. }
\label{fig:BC_initial_marginal}
\end{figure}

\begin{figure}
\centering
\subfigure[(BC) t=10]{
\includegraphics[scale=0.26]{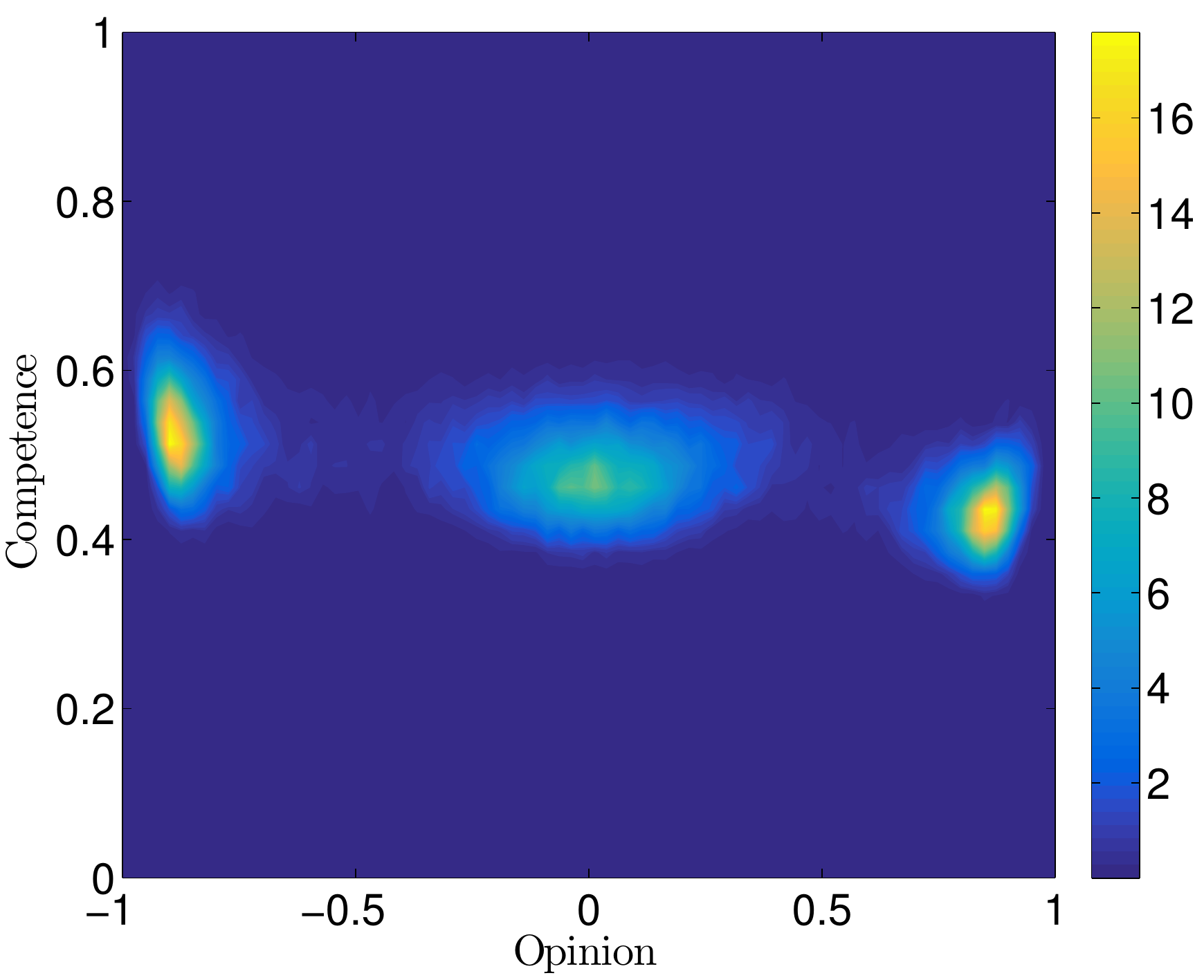}}
\subfigure[(BC) t=30]{
\includegraphics[scale=0.26]{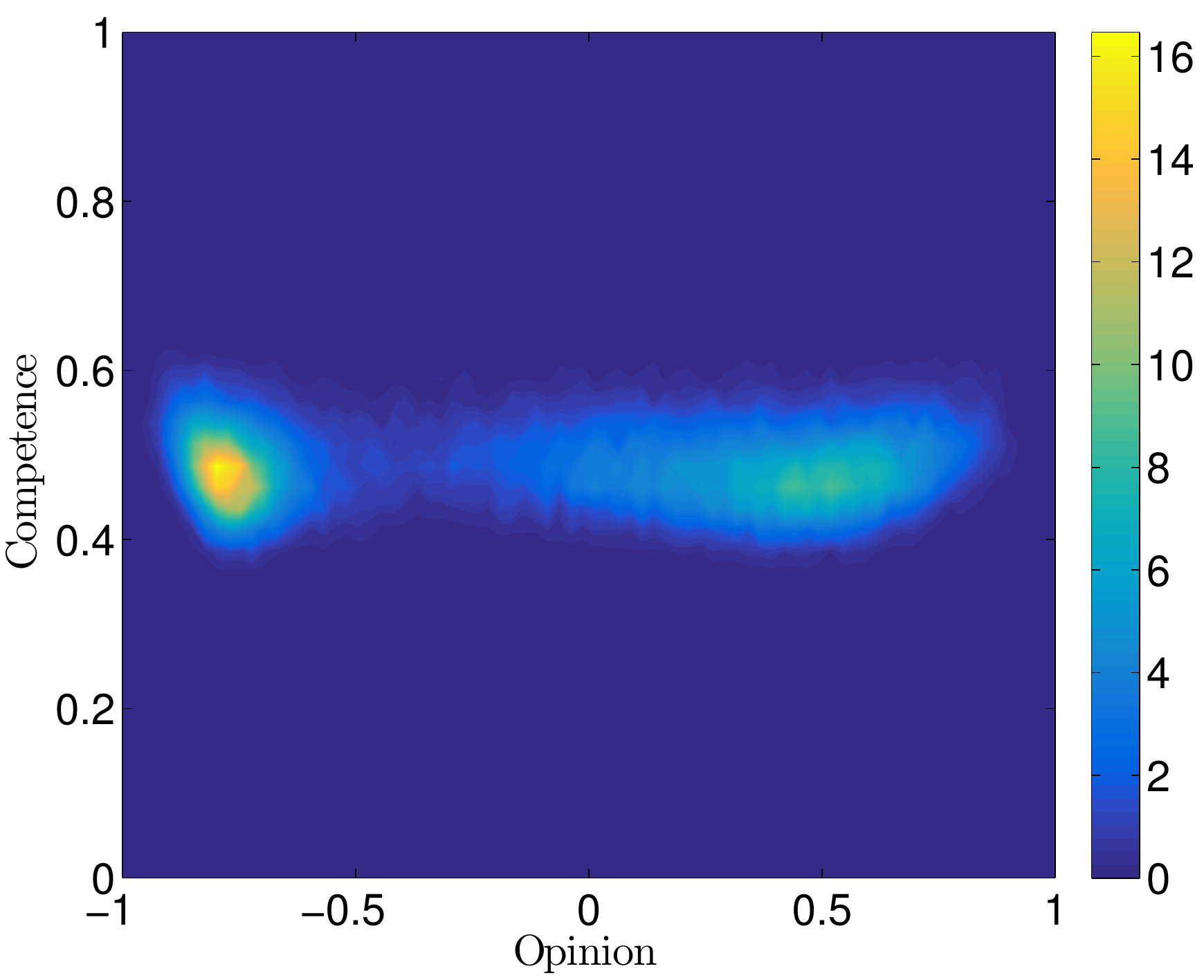}}
\subfigure[(BC) t=100]{
\includegraphics[scale=0.26]{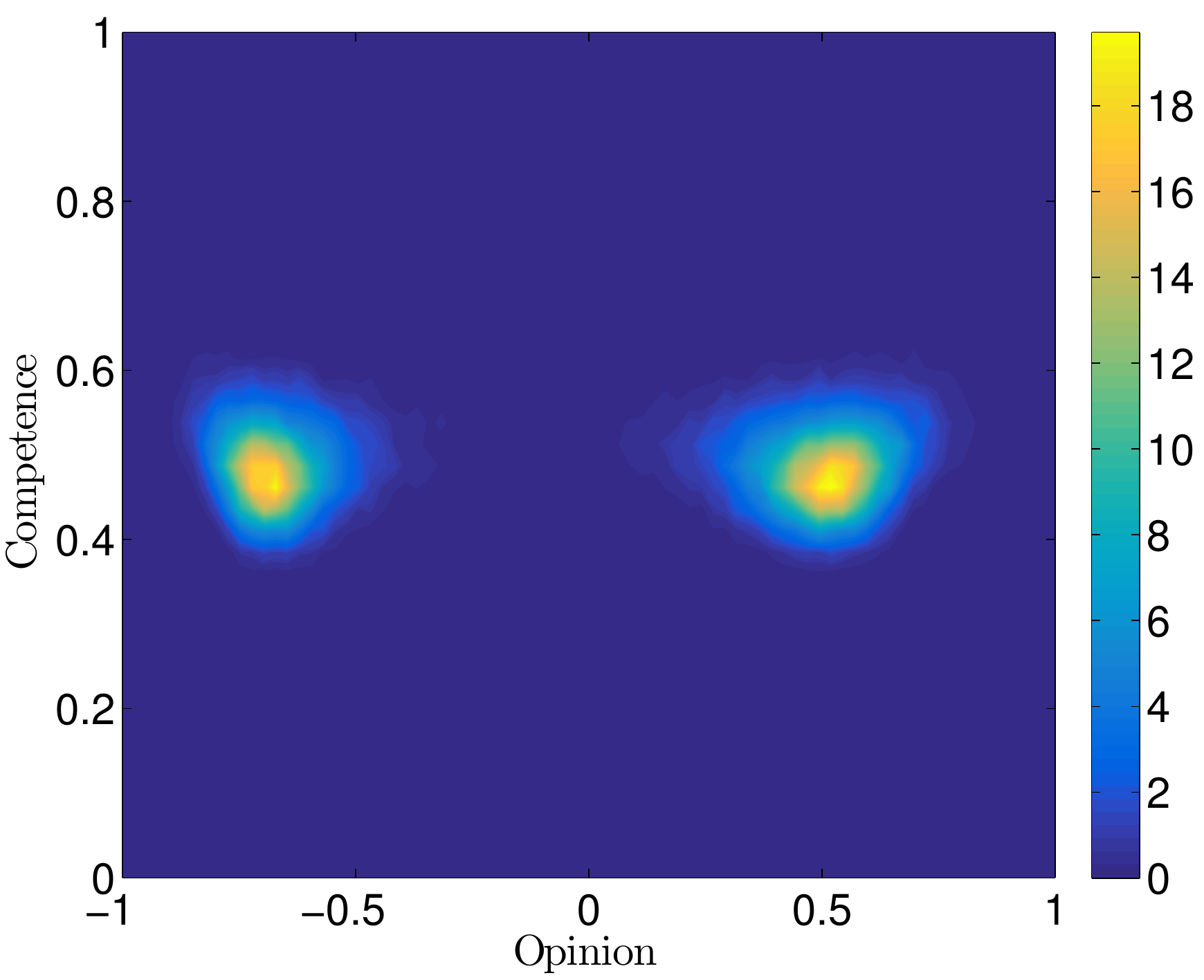}}\\
\subfigure[(BC-EB) t=10]{
\includegraphics[scale=0.26]{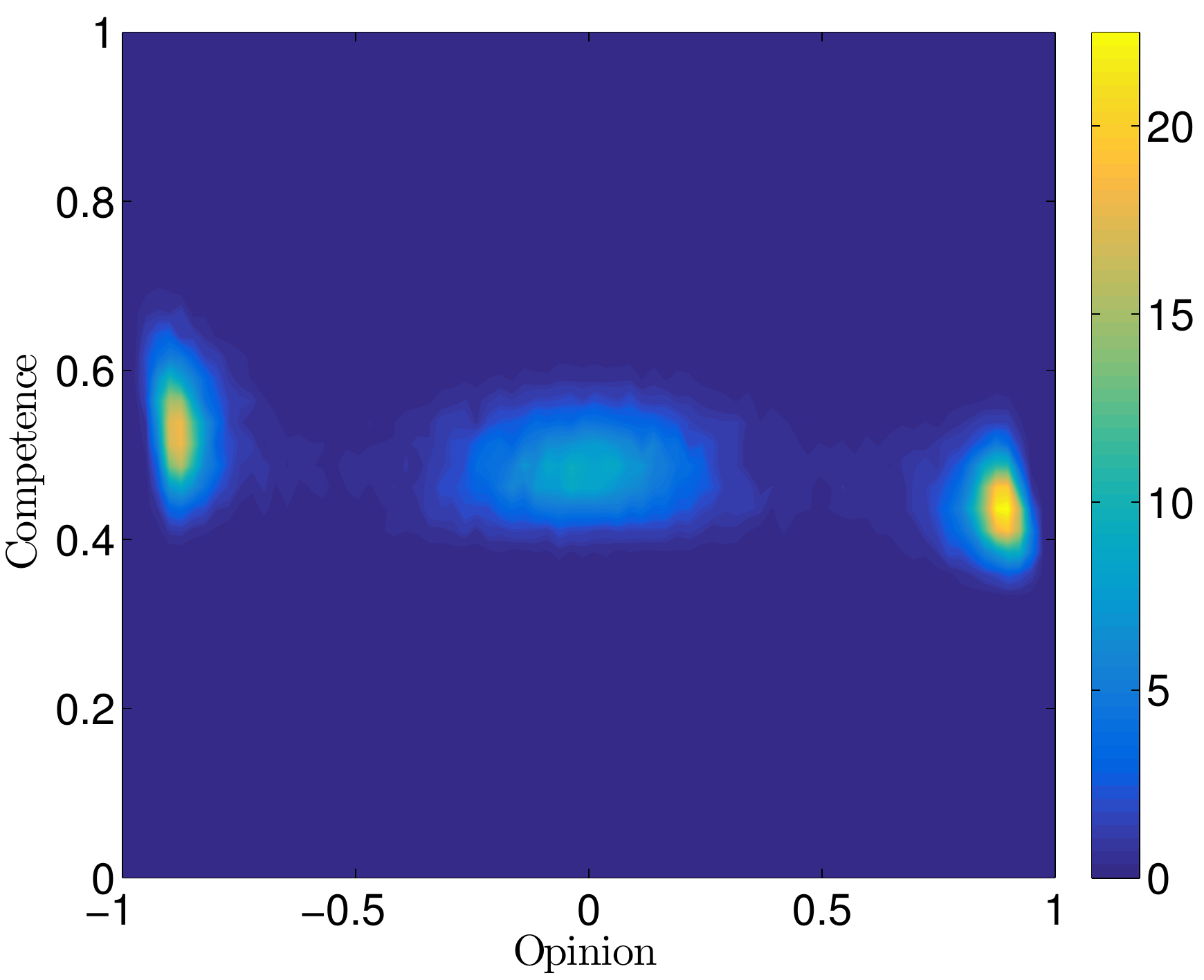}}
\subfigure[(BC-EB) t=30]{
\includegraphics[scale=0.26]{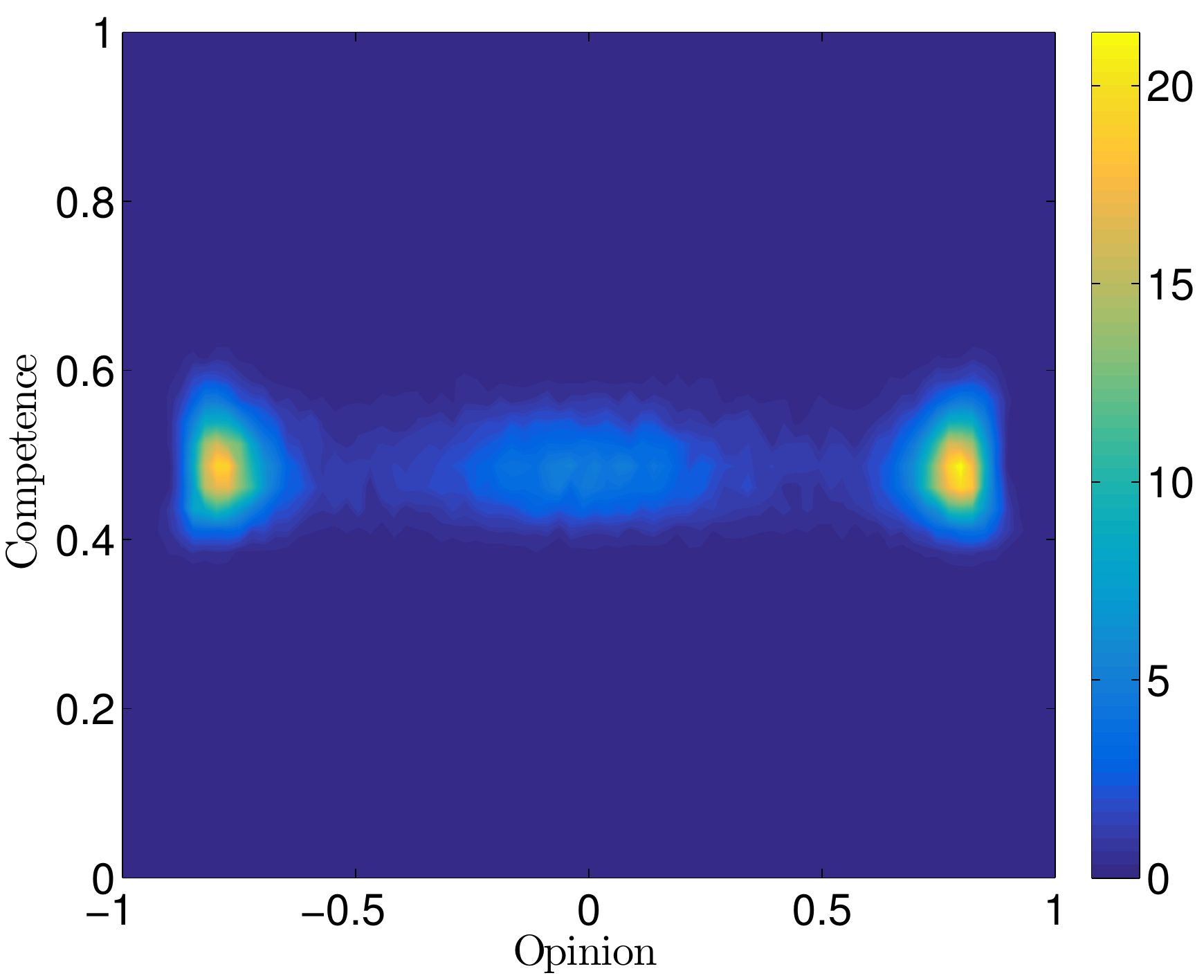}}
\subfigure[(BC-EB) t=100]{
\includegraphics[scale=0.26]{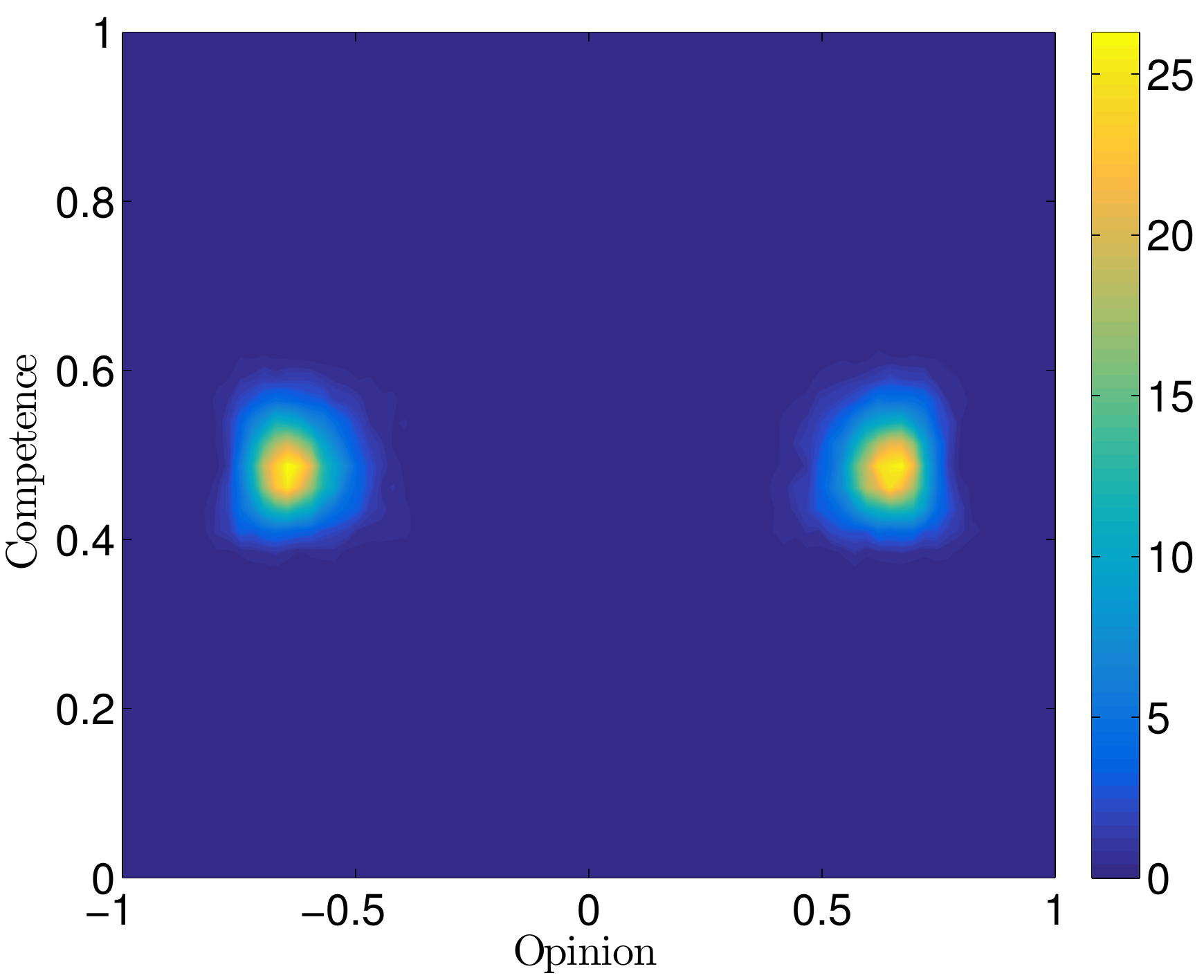}}
\caption{\textbf{Test 3}. Evolution of the multi-agent-system in the bounded confidence case (BC), in the first row, and under the action of the equality bias function $\Phi_2(x)$ (BC-EB), in the second row.}
\label{fig:BC_collective}
\end{figure}

\section{Conclusion}
We introduced and discuss kinetic models of multi-agent systems describing the process of decision making. The models are obtained in the limit of a large number of agents and weight the opinion of each agent through its competence. The binary interaction dynamics involves both the agents' opinion and competence, so that less competent agents can learn during interactions from the more competent ones. 
This lead to an optimal decision process where the results is a direct consequence of the agents' competence. The introduction of an equality bias in the model is obtained by considering a suitable function with a shape analogous to the one experimentally found in \cite{KD}. Numerical results show that the presence of an equality bias leads the group to suboptimal decisions and in some cases to the emergence of the opinion of the less competent agents in the group.

\end{document}